\documentclass[12pt]{article}
\usepackage{amsmath}
\usepackage{amssymb}
\usepackage{dsfont}
\usepackage{amsthm}
\usepackage{graphicx}
\usepackage{enumitem}
\usepackage{booktabs}
\theoremstyle{definition}
\usepackage{geometry}
\newtheorem{theorem}{Theorem}
\newtheorem{lemma}{Lemma}
\newtheorem{proposition}{Proposition}
\newtheorem{assumption}{Assumption}
\newtheorem{definition}{Definition}
\newtheorem{example}{Example}
\newtheorem{remark}{Remark}
\newtheorem{algorithm}{Algorithm}

\usepackage{makecell}

\usepackage{hyperref}
\usepackage{natbib}
\setcitestyle{authoryear,semicolon,aysep{,}}

\title{\textbf{Fixed effects as generated regressors}}
\author{Jiaqi Huang\\
Department of Economics\thanks{I am grateful to Patrik Guggenberger, Marc Henry, Sung Jae Jun and Andres Aradillas-Lopez for their generous comments on the paper. I would also like to thank my classmates in ECON 555 for constructive discussions, especially Taegyu Yang. All errors are mine.}, Pennsylvania State University}

\begin{document}
\maketitle

\begin{abstract}
Many economic models feature moment conditions that involve latent variables. When the latent variables are individual fixed effects in an auxiliary panel data regression, we construct orthogonal moments that eliminate first-order bias induced by estimating the fixed effects. Machine Learning methods and Empirical Bayes methods can be used to improve the estimate of the nuisance parameters in the orthogonal moments. We establish a central limit theorem based on the orthogonal moments without relying on exogeneity assumptions between panel data residuals and the cross-sectional moment functions. In a simulation study where the exogeneity assumption is violated, the estimator based on orthogonal moments has smaller bias compared with other estimators relying on that assumption.  An empirical application on experimental site selection demonstrates how the method can be used for nonlinear moment conditions. 


\end{abstract}

\section{Introduction}

In this paper, we will consider a model in which the researcher is interested in a parameter $\mu_0$ that can be characterized by a set of moment conditions:
\begin{equation}
\mathbb{E}[m(W_i,\alpha_i,\mu_0)]=0,
\label{eq:mom mu}
\end{equation}
where $W_i$ will be observables that are specific to an individual $i$, while $\alpha_i$ is unobservable to the researcher.

If the researcher has no information about $\alpha_i$, one can derive observable implications from (\ref{eq:mom mu}) which typically results in partial identification of $\mu_0$ (see \cite{entropic,ot,li2021}).  However, oftentimes the researcher has a panel data model that specifies $\alpha_i$ as fixed effects. Formally, for each individual $i$, the researcher observes $(Y_{it},X_{it}')_{t=1}^T$ that are related to $\alpha_i$ through the linear panel data model:
\begin{equation}
Y_{it}=X_{it}'\beta_0+\alpha_i+u_{it}, 
\end{equation}
with $u_{it}$ satisfying sequential exogeneity to accommodate dynamic panel data models.\footnote{Models with additional time fixed effects can be reduced to the setup by first taking the within transformation for each time period (subtracting the mean of each variable within each time period. See Section \ref{application} for an example.} We do not impose any independence assumption between $(Y_{it},X_{it}')_{t=1}^T$ and $W_i$. 

Many interesting economic problems share the above structure. Typical examples include value-added for teachers \citep{chetty2014a} and the match quality between firm-product pairs \citep{kim2024}. More concretely, \cite{chetty2014} and \cite{jackson} study whether students who are matched with teachers with higher value-added in primary school might also have better long-run outcomes in college attendance and earnings. The latent variable $\alpha_i$ represents the value-added on students' test scores for teacher $i$, and $\mu_0$ is the slope coefficient of a linear regression of students' average (for each teacher $i$) long-run outcomes on $\alpha_i$'s.  
The linear panel data model would be specified using $Y_{it}$ as the test score for student $t$ who has teacher $i$, and $X_{it}$'s are covariates that include lagged test scores and socio-economic status of the student. In the context of firm dynamics, \cite{kim2024} show that the firm-product match quality is an important predictor for firm's product adding and dropping behaviors. In their specification, $\alpha_i$ is the firm-product match quality where $i$ represents a firm-product pair and $t$ represents time. The panel data model would use product-firm  specific value of shipment as $Y_{it}$, product tenure and  firm size as the covariates $X_{it}$. The parameter of interest $\mu_0$ is the average marginal effect of a shift in product-firm match quality on the probability that the firm adds or drops that product.

Using the panel data model, one can easily recover an estimate $\hat{\alpha}_i$ of $\alpha_i$, but directly plugging the estimated $\hat{\alpha}_i$ into the moment condition (\ref{eq:mom mu}) and ignoring the potential measurement error in $\hat{\alpha}_i$ might render the inference results invalid. This paper aims to provide an inference procedure for $\mu_0$ that takes into account the error in estimating $\alpha_i$'s properly. Before describing the approach, we first consider a special case of (\ref{eq:mom mu}) which has drawn a lot of attention in the literature \citep{CGK2025,xie2025,battaglia2024,deeb2021}.

\begin{example}\label{ex:ols}

Consider a simple example, where $\alpha_i$ is defined via $Y_{it}=X_{it}'\beta_0+\alpha_i+u_{it}$. We are interested in the regression coefficient $\mu_{02}$ for the regression
\begin{equation}
W_i = \mu_{01}+\mu_{02} \alpha_i+v_i,
\label{eq:ex}
\end{equation}
where $\mathbb{E}(v_i|\alpha_i)=0$. 

The two moment conditions for identifying $\mu_{01}$ and $\mu_{02}$ are:
\begin{equation}
\mathbb{E}\begin{bmatrix} W_i -\mu_{01}-\mu_{02} \alpha_i \\ \alpha_i(W_i-\mu_{01}-\mu_{02} \alpha_i)\end{bmatrix}= \begin{bmatrix} 0 \\ 0 \end{bmatrix}.
\label{eq:ex mom}
\end{equation}

\end{example}
\bigskip

The teacher value-added  example in \cite{chetty2014} follows the above structure where $W_i$ is the average long-run outcome of the students who have teacher $i$ while $\alpha_i$ is teacher $i$'s value-added.

The need to estimate $\alpha_i$'s and to use the estimates in the moment conditions (\ref{eq:mom mu}) suggests having small average mean squared error (MSE) $\frac{1}{N} \sum_{i=1}^N \mathbb{E}(\hat{\alpha}_i-\alpha_i)^2$ could be important for estimating $\mu_0$. 
Recognizing this fact, Empirical Bayes (EB) corrections of $\hat{\alpha}_i$'s are suggested to alleviate the potential impact of the measurement error \citep{angrist2017,CGK2025}. Results using EB corrected versions of $\alpha_i$'s typically abstract away from the panel data structure that defines $\alpha_i$ and start by assuming that  an estimator $\hat{\alpha}_i$ of $\alpha_i$ is available which has a normal distribution with mean $\alpha_i$ and some variance $\sigma_i^2$ \citep{xie2025,CGK2025}. The EB corrections are then applied to $\hat{\alpha}_i$'s. \cite{xie2025} establishes a central limit theorem where estimates of $\mu_0$ are obtained by replacing $\alpha_i$'s in the moments by the EB corrected estimates in Example \ref{ex:ols}. However, the above result requires i) independence between $u_{it}$ and $v_i$; ii) $\alpha_i$ being independent of $\mbox{Var}_i(\sum_{t=1}^T u_{it})$. \cite{CGK2025} recognizes that the latter independence assumption is restrictive. Instead, they advocate combining the classical correction of error-in-variables regression in \cite{deaton1985} with bootstrapped confidence interval, which they showed  is a valid inference procedure that relaxes the independence between $\alpha_i$ and $\mbox{Var}_i(\sum_{t=1}^T u_{it})$. However, they still need to assume orthogonality condition between $u_{it}$ and $v_i$, and their approach is hard to extend to nonlinear moments.

The required independence assumption $u_{it}$ and $v_i$ may not hold in some applications of interest. For example, in the teacher value-added example, if $W_i$ represents average college attendance rate of the students, it can be correlated with $u_{it}$ for the same set of students. In the firm-product match quality example, an idiosyncratic shock in demand could affect the product value of shipment and also induce the firm to drop the product.

To address this issue, we instead exploit the panel data structure that defines $\alpha_i$'s. Although it can be more restrictive in some settings where estimators of $\alpha_i$ does not come from a panel data model, it enables us to relax the exogeneity assumption between $u_{it}$ and $v_i$.  We borrow insights  from \cite{weidner2024} and consider asymptotics in which both the dimension of individuals $N$ and time periods of panel data $T$ go to infinity. This allows us to use debiased machine learning (DML) techniques in \cite{cher2022} to construct orthogonal moment conditions that eliminate any first order effect of the measurement error in $\hat{\alpha}_i$ as long as $T$ goes to infinity sufficiently fast compared with $N$. More specifically, we require $\lim_{N\rightarrow \infty} \frac{N^{\frac{1}{2}}}{T}=0$ and $\lim_{N\rightarrow \infty} \frac{N}{T}=\infty$. As long as the above limits hold, the asymptotic distribution of the estimator $\hat{\mu}$ using the orthogonalized moments will not depend on the specific rate at which $T$ goes to infinity relative to $N$. Constructing the orthogonal moment conditions requires the sequential exogeneity assumption to derive the adjustment term to the moment function as shown in \cite{ichimura}. This shows that taking the panel data structure into account can offer potential improvements on the estimator $\hat{\mu}$. Furthermore, shrinkage methods like EB and Stein's Unbiased Risk Estimation (SURE) can be readily applied to estimates of $\alpha_i$, because with orthogonal moments how $\alpha_i$'s are estimated does not affect the asymptotic distribution.\footnote{SURE has been studied in \cite{xie2012}, and for applications in econometrics, see \cite{kwon}.} We further derive a central limit theorem that characterizes the limit distribution of $\hat{\mu}$ based on orthogonal moments. The result is the first that unites the cross-sectional moments and linear panel data model without imposing parametric assumptions on the panel data part of the model. \footnote{For the case of the panel data model being parametric, see \cite{weidner2024}. }

One stark difference of our approach from the standard DML literature is in the procedure of cross-fitting. DML requires the nuisance parameter $\alpha_i$ to be learned using samples from other folds and then use it to construct the sample moments $m(W_i,\hat{\alpha}_i,\mu)$ for the current fold. However, this is not possible in our context, because only data about individual $i$ contains information about $\alpha_i$. Hence any reasonable estimator of $\alpha_i$ must depend on the data for individual $i$. We adapt the cross-fitting procedure by first estimating $\beta_0$ in the panel data model using samples from other folds (denote the estimator as $\tilde{\beta}_l$), and then an estimator $\tilde{\alpha}_{il}$ of $\alpha_i$ is constructed as the difference between averages of $Y_{it}$ over time and averages of $X_{it}'\tilde{\beta}_l$ over time. The difference $\tilde{\alpha}_{il}-\alpha_i$ depends on $i$-th sample only through the panel data residuals $u_{it}$. We showed that with suitable controls on the correlation between $u_{it}$ and $\frac{\partial m(W_i,\alpha_i,\mu_0)}{\partial \alpha_i}$, using $i$-th sample to estimate both $\alpha_i$ and $\mu_0$ will not affect the limiting distribution of the orthogonal moments (see condition (ii) of Assumption \ref{as:alpha} for a precise statement).

In a simulation study with linear dynamic panel and linear regression model for (\ref{eq:mom mu}), we showed that when $u_{it}$ and $W_i$ are correlated, the approaches used in \cite{CGK2025} and \cite{xie2025} that ignore the potential endogeneity issue have non-negligible bias and consequently larger root mean squared error (RMSE). In addition, standard test statistics based on their estimators do not properly control the null rejection probabilities. In contrast, estimators that employ orthogonal moment conditions have null rejection probabilities closer to the nominal size, and employing an EB (or SURE) corrections on the $\hat{\alpha}_i$'s will perform slightly better in terms of null rejection probabilities but may incur larger RMSE. However, using orthogonal moment conditions results in larger variance in the estimator in the simulation, and can lose power compared with \cite{CGK2025} if endogeneity concerns are not present.\footnote{The estimator used in \cite{CGK2025} achieves the semi-parametric efficiency bound under their assumptions.}

To demonstrate the power of the method in dealing with non-linear moment conditions, we further consider estimating coefficients of a logit model as an empirical illustration. The moment conditions in this context set the expectation of the score function to zero. Specifically, we focus on the experimental site selection of the agricultural catastrophe insurance (ACI) program in China. \cite{policy} have documented consistent patterns of positive site selection in the universe of policy experiments conducted by Chinese central government, but they did not commit to a particular explanation of the positive selections. As implementation details are relegated to local government and good practices discovered during the experimentation phase are encouraged to be rolled out across the whole country, the central government has incentive to pick experimental sites that are more specialized in the specific policy domain and could serve as ``model sites''. We test the hypothesis in the context of ACI, where the key variable of interest ($\alpha_i$) is the level of specialization in grain production for each county. As the specialization level is unobserved to the researcher, it is captured as fixed effects in a linear panel data model with $Y_{it}$ being the grain production per rural employment.  We then focus on the parameter associated with $\alpha_i$ that enters a linear logit model where the response is whether the county is selected as the experimental site. Results show that the coefficient associated with the level of specialization is positive and significant, but using orthogonal moments reduce the magnitude of the coefficient from around 0.56 (using simple plug-in methods) to around 0.4. EB and SURE corrections have more limited impacts on the estimates. Therefore, the evidence favors the theory that governments are choosing sites that are more experienced in agricultural production, and calls for usage of orthogonal moments when endogeneity between $u_{it}$ and $W_i$ is of concern.

\bigskip

\textbf{Literature}
\bigskip

This paper is related to literature on correcting measurement error for estimated latent variables. In addition to the classical approach like error-in-variables regression as in \cite{deaton1985} and the generated regressor problem in \cite{pagan}, more recent contributions include \cite{CGK2025}, \cite{xie2025}, \cite{battaglia2024} and \cite{deeb2021}. \cite{battaglia2024} uses an asymptotic regime so that the measurement error is similar in magnitude to the sampling error of the model (\ref{eq:mom mu}) and derives formulas for bias correction of the estimator. They focused on the case where $\alpha_i$'s are estimated using Machine Learning Methods. \cite{deeb2021} focuses on the teacher value-added context and derives moment conditions that can help estimate $\mu_0$ in Example \ref{ex:ols}. The results of these papers rely on the linear regression structure that defines $\mu_0$, and generalizations to more general models are not straightforward.

The long panel data asymptotics with both $N$ and $T$ going to infinity has been studied in \cite{hahn2002}. See also \cite{hahn2004} for the non-linear panel data context. The main objective is to derive bias correction methods for $\beta_0$ to combat the incidental parameter problem of estimating infinitely many $\alpha_i$'s \citep{neyman}. Other approaches to circumvent the incidental parameter problem include reparametrization and integrating out the $\alpha_i$'s in a Bayesian framework \citep{lancaster}.

This paper also uses techniques from the DML literature as in \cite{dml,cher2022}. The paper that is most closely related is \cite{weidner2024}, which also features a panel data first stage to recover latent variables that is going to be used in a model like (\ref{eq:mom mu}). However, their focus is on non-linear panel data models that are fully parametric and enables the researcher to exploit the likelihood function and construct higher-order orthogonal moments. This paper extends their approach to a semi-parametric context without making any distributional assumptions on $u_{it}$. The construction of the orthogonal moment functions and estimation procedure follows from \cite{cher2022} and \cite{ichimura}. 

The paper is also related to the literature on EB corrections for a parallel decision problem. Examples that use Empirical Bayes to correct for potential measurement error include \cite{angrist2017}. Standard EB corrections may not work well when the error distribution is far from normal. Recent literature suggests choosing the shrinkage parameter based on the idea of Stein's Unbiased Risk Estimation (SURE) \citep{kwon,xie2012,brown}. In our framework, any corrections using EB or SURE on the estimates of $\alpha_i$ can be incorporated, will not affect the asymptotic distribution of $\hat{\mu}$, and could be beneficial to the performance of the estimator in finite sample.

\bigskip

The plan of the paper is as follows. Section 2 formally introduces the model and the concept of Neyman orthogonality. Section 3 describes the orthogonal moment conditions and the cross-fitting procedure, while Section 4 gives the assumptions that establishes a central limit theorem for the estimator proposed in Section 3. Section 5 conducts a simulation study of the estimator based on orthogonal moments compared with other estimators in the literature. The empirical illustration of experimental site selection is detailed in section 6. All proofs are relegated to the appendices.

\section{Individual fixed effects in cross-sectional moments}

We will consider a long-panel data setting in which the cross-sectional dimension is indexed by $i$ and the time dimension is indexed by $t$. Suppose we have a balanced panel with $N$ number of individuals observed across $T$ time periods. For each individual $i$ at time $t$, the researcher observes a vector of observables $Z_{it}=(Y_{it},X_{it}')$ and for each individual $i$, the researcher observes a vector $W_i$ that only varies at the cross-sectional level.  Assume that in the first-stage, the following fixed effect model is estimated to recover the set of fixed effects $\alpha_i$ for each individual $i$.\footnote{We will not restrict to a particular estimator for the first-stage model, but provide sufficient conditions on the estimates $\beta$ and $\alpha_i$ that can eliminate the estimation bias in the second stage.}  
\begin{equation}
Y_{it}=X_{it}'\beta_0+\alpha_i+u_{it},
\label{eq:1st stage}
\end{equation}
where
$Y_{it}\in \mathbb{R}$ is the scalar outcome variable, and $X_{it}\in \mathbb{R}^p$ is a vector of observables.

\begin{assumption}\label{ass:cross}
The parameter $\mu_0$ solves
    \begin{equation}
		\mathbb{E}[m(W_{i},\alpha_i,\mu_0)]=0,
		\label{eq:2nd cross}
		\end{equation}
		for some moment function $m(W_i,\alpha_i,\mu_0)$ that is individual specific and $\alpha_i$ is unobserved but defined via (\ref{eq:1st stage}).
		
		The idiosyncratic shock $u_{it}$ is assumed to be sequentially exogenous and have finite variance:
		\begin{equation}
		E[u_{it}|X_{it},\alpha_i,\Phi_{it}]=0\quad \forall t=1,\cdots,T,\quad E[u_{it}^2]=\sigma^2. 
		\label{eq:seq exogeneity}
		\end{equation}
		where $\Phi_{it} = \{(Y_{it'},X_{it'})_{t'=1}^{t-1}\}$ is the filtration up to time $t-1$.
\end{assumption}
\bigskip

One approach in dealing with the moment conditions is to use the plug-in principle by replacing $\alpha_i$'s
 with their corresponding estimates $\hat{\alpha}_i$,  which gives us the sample analogue estimating equation for $\hat{\mu}$:

\[\sum_{i=1}^N m(W_i,\hat{\alpha}_i,\hat{\mu})=0\]
under Assumption \ref{ass:cross}.

Assuming the moment conditions $m(\cdot)$ 
are sufficiently smooth, a standard Taylor expansion of the above moment condition around $(\hat{\alpha}_i,\mu_0)$ yields:
\[-\frac{1}{N} \sum_{i=1}^N  m(W_{i},\hat{\alpha}_i,\hat{\mu})=\mathbb{E}\left( \frac{\partial m(W_i,\hat{\alpha}_i,\mu^*)}{\partial \mu})+o_p(1)\right) (\hat{\mu}-\mu_0),
\]
for some intermediate value $\mu^*$. The RHS of the above equation is standard in M-estimation, and the main difficulty lies in establishing the asymptotic distribution of the LHS.

Following insights from \cite{weidner2024}, we can do a further Taylor Expansion for the LHS of the above display (for simplicity drop the minus sign):

\begin{align}
 \frac{1}{N} \sum_{i=1}^N m(W_i,\hat{\alpha}_i,\mu_0) &=\underbrace{\frac{1}{N} \sum_{i=1}^N  m(W_i,\alpha_i,\mu_0)}_{I} \nonumber \\
&+\underbrace{\Big( \frac{1}{N} \sum_{i=1}^N \frac{\partial m(W_i,\alpha_i,\mu_0)}{\partial \alpha_i }-\mathbb{E}\left(\frac{\partial m(W_i,\alpha_i,\mu_0)}{\partial \alpha_i }\right) \Big) (\hat{\alpha}_i-\alpha_i)}_{II} \nonumber \\
&+\underbrace{\frac{1}{N} \sum_{i=1}^N \mathbb{E}\left( \frac{\partial m(W_i,\alpha_i,\mu_0)}{\partial \alpha_i }\right) (\hat{\alpha}_i-\alpha_i)}_{III}+O_p(\sup_i |\hat{\alpha}_i-\alpha_i|^2),
\label{eq:exp1 cross}
\end{align}
under Assumption \ref{ass:cross} and also assuming the Hessian of $m(W_i,\alpha_i,\mu_0)$ with respect to $\alpha_i$'s is uniformly bounded almost surely.

The first term is the sum of moment conditions evaluated at the true parameter values, and a standard central limit theorem applies when we rescale the whole expression by $\sqrt{N}$.  
In standard DML setting, the term $II$ in (\ref{eq:exp1 cross}) will be $o_p(1)$ by using sample splitting and cross-fitting method. For example, by estimating $\hat{\alpha}_i$  on half of the sample and then estimating $\mu_0$ using the other half of the sample. However, in the current context, $\hat{\alpha}_i$ has to depend on the data for individual $i$ and hence could be correlated with $\frac{\partial m(W_i,\alpha_i,\mu_0)}{\partial \alpha_i }-\mathbb{E}(\frac{\partial m(W_i,\alpha_i,\mu_0)}{\partial \alpha_i })$. Therefore, standard cross-fitting techniques do not apply in the current context and they need to be adapted to make $II$ negligible asymptotically.   We will show how to do this in the next section. The term $III$ 
will be zero if
$\mathbb{E}(\frac{\partial m(W_i,\alpha_i,\mu_0)}{\partial \alpha_i })=0$. This condition is called Neyman Orthogonality and is widely used in the literature \citep{newey1994,dml,weidner2024}. Finally, the term $O_p(\sup_i|\hat{\alpha}_i-\alpha_i|^2)=O_p(\frac{1}{T})$, and when $T$ increases sufficiently fast with $N$, it will not affect the asymptotic distribution of $\hat{\mu}$. We now formally introduce the definition of Neyman Orthogonality: 
\begin{definition}\label{def:ney}
If for a moment function $m(W_i,\alpha_i,\mu_0)$ 
\begin{equation}
\mathbb{E}\left( \frac{\partial m(W_i,\alpha_i,\mu_0)}{\partial \alpha_i}\right)=0 ,
\label{eq:neyman orth}
\end{equation}
holds, then  $m(\cdot)$ satisfies the Neyman Orthogonality Condition.
\end{definition}
\bigskip

\section{Inference procedure}
\subsection{Constructing Orthogonal Moment Conditions}
As shown in (\ref{eq:exp1 cross}), Neyman orthogonality is useful in controlling $\frac{1}{\sqrt{N}} \sum_{i=1}^N m(W_i,\hat{\alpha}_i,\mu_0)$ when the true $\alpha_i$'s are replaced by their corresponding estimates. 
 Here we will address the question of how to modify moment conditions $\mathbb{E}(m(W_i,\alpha_i,\mu_0))=0$ to achieve Neyman Orthogonality, but still preserve identification of $\mu_0$.  

Before entering into the details, first note that (\ref{eq:seq exogeneity}) implies the following moment restrictions on $\alpha_i$:
\begin{equation}
E[b(X_i,\alpha_i) (Y_{iT}-X_{iT}'\beta_0-\alpha_i)]=0,
\label{eq:no m1}
\end{equation}
for any function $b(X_i,\alpha_i)$ with $X_i=(X_{i1}',\cdots,X_{iT}')'$.\footnote{One can also exploit sequential exogeneity by choosing a different period $t'$ than the final period $T$ and restricting the set of $X_i'$ to be the $X_{it}'$ up to period $t'-1$. This corresponds to the endogenous orthogonality conditions in \cite{ichimura} which requires additional regularity conditions. For simplicity, we do not pursue further in that direction.}

One way to construct orthogonal moment conditions is through modifying the original moment conditions by adding the influence function of the original moment condition with respect to $\alpha_i$ as an adjustment term like \cite{cher2022}, \cite{ichimura} and \cite{weidner2024}. More specifically, we can choose a function $a(X_i,\alpha_i,\mu)$ and form a modified moment condition $m^*(W_i,\alpha_i,Z_i,\beta,\mu,a)$ as:\footnote{Here $Z_i$ denotes $(Z_{i1}',\cdots,Z_{iT}')'$, with $(Y_{it},X_{it}')$ being a component of $Z_{it}$. }
\[m^*(W_i,\alpha_i,Z_i,\beta_0,\mu,a)=m(W_i,\alpha_i,\mu)+a(X_i,\alpha_i,\mu) (Y_{iT}-X_{iT}'\beta_0-\alpha_i).\]

By the sequential exogeneity assumption, we have for any $\mu$,
\begin{align*}
& \mathbb{E}[a(X_i,\alpha_i,\mu) (Y_{iT}-X_{iT}'\beta_0-\alpha_i)] = \mathbb{E}[ \mathbb{E}(a(X_i,\alpha_i,\mu) (Y_{iT}-X_{iT}'\beta_0-\alpha_i)|X_i,\alpha_i)] \\
&=\mathbb{E}[ a(X_i,\alpha_i,\mu) \mathbb{E}( Y_{iT}-X_{iT}'\beta_0-\alpha_i|X_i,\alpha_i)]=0,
\end{align*}
hence the adjustment term will always have mean zero. Therefore, $\mu$ satisfies \\
$\mathbb{E}[m^*(W_i,\alpha_i,Z_i,\beta_0,\mu,a)]=0$ if and only if $\mathbb{E}[m(W_i,\alpha_i,\mu]=0$. As long as $\mu_0$ can be identified by the original moment conditions $\mathbb{E}[m(W_i,Z_i,\mu)]=0$, then it can also be identified by the modified ones $\mathbb{E}[m^*(W_i,\alpha_i,Z_i,\beta_0,\mu,a)]=0$. This is because $\mathbb{E}[m^*(W_i,\alpha_i,Z_i,\beta_0,\mu,a)]=\mathbb{E}[m(W_i,\alpha_i,\mu]$ for any $\mu$ and $a(X_i,\alpha_i,\mu)$. 

Differentiating with respect to $\alpha_i$, and then taking conditional expectation with respect to $(X_i,\alpha_i)$, we have
\begin{align*}
\mathbb{E}\left( \frac{\partial m^*(W_i,\alpha_i,Z_i,\beta_0,\mu_0,a)}{\partial \alpha_i}|X_i,\alpha_i\right)&=\mathbb{E}\left( \frac{\partial m(W_i,\alpha_i,\mu_0)}{\partial \alpha_i}|X_i,\alpha_i\right)-a(X_i,\alpha_i,\mu_0)\\
&-\frac{\partial a(X_i,\alpha_i,\mu_0)}{\partial \alpha_i} \mathbb{E}(Y_{iT}-X_{iT}'\beta-\alpha_i|X_i,\alpha_i)\\
&= \mathbb{E}\left( \frac{\partial m(W_i,\alpha_i,\mu_0)}{\partial \alpha_i}|X_i,\alpha_i\right)-a(X_i,\alpha_i,\mu_0).
\end{align*}
The last step follows by sequential exogeneity assumption in (\ref{eq:seq exogeneity}). The Neyman Orthogonality Condition can be achieved by choosing $a(X_i,\alpha_i,\mu_0)$ so that the above display equals zero:\footnote{The fact that $\mathbb{E}(m^*(W_i,\alpha_i,Z_i,\beta_0,\mu_0,a))=0$ can be easily checked via law of iterated expectation by conditioning on $(X_i,\alpha_i)$ first.}
 
\[a(X_i,\alpha_i,\mu_0)= \mathbb{E}\left( \frac{\partial m(W_i,\alpha_i,\mu_0)}{\partial \alpha_i}|X_i,\alpha_i \right)\]

In addition, 
\[\mathbb{E}\left( \frac{\partial m^*(W_i,\alpha_i,Z_i,\beta_0,\mu_0,a)}{\partial a} \right) =-\mathbb{E}(Y_{iT}-X_{iT}'\beta-\alpha_i)=0\]
so the new moments $m^*(W_i,\alpha_i,Z_i,\beta_0,\mu_0,a)$ is also orthogonal to the new nuisance function $a(X_i,\alpha_i,\mu_0)$. With some abuse of notation, from now on, $a(X_i,\alpha_i,\mu_0)$ is used exclusively to denote $\mathbb{E}\left( \frac{\partial m(W_i,\alpha_i,\mu_0)}{\partial \alpha_i}|X_i,\alpha_i \right)$.

This is the strategy proposed and implemented for nuisance parameters that are typically non-parametric functions of the DGP (See e.g. \cite{dml,cher2022,ichimura}) rather than fixed effect estimates. The component $X_i$ in $a(\cdot)$ concatenates $X_{it}$'s  which will be a high-dimensional object under the data generating process of $N,T\rightarrow \infty$.\footnote{\cite{weidner2024} use the above construction but in their applications, $W_i$ only involves $X_i$. Consequently, $\mathbb{E}\left( \frac{\partial m(W_i,\alpha_i,\mu_0)}{\partial \alpha_i}|X_i,\alpha_i \right)=\frac{\partial m(W_i,\alpha_i,\mu_0)}{\partial \alpha_i}$, which is a significant simplification. } Some sparsity restriction is needed for the conditional expectation $\mathbb{E}\left( \frac{\partial m(W_i,\alpha_i,\mu_0)}{\partial \alpha_i}|X_i,\alpha_i\right)$ to be estimated with reasonable precision by employing standard high-dimensional regression tools like Elastic Net.
\medskip 

\textbf{Example 1 (continued)}

From the moment conditions in (\ref{eq:ex mom}), let $m_1(W_i,\alpha_i,\mu_0)=W_i -\mu_{01}-\mu_{02} \alpha_i$, and $m_2(W_i,\alpha_i,\mu)=\alpha_i(W_i-\mu_{01}-\mu_{02} \alpha_i)$ be the two moment functions. By taking derivatives, we see that the $a_1(X_i,\alpha_i,\mu_0)= -\mu_{02}$  and $a_2(X_i,\alpha_i,\mu_0) = \mathbb{E}(v_i - \mu_{02} \alpha_i|X_i,\alpha_i)$.

\subsection{Implementation}\label{procedure}

Constructing the orthogonal moment $m^*(W_i,\alpha_i,Z_i,\beta_0,\mu_0,a)$ requires estimates of $\alpha_i$, $\beta_0$, and $a(X_i,\alpha_i,\mu_0)$ which in turn needs a preliminary estimator of $\mu_0$. In order to avoid the recursiveness in estimating $\mu_0$ and control the bias due to using the same data to estimate both the nuisance parameters and the parameter of interest $\mu_0$, cross-fitting procedures are widely adopted in the literature \citep{dml,cher2022}. Although standard cross-fitting cannot deal with the estimation error in $\alpha_i$, it is still useful to control the error in the estimation of $a(X_i,\alpha_i,\mu_0)$ (see Assumption \ref{as:rr} below). In particular, sample splitting and cross-fitting enable us to form an estimate of the function $a(X_i,\alpha_i,\mu_0)$ without using observation $i$. Analysis of the plug-in orthogonal moments $m^*(W_i,\hat{\alpha}_i,Z_i,\hat{\beta},\mu,\hat{a})$ can then be performed by conditioning on the observations used to estimate $a(X_i,\alpha_i,\mu_0)$, and hence one can treat $\hat{a}(X_i,\alpha_i,\mu_0)$ as a fixed function and does not need to account for the effect of individual $i$ on estimating $a(X_i,\alpha_i,\mu_0)$.   Adapting the standard cross-fitting procedure to our context gives the following version:

\begin{algorithm}\label{alg}

\begin{enumerate}
	\item Split the cross-sectional dimension of the data into $L$ groups $\{I_l\}_{l=1}^L$ with $|I_l|\geq \lfloor N/L\rfloor$, for $l=1,\cdots,L$. Here $\lfloor x \rfloor$ denotes the greatest integer smaller than $x$. The number of folds $L$ is fixed and does not change with the sample size. 
	
	\item \textbf{Repeat for each $l=1,\cdots, L$}
	\begin{enumerate}
	
	   \item Goal of this step is to obtain a preliminary estimate $\tilde{\mu}_l$ of $\mu_0$ using data not in fold $l$, this is needed to estimate $a(X_i,\alpha_i,\mu_0)$ which depends on the true $\mu_0$: 
		
		\textbf{Repeat for each $l'\neq l$}
	  
		\begin{itemize}
		\item For observations $i\notin I_l \cup I_{l'}$ with $l'\neq l$ compute the Fixed effect regression (\ref{eq:1st stage}) and obtain $\tilde{\beta}_{ll'}$. 
		\item For each $i\in I_{l'}$, estimate $\tilde{\alpha}_{ill'}$ via $\bar{\bar{Y}}_i-\bar{\bar{X}}_i'\tilde{\beta}_{ll'}$, where $\bar{\bar{Y}}_i=\frac{1}{T-1} \sum_{t=1}^{T-1} Y_{it}$ and similarly for $\bar{\bar{X}}_i$.  Alternatively, one can further modify $\tilde{\alpha}_{ill'}$ using Empirical Bayes methods by shrinking the estimates $\tilde{\alpha}_{ill'}$ for $i\in I_{l'}$. 
	\end{itemize}
	
	\textbf{End for loop of $l'$}
	\medskip
	
			Using observations $i\notin I_l$, obtain an initial estimate $\tilde{\mu}_l$ by solving the GMM moment condition with the original (unorthogonalized) moments using $\{\tilde{\alpha}_{ill'}\}_{l'\neq l}$.
	
		\item Goal of this step is to obtain estimate $\alpha_i$ and also $\beta_0$ for fold $I_l$: 
		\medskip
		
		For observations $i\notin I_l$ , compute the Fixed effect regression (\ref{eq:1st stage}) and obtain $\tilde{\beta}_{l}$. For each $i\in I_l$, estimate $\tilde{\alpha}_{il}$ via $\bar{\bar{Y}}_i-\bar{\bar{X}}_i'\tilde{\beta}_{l}$. Alternatively, one can further modify $\tilde{\alpha}_{il}$ using Empirical Bayes methods by shrinking the estimates $\tilde{\alpha}_{il}$ for $i\in I_l$. 
		
	\item Goal of this step is to obtain an estimate $\tilde{a}_{il}$ of the $a(X_i,\alpha_i,\mu_0)$ using observations NOT in $I_l$.
	\medskip
	
	We assume that $a(X_i,\alpha_i,\mu_0)$ can be well approximated by the linear span of a dictionary of transformations of $(X_i,\alpha_i)$
	\begin{equation}
	\frac{\partial m(W_i,\tilde{\alpha}_{ill'},\tilde{\mu}_l)}{\partial \alpha_i} = b(X_i,\tilde{\alpha}_{ill'})'\pi_0 +\epsilon_i,
	\label{eq:rr}
	\end{equation}
	where $\tilde{\alpha}_{ill'}$ is the estimator for $\alpha_i$ in step (a). As one instance of the basis function $b(X_i,\tilde{\alpha}_{ill'})$, it can be $[X_{i1}',\cdots,X_{iT}',\tilde{\alpha}_{ill'}]'$ and possibly some higher-order interaction terms between $X_i$ and $\tilde{\alpha}_{ill'}$.  
	
	 Typically, $b(X_i,\alpha_i)$ will be a high-dimensional object and regularization is needed to ensure desirable performance. Denote the estimated coefficient $\pi_0$ as $\tilde{\pi}_l$, for $i\in I_l$, $a(X_i,\alpha_i,\mu_0)$ is estimated by $\tilde{a}_{il}(X_i,\tilde{\alpha}_{il})=b(X_i,\tilde{\alpha}_{il})'\tilde{\pi}_l$. 
	
	\end{enumerate}
	
	\textbf{End for loop of $l$}
	\medskip
	
	\item The debiased sample moment functions are: 
	\begin{align}
	\hat{m}^*(\mu)&\equiv \frac{1}{N} \sum_{l=1}^L \sum_{i\in I_l} m(W_i,\tilde{\alpha}_{il},\mu)+\hat{\psi},\nonumber \\
	\hat{\psi} &= \frac{1}{N} \sum_{l=1}^L \sum_{i\in I_l} \tilde{a}_{il} (Y_{iT}-X_{iT}' \tilde{\beta}_l -\tilde{\alpha}_{il}).
	\label{eq:orth sample mom}
	\end{align}
	
	\item The final estimate of $\mu_0$ is:
	\begin{equation}
	\hat{\mu} = \arg_{\mu} \min \hat{m}^*(\mu)' \Upsilon_N \hat{m}^*(\mu),
	\label{eq:mu est}
	\end{equation}
	where $\Upsilon$ is a positive-definite weighting matrix. The optimal weighting matrix is $\Upsilon=\Omega^{-1}$ with $\Omega=E(m^*(W_i,\alpha_i,Z_i,\beta_0,\mu_0,a)m^*(W_i,\alpha_i,Z_i,\beta_0,\mu_0,a)')$, and it can be estimated by 
	\begin{equation}
	\hat{\Upsilon}=(\frac{1}{N} \sum_{l=1}^L \sum_{i\in I_l} (m(W_i,\tilde{\alpha}_{il},\tilde{\mu}_l)+\hat{\psi}_{il})(m(W_i,\tilde{\alpha}_{il},\tilde{\mu}_l)+\hat{\psi}_{il})')^{-1},
	\label{eq:Omegahat}
	\end{equation}
	where $\hat{\psi}_{il}=\tilde{a}_{il} (Y_{iT}-X_{iT} \tilde{\beta}_l -\tilde{\alpha}_{il})$ .
	
	\item The estimator for the asymptotic variance takes the usual sandwich form:
	\begin{equation}
	\hat{V} = (\hat{G}'\hat{\Upsilon}\hat{G})^{-1} \hat{G}'\hat{\Upsilon} \hat{\Omega} \hat{\Upsilon} \hat{G} (\hat{G}'\hat{\Upsilon}\hat{G})^{-1},
	\label{eq:vhat}
	\end{equation}
	with
	\[\hat{G}=\frac{1}{N} \sum_{l=1}^L \sum_{i\in I_l}\frac{\partial m(W_i,\tilde{\alpha}_{il},\hat{\mu})}{\partial \mu},\]
	\[\hat{\Omega}=\frac{1}{N} \sum_{l=1}^L \sum_{i\in I_l} (m(W_i,\tilde{\alpha}_{il},\hat{\mu})+\hat{\psi}_{il})(m(W_i,\tilde{\alpha}_{il},\hat{\mu})+\hat{\psi}_{il})' .\]
	
\end{enumerate}

The user chosen quantities include the number of folds $L$ used in the cross-fitting, the dictionary of transformation $b(X_i,\tilde{\alpha}_{ill'})$ used in estimating $a(X_i,\alpha_i,\mu_0)$, the regularization parameters used in estimating $a(X_i,\alpha_i,\mu_0)$ and also whether shrinkage methods are used in estimating $\alpha_i$'s. 

\end{algorithm}
\bigskip

Before presenting the conditions that guarantee the validity of algorithm \ref{alg}, the following remarks are helpful for implementation.
\begin{remark}
\begin{enumerate}
	\item As in \cite{cher2022}, $\mu$ is not a variable to be minimized in $\hat{\psi}_{il}$. In fact, the $\mu$ that we are minimizing over when solving (\ref{eq:mu est}) only enters the original moment function $m(W_i,\alpha_i,\mu)$. 
	\item Unlike \cite{dml,cher2022} and \cite{ichimura}, the estimate of $\alpha_i$ for $i\in I_l$ cannot be made completely independent of samples in fold $l$. It is one particular feature of the fixed effect model that only data on the $i$-th individual is informative of the fixed effect $\alpha_i$. The strategy of constructing estimator of $\alpha_i$ is to use cross-fitting to obtain an estimator of $\tilde{\beta}_l$ not using sample $i$ first, and then use $\bar{\bar{Y}}_i-\bar{\bar{X}}_i'\tilde{\beta}_{l}$ to estimate $\alpha_i$. Under this construction, the difference $\tilde{\alpha}_{il}-\alpha_i$ depends on sample $i$ only through the panel data residuals $(u_{i1},\cdots,u_{iT-1})$.   With suitable controls on the correlation between $u_{it}$ and $\frac{\partial m(W_i,\alpha_i,\mu_0)}{\partial \alpha_i} -a(X_i,\alpha_i,\mu_0)$, the fact that $\tilde{\alpha}_{il}$ depends on $u_{it}$ will not affect first-order asymptotics (see Assumption \ref{as:alpha} below). 
	\item Estimating $a(X_i,\alpha_i,\mu_0)$ requires running a high-dimensional regression as in (\ref{eq:rr}). Machine learning methods that exploit sparsity of the function should be used to recover an estimate of $a(X_i,\alpha_i,\mu_0)$. Although lasso is the standard option in the literature \citep{dml,cher2022}, it may not be suitable in the current context. Specifically, \cite{EN} and \cite{AEN} document that lasso performs poorly when the regressors are highly correlated. Here the correlation between $X_{it}$'s and $\alpha_i$ could be high and make lasso undesirable. Instead, elastic net or adaptive elastic net incorporates an additional $L_2$ penalty which helps to deal with the high correlation between the regressors and could perform better than lasso \citep{AEN,EN}. 
	\item One can employ different estimators for $\tilde{\alpha}_{il}$ and $\tilde{\alpha}_{ill'}$ instead of what one typically obtains using standard fixed effect dummy variable regression. In particular, using Empirical Bayes to improve the estimator precision has become quite popular in the literature \citep{kwon,angrist2017,reviewteacher}. The orthogonal moment construction is important for Empirical Bayes estimates to be used in the context. Otherwise, the shrinkage will introduce bias in the estimates of $\alpha_i$ which may not disappear in the asymptotic distribution.  Moreover, simulation evidence suggests that using Empirical Bayes estimates leads to improvement in size control especially when the time dimension is relatively short.
	
	There are different ways to implement Empirical Bayes (EB) regularization of the estimates $\tilde{\alpha}_{il}$. The case of $\tilde{\alpha}_{ill'}$ is almost the same and is omitted. The basic idea is to assume that $\alpha_i \sim \mathcal{N}(\bar{\alpha},\sigma_{\alpha}^2)$, and employ a normal approximation to the usual dummy variable regression  $\tilde{\alpha}_{il}\sim \mathcal{N}(\alpha_i,u_i^2)$ with $u_i^2= \mbox{Var}(\frac{1}{T-1} \sum_{t=1}^{T-1} u_{it})$.\footnote{I ignore the bias of $\bar{X}_i'(\beta-\tilde{\beta}_{l})$ which converges to zero at rate $\sqrt{NT}$ as in \cite{kwon,angrist2017}. } The Bayes estimate of $\alpha_i$ with the prior $\alpha_i \sim \mathcal{N}(\bar{\alpha},\sigma_{\alpha}^2)$ is then
	\[\tilde{\alpha}_{il}^{ Bayes} = \frac{u_i^2}{u_i^2+ \sigma_{\alpha}^2} \bar{\alpha} +\frac{\sigma_{\alpha}^2}{\sigma_{\alpha}^2+u_i^2} \tilde{\alpha}_{il}\]
	
	We can then replace the unknown quantities by estimates of themselves: when $u_{it}$ are iid across time, $\hat{u}_i^2$ can be obtained by $\hat{u}_i^2=\frac{1}{(T-1)^2} \sum_{t=1}^{T-1} \hat{u}_{it}^2$ with $\hat{u}_{it}=Y_{it}-X_{it}' \tilde{\beta}_l -\tilde{\alpha}_{il}$. When $u_{it}$'s are serially correlated across time, standard Newey-West estimator that is robust to autocorrelation can be used. $\bar{\alpha}$ can be estimated by taking the sample mean of $\tilde{\alpha}_{il}$ for $i\in I_l$: $\hat{\bar{\alpha}}_{l} = \frac{1}{|I_l|} \sum_{i\in I_l}  \tilde{\alpha}_{il}$. Using the assumption that the cross-sectional draws are iid, $\frac{1}{|I_l|-1} \sum_{i\in I_l} u_i^2+ \sigma_{\alpha}^2$ can be estimated by $\frac{1}{|I_l|-1} \sum_{i\in I_l} (\tilde{\alpha}_{il}-\hat{\bar{\alpha}}_{l})^2$, which leads us to an estimate of $\sigma_{\alpha}^2$ as 
	\[\frac{1}{|I_l|-1} \sum_{i\in I_l} (\tilde{\alpha}_{il}-\hat{\bar{\alpha}}_{l})^2- \frac{1}{|I_l|-1} \sum_{i\in I_l} \hat{u}_i^2\]
	
	Another approach is suggested by \cite{kwon} and \cite{xie2012} who use Stein's Unbiased Risk Estimation (SURE) to select the best performing $\sigma_{\alpha}^2$ in terms of risk for a fixed estimator $\hat{\bar{\alpha}}_{l} = 
	\frac{1}{|I_l|} \sum_{i\in I_l}  \tilde{\alpha}_{il}$ of $\bar{\alpha}_l$.\footnote{\cite{kwon} also considers shrinkage to another constant or a pre-estimated function. Here I focus on shrinkage to the sample mean in the ith fold for simplicity. } In my context, using the squared loss, the risk of employing the empirical Bayes estimator $\tilde{\alpha}_{il}^{EB}$ for the fold $l$ is 
	\[R(\{\tilde{\alpha}_{il}^{EB}\}_{i\in I_l},\{\alpha_i\}_{i\in I_l})=\frac{1}{|I_l|} \sum_{i\in I_l}\mathbb{E}(\tilde{\alpha}_{il}^{EB}-\alpha_i)^2\]
	
	It turns out that the risk can be estimated via SURE for any choice of $\sigma_{\alpha}^2$. An unbiased estimator (conditional on $\alpha_i$ for $i\in I_l$) can be defined as $URE(\{\alpha_i\}_{i\in I_l},\sigma_{\alpha}^2)=\frac{1}{|I_l|} \sum_{i\in I_l} URE_i(\alpha_i,\sigma_{\alpha}^2)$ for (see pp.11 of \cite{kwon})
	\[URE_i(\alpha_i,\sigma_{\alpha}^2)=\hat{u}_i^2 + \frac{\hat{u}_i^4 (\hat{\alpha}_i-\hat{\bar{\alpha}})^2}{(\sigma_{\alpha}^2+\hat{u}_i^2)^2}-2\frac{\hat{u}_i^4}{\sigma_{\alpha}^2+\hat{u}_i^2}\]
	
	\cite{kwon} and \cite{xie2012} show that under general conditions, choosing the $\sigma_{\alpha}^2$ that minimizes $URE(\{\alpha_i\}_{i\in I_l},\sigma_{\alpha}^2)$ can perform asymptotically as good as minimizing the true loss $\frac{1}{|I_l|}  \sum_{i\in I_l} (\tilde{\alpha}_{il}^{EB}-\alpha_i)^2$. Let $\sigma_{\alpha,SURE}^2$ be the minimizer of $URE(\{\alpha_i\}_{i\in I_l},\sigma_{\alpha}^2)$. Then the SURE estimator of $\tilde{\alpha}_{il}^{ SURE}$ is defined as:
	\begin{equation}
	\tilde{\alpha}_{il}^{ SURE} = \frac{\hat{u}_i^2}{\hat{u}_i^2+ \sigma_{\alpha,SURE}^2} \hat{\bar{\alpha}}_i +\frac{\sigma_{\alpha,SURE}^2}{\sigma_{\alpha,SURE}^2+\hat{u}_i^2} \tilde{\alpha}_{il}
	\label{eq:alpha SURE}
	\end{equation}
	
	As both $\tilde{\alpha}_{il}^{EB}$ and $\tilde{\alpha}_{il}^{SURE}$ applies shrinkage that depends on all the data in fold $l$, the modifications are going to introduce dependence in the estimates of $\alpha_i$'s even with cross-fitting techniques. When dealing with the introduced dependence, I will use $\tilde{\alpha}_{il}^*$ to represent both cases where either the EB method or the SURE method is used to obtain more precise estimates of $\alpha_i$'s. 
	
\end{enumerate}

\end{remark}

\section{Asymptotic Theory}\label{as theory}

In this section, we give some high-level conditions that guarantee that the estimator $\hat{\mu}$ proposed in Section \ref{procedure} will be $\sqrt{N}$-asymptotically normal under the asymptotics that both $N,T \rightarrow \infty$, and $\hat{V}$ will be a consistent estimator of the asymptotic variance. Importantly, as in \cite{cher2022}, I will first give conditions such that the orthogonal moment $\hat{m}^*(\mu_0)$ defined in (\ref{eq:orth sample mom}) satisfies
\begin{equation}
\frac{1}{\sqrt{N}}\hat{m}^*(\mu_0) =\frac{1}{\sqrt{N}} \sum_{i=1}^N m^*(W_i,\alpha_i,Z_i,\beta_0,\mu_0,a) +o_p(1)
\label{eq:as equiv}
\end{equation}

\begin{assumption}\label{as:rr}
Suppose $T \propto N^g$ for some $1/2<g<1$. There exists some constant $1/2>\zeta \geq 1-g$ such that
\begin{enumerate}[label=(\roman*)]
	\item $\tilde{\beta}_l-\beta = O_p(\frac{1}{\sqrt{NT}})$ for all $l=1,\cdots,L$
	\item $\mathbb{E}[(\tilde{a}_{il}-a(X_i,\alpha_i,\mu_0))^2| \mathcal{W}_l^C]= o_p(N^{-\zeta})$ for all $l=1,\cdots,L$ where $\mathcal{W}_l^C$ denotes the data that are not in $I_l$. 
	\item The estimated function $\tilde{a}_{il}$ is square integrable. 
	\item $\frac{1}{|I_l|} \sum_{i\in I_l} (\tilde{a}_{il}-a(X_i,\alpha_i,\mu_0))^2=o_p(N^{-\zeta})$ for all $l=1,\cdots,L$. 
\end{enumerate}

\end{assumption}
\medskip

Condition (i) imposes a rate of convergence for the parameter $\beta$ in the first-stage panel data regression. The rate of convergence $O_p(\frac{1}{\sqrt{NT}})$ is standard for the panel data setting. Condition (ii) is a high-level convergence rate for prediction error of  the function $a(X_i,\alpha_i,\mu_0)$. The restriction $\zeta\geq 1-g$ is needed to control $(\tilde{\alpha}_{il}-\alpha_i)(\tilde{a}_{il}-a(X_i,\alpha_i,\mu_0)$. As $\tilde{\alpha}_{il}-\alpha_i = O_p(\frac{1}{\sqrt{T}})=O_p(N^{-\frac{g}{2}})$, the condition $\zeta\geq 1-g$ ensures that $(\tilde{\alpha}_{il}-\alpha_i)(\tilde{a}_{il}-a(X_i,\alpha_i,\mu_0)=o_p(\frac{1}{\sqrt{N}})$.  More primitive conditions can be specified depending on different types of learning algorithms that are used to learn the function $a(X_i,\alpha_i,\mu_0)$. Condition (iii) focuses attention on functions that are square-integrable which is standard in non-parametric regressions. 
Condition (iv) is implied by condition (ii) by observing that conditional on $\mathcal{W}_l^C$, we have

\[\frac{1}{|I_l|} \sum_{i\in I_l} [(\tilde{a}_{il}-a(X_i,\alpha_i,\mu_0))^2-\mathbb{E}((\tilde{a}_{il}-a(X_i,\alpha_i,\mu_0))^2|\mathcal{W}_l^C)]=O_p(N^{-1/2})=o(N^{-\zeta}),\]
where the first equality observes that it is an average of iid zero mean random variables and uses lemma 6.1 of \cite{dml} to convert conditional convergence to unconditional convergence. The second one uses $\zeta<1/2$ in Assumption \ref{as:rr}.

\begin{assumption}\label{as:alpha}
Let $\tilde{\alpha}_{il}^*$ be the estimator of $\alpha_i$ used in step 2(b) of Algorithm \ref{alg} with EB or SURE shrinkage. Assume $T \propto N^g$ for some $1/2<g<1$, and
\begin{enumerate}[label=(\roman*)]
	\item $\frac{1}{|I_l|} \sum_{i\in I_l} (\tilde{\alpha}_{il}^*-\alpha_i)^2 = O_p(T^{-1})$ for all $l=1,\cdots,L$. 
	\item $\mathbb{E}((\frac{\partial m(W_i,\alpha_i,\mu_0) }{\partial \alpha_i}- a(X_i,\alpha_i,\mu_0))^2\bar{\bar{u}}_i^2)=o(1)$, $\mathbb{E}((\frac{\partial m(W_i,\alpha_i,\mu_0) }{\partial \alpha_i}- a(X_i,\alpha_i,\mu_0))\bar{\bar{u}}_i)=o_p(N^{-1/2})$ for $\bar{\bar{u}}_i=\frac{1}{T-1} \sum_{t=1}^{T-1} u_{it}$. 
	\item $\mbox{Var}(\alpha_i)=\sigma_{\alpha}^2>0$
	\item $\mathbb{E}[(T \hat{u}_{i}^2)^{2+\delta}|\mathcal{W}_l^C]<\infty$ for $\delta$ that satisfies $g>\frac{1}{2}+\frac{1}{2+\delta}$
\end{enumerate}
Also assume the first two conditions hold with $\tilde{\alpha}_{il}^*$ replaced with $\tilde{\alpha}_{il}=\bar{\bar{Y}}_i-\bar{\bar{X}}_i'\tilde{\beta}_{l}$. \footnote{Condition (iii) and (iv) are not needed if no shrinkage method is used for estimating $\alpha_i$: $\tilde{\alpha}_{il}=\bar{\bar{Y}}_i-\bar{\bar{X}}_i'\tilde{\beta}_{l}$.} 

\end{assumption}
\medskip

Condition (i) of Assumption \ref{as:alpha} imposes a rate of convergence of the loss for the parallel estimation problem for all $\alpha_i$'s with $i\in I_l$.  Although asymptotically using the original $\tilde{\alpha}_{il}$ based on  averages over time leads to the same asymptotic results, for finite sample performance, using the Empirical Bayes variants and SURE can lead to null rejection probabilities that are closer to the nominal size. In addition, taking square of condition (i) on both sides, we see that
\[\frac{1}{N} \sum_{i,j\in I_l} (\tilde{\alpha}_{il}^*-\alpha_i)^2 (\tilde{\alpha}_{jl}^*-\alpha_i)^2=O_p(\frac{N}{T^2}),\]
which implies 
\[\frac{1}{N} \sum_{i\in I_l} (\tilde{\alpha}_{il}^*-\alpha_i)^4=O_p(\frac{N}{T^2})=o_p(1).\]
When $T \propto N^g$ for $g>1/2$, then $O_p(\frac{N}{T^2})=o_p(1)$.  

Condition (ii) is needed to ensure that $\frac{1}{\sqrt{N}} \sum_{i\in I_l} (\frac{\partial m(W_i,\alpha_i,\mu_0) }{\partial \alpha_i}- a(X_i,\alpha_i,\mu_0))(\tilde{\alpha}_{il}^*-\alpha_i)=o_p(1)$. Even if $\tilde{\alpha}_{il}$ is constructed using data on individual $i$, the condition makes the correlation of the estimation error with the mean zero term $\frac{\partial m(W_i,\alpha_i,\mu_0) }{\partial \alpha_i}- a(X_i,\alpha_i,\mu_0)$ negligible. 
The moment conditions are not restrictive as it involves the time averages $\bar{\bar{u}}_i$. Typically, under standard assumptions on weakly dependent data $\sqrt{T} \bar{\bar{u}}_i=O_p(1)$, hence one would expect $T \mathbb{E}((\frac{\partial m(W_i,\alpha_i,\mu_0) }{\partial \alpha_i}- a(X_i,\alpha_i,\mu_0))^2\bar{\bar{u}}_i^2)<\infty$, justifying the first part of condition (ii). The second part is similar to condition (ii) in assumption 1 of \cite{battaglia2024}. A sufficient condition is $\frac{\partial m(W_i,\alpha_i,\mu_0) }{\partial \alpha_i}- a(X_i,\alpha_i,\mu_0)$ being mean independent of the time varying idiosyncratic shocks $u_{it}$, which guarantees that $\mathbb{E}((\frac{\partial m(W_i,\alpha_i,\mu_0) }{\partial \alpha_i}- a(X_i,\alpha_i,\mu_0))\bar{\bar{u}}_i)=0$.  This is the case when (\ref{eq:1st stage}) is a dynamic panel data model ($X_{it}$ includes $Y_{i,t-1}$): for any $t<T$,

\begin{align*}
& \mathbb{E}\left( u_{it} \frac{\partial m(W_i,\alpha_i,\mu_0) }{\partial \alpha_i} \right)= \mathbb{E}[\mathbb{E}\left(u_{it} \frac{\partial m(W_i,\alpha_i,\mu_0) }{\partial \alpha_i} | \alpha_i, X_i \right)]\\
&= \mathbb{E}[u_{it} \mathbb{E}\left(\frac{\partial m(W_i,\alpha_i,\mu_0) }{\partial \alpha_i} | \alpha_i, X_i\right)]= \mathbb{E}(u_{it} a(X_i,\alpha_i,\mu_0)),
\end{align*}
where the second last line follows because $X_i$ includes $Y_{it}$ and $X_{it}$, so $u_{it}$ for $t<T$ is a deterministic function of $X_i$ and $\alpha_i$. If one constructs the estimator $\tilde{\alpha}_{il}$ using only data up to period $T-1$, then the averages will be orthogonal to the deviation $\frac{\partial m(W_i,\alpha_i,\mu_0) }{\partial \alpha_i}- a(X_i,\alpha_i,\mu_0)$. When (\ref{eq:1st stage}) does not include a dynamic specification but the researcher is willing to assume that $u_{it}$'s are mean independent across $t$, one can simply add $(Y_{i1},\cdots,Y_{it-1})$ into $X_i$. We do this in the empirical application in Section \ref{application}. A similar law of iterated expectation argument can show that condition (ii) still holds.

 The last two conditions are only needed when one uses data in fold $l$ to perform shrinkage of $\tilde{\alpha}_{il}$. Condition (iii) requires that there exists non-trivial individual heterogeneity in $\alpha_i$'s, the assumption is needed to ensure that in constructing the EB or SURE type estimators the shrinkage factor $\frac{u_{i}^2}{u_{i}^2+ \sigma_{\alpha}^2}$ will not involve division by a number close to zero. Combined with condition (iii), condition (iv) is a mild moment existence condition that guarantees that $\sup_{i\in I_l}\frac{\hat{u}_{i}^2}{\hat{u}_{i}^2+ \hat{\sigma}_{\alpha}^2} =o_p(N^{-1/2})$. For instance, under the assumption that $u_{it}$ is iid across both time and individuals, an estimator for $u_{i}^2=\mbox{Var}(\frac{1}{T-1} \sum_{t=1}^{T-1} u_{it})$ is $\frac{1}{(T-1)^2} \sum_{t=1}^{T-1} (\hat{u}_{it})^2$ for $\hat{u}_{it}=Y_{it}-X_{it}'\tilde{\beta}_l-\tilde{\alpha}_{il}$. Condition (iv) then requires that 
\[\mathbb{E}\{[\frac{1}{T-1}\sum_{t=1}^{T-1} (\hat{u}_{it})^2]^{2+\delta}\}<\infty ,\]
which holds under general moment existence conditions on $(X_{it},u_{it},\alpha_i)$. A similar argument can be constructed for the case where $u_{it}$ might be serially correlated across time. To control $\sup_{i\in I_l}\hat{u}_{i}^2$, there exists a trade-off between the existence of higher moments and the growth rate of the time dimension $T$ with $N$, as manifested in the condition $g>\frac{1}{2}+\frac{1}{2+\delta}$. In general, a smaller growth rate implies that we are taking supremum over a larger group of individuals for the same level of randomness due to averaging across time. Hence higher-order moments are needed to ensure that the quantity is $o_p(N^{-1/2})$. 

In addition to the rate conditions imposed above, we also collect the needed conditions on existence of moments:
\begin{assumption}\label{as:mom}

\begin{enumerate}[label=(\roman*)]
  \item $\sup_{\alpha,W_i} \| \frac{\partial^2 m(W_i,\alpha_i,\mu_0) }{\partial \alpha_i^2} \|=C_m <\infty$, $\mathbb{E}(u_{iT}^2|X_i,\bar{\bar{Y}}_i,\alpha_i)\leq C<\infty$ with probability approaching 1
	\item $\mathbb{E}(\|X_{it}\|^2)< \infty$ for all $t=1,\cdots,T$
	\item $\mathbb{E}(a(X_i,\alpha_i,\mu_0)^2 \| X_{iT}\|^2)< \infty$, $\mathbb{E}(\tilde{a}_{il}^2 \|X_{iT}\|^2)< \infty$,  
	$\mathbb{E}(\frac{\partial m(W_i,\alpha_i,\mu_0) }{\partial \alpha_i}- a(X_i,\alpha_i,\mu_0))^4<\infty$
	\item $\mathbb{E}(\|(\frac{\partial m(W_i,\alpha_i,\mu_0) }{\partial \alpha_i}- a(X_i,\alpha_i,\mu_0))\bar{\bar{X}}_i'\|)<\infty$ for $\bar{X}_i=\frac{1}{T} \sum_{t=1}^T X_{it}$. 
\end{enumerate}
\end{assumption}
\medskip

Condition (i) imposes a bound on the second-order derivative of $m(W_i,\alpha_i,\mu_0)$ with respect to $\alpha_i$. In Example \ref{ex:ols}, the moment function is quadratic in $\alpha_i$, and it will be satisfied automatically. In addition, $\mathbb{E}(u_{iT}^2|X_i,\bar{\bar{Y}}_i,\alpha_i)$ is a restriction on the conditional variance of $u_{iT}$. Conditions (ii)-(iv) guarantee the existence of moments so that suitable Law of Large Numbers can be applied.  

Assumption \ref{as:mom eb} lists additional moment conditions that are needed when $\tilde{\alpha}_{il}^*$ uses the EB or SURE variant. 

\begin{assumption}\label{as:mom eb}
The following assumption is needed for estimates of $\alpha_i$ that applies Empirical Bayes shrinkage or SURE:

	\[\mathbb{E}(|(\frac{\partial m(W_i,\alpha_i,\mu_0) }{\partial \alpha_i}- a(X_i,\alpha_i,\mu_0))(\alpha_i+\bar{\bar{u}}_i-\mathbb{E}(\alpha))|)<\infty\]
\end{assumption}
\medskip

\begin{proposition}\label{prop:omega}
Under Assumption \ref{ass:cross}-\ref{as:mom eb}, and Assumption \ref{as:mom eb} is needed only if EB or SURE variants of estimators $\tilde{\alpha}_{il}^*$ for $\alpha_i$ are used, for $\zeta$ and $g$ satisfying $\zeta+g\geq 1$, we have
\begin{align}
& \frac{1}{\sqrt{N}} \sum_{l=1}^L \sum_{i\in I_l} m(W_i,\tilde{\alpha}_{il},\mu_0)+\tilde{a}_{il} (Y_{iT}-X_{iT}' \tilde{\beta}_l -\tilde{\alpha}_{il})\nonumber \\
& =\frac{1}{\sqrt{N}} \sum_{i=1}^N m(W_i,\alpha_i,\mu_0)+ a(X_i,\alpha_i,\mu_0)(Y_{iT}-X_{iT}'\beta_0-\alpha_i) +o_p(1)
\label{eq:omega}
\end{align}

\end{proposition}

\medskip

Proposition \ref{prop:omega} is the key result that establishes that under the orthogonal moment construction and cross-fitting, the asymptotic distribution is the same as if the true $\alpha_i$ and the true function $a(X_i,\alpha_i,\mu_0)$ are known at $\mu=\mu_0$. The fact that both $\alpha_i$ and the conditional mean function $a(X_i,\alpha_i,\mu_0)$ are estimated will not affect the asymptotic distribution of the estimator $\hat{\mu}$ other than accounting for the addition of the orthogonal correction term $a(X_i,\alpha_i,\mu_0)(Y_{iT}-X_{iT}'\beta_0-\alpha_i)$. For $m^*(W_i,\alpha_i,Z_i,\beta_0,\mu_0,a)= m(W_i,\alpha_i,\mu_0)+ a(X_i,\alpha_i,\mu_0)(Y_{iT}-X_{iT}'\beta_0-\alpha_i)$, the second line of (\ref{eq:omega}) is a normalized sum of iid mean zero random variables, under the condition that $\Omega\equiv \mathbb{E}(m^*(W_i,\alpha_i,Z_i,\beta_0,\mu_0,a) m^*(W_i,\alpha_i,Z_i,\beta_0,\mu_0,a)')$ (with $\mathbb{E}(\|m^*(W_i,\alpha_i,Z_i,\beta_0,\mu_0,a)\|^2)<\infty$), standard Central Limit Theorem will imply
\[\frac{1}{\sqrt{N}} \sum_{i=1}^N m^*(W_i,\alpha_i,Z_i,\beta_0,\mu_0,a) \xrightarrow{d} \mathcal{N}(0,\Omega)\]

The following condition ensures that $\Omega$ can be consistently estimated by $\hat{\Upsilon}^{-1}$ in (\ref{eq:Omegahat}).

\begin{assumption}\label{as:mom omega}

\begin{enumerate}[label=(\roman*)]
	\item $(\frac{1}{N} \sum_{i\in I_l} (\tilde{a}_{il}-a(X_i,\alpha_i,\mu_0))^4)(\frac{1}{N} \sum_{i\in I_l} (\alpha_i-\tilde{\alpha}_{il})^4)=o_p(1)$
	\item $\mathbb{E}(\|m^*(W_i,\alpha_i,Z_i,\beta_0,\mu_0,a)\|^2)<\infty$
	\item $\tilde{\mu}_l \xrightarrow{p} \mu_0$
	\item $\frac{\partial m(W_i,\alpha_i,\mu)}{\partial \mu'}$, $\frac{\partial^2 m(W_i,\alpha_i,\mu)}{\partial \alpha_i \partial \mu'}$ and $\frac{\partial^3 m(W_i,\alpha_i,\mu)}{\partial \alpha_i^2 \partial \mu'}$ exists and is continuous in a neighborhood $\mathcal{N}(\mu_0)$ of $\mu_0$ and for almost all $\alpha_i$. In addition, $\sup_{\alpha_i,W_i} \sup_{\mu\in \mathcal{N}(\mu_0)} \frac{\partial^3 m(W_i,\alpha_i,\mu)}{\partial \alpha_i^2 \partial \mu'} <C_m<\infty$. 
	\item $\mathbb{E}(\sup_{\mu\in \mathcal{N}(\mu_0)}\|\frac{\partial^2 m(W_i,\alpha_i,\mu)}{\partial \alpha_i \partial \mu'}\|^4)<\infty$,
	
\end{enumerate}

\end{assumption}
\medskip

Condition (i) is not restrictive given $\frac{1}{N} \sum_{i\in I_l} (\alpha_i-\tilde{\alpha}_{il})^4=o_p(1)$ and condition (ii) in Assumption \ref{as:rr}. Conditions (ii) and (v) are suitable regularity conditions in order to apply Law of Large Numbers. 
Condition (iii) assumes consistency of the preliminary estimator of $\tilde{\mu}_l$, which is established in Appendix \ref{ap:orth}. Condition (iv) assumes that the function $m(W_i,\alpha_i,\mu)$ is sufficiently smooth in both $\alpha_i$ and $\mu$, which is satisfied in Example \ref{ex:ols} with $m(W_i,\alpha_i,\mu)$ being quadratic in $\alpha_i$ and linear in $\mu$.

\begin{proposition}\label{prop:omegahat}
Under Assumption \ref{ass:cross}-\ref{as:mom omega} (Assumption \ref{as:mom eb} is needed only when EB or SURE variants are used to estimate $\alpha_i$), we have
\begin{equation}
\frac{1}{N} \sum_{l=1}^L \sum_{i\in I_l} (m(W_i,\tilde{\alpha}_{il},\tilde{\beta}_l,\tilde{\mu}_l,\tilde{a}_{il})+\hat{\psi}_{il})(m(W_i,\tilde{\alpha}_{il},\tilde{\beta}_l,\tilde{\mu}_l,\tilde{a}_{il})+\hat{\psi}_{il})' \xrightarrow{p} \Omega
\label{eq:con omega}
\end{equation}

\end{proposition}
\medskip

To establish the validity of (\ref{eq:vhat}) for estimating the asymptotic variance of $\hat{\mu}$, we still need to show that $\hat{G}$ consistently estimates
$G=\mathbb{E}(\frac{\partial m(W_i,\alpha_i,\mu_0)}{\partial \mu})$. The following conditions ensure the convergence of the estimator for the Jacobian matrix $G$, and is similar to assumption 5 of \cite{cher2022}.

\begin{assumption}\label{as:G}
The limit Jacobian matrix $G$ exists and is full-rank. Moreover, for some function $F(W_i)$

$\sup_{\alpha} \|\frac{\partial^2 m(W_i,\alpha_i,\mu)}{\partial \alpha_i \partial \mu'}\|^2\leq F(W_i)$ for all $\mu$ within a neighborhood $\mathcal{N}(\mu_0)$ of $\mu_0$ and $\mathbb{E}(F(W_i))<\infty$. 

In addition, for $\mu\in \mathcal{N}(\mu_0)$, some constant $C>0$ and some function $F_{\alpha}(W_i,\alpha_i)$
 \begin{equation}
 \|\frac{\partial m(W_i,\alpha_i,\mu)}{\partial \mu}-\frac{\partial m(W_i,\alpha_i,\mu_0)}{\partial \mu}\|\leq F_{\alpha}(W_i,\alpha_i) \|\mu-\mu_0\|^{1/C}, \quad \mathbb{E}(F_{\alpha}(W_i,\alpha_i))<\infty. 
\label{eq:G no mu}
\end{equation}

\end{assumption}
\medskip

Assumption \ref{as:G} is easy to verify in Example \ref{ex:ols}. With $m(W_i,\alpha_i,\mu)$ being quadratic in $\alpha_i$ and linear in $\mu$, we can take $F(W_i)=0$ and also $F_{\alpha}(W_i,\alpha_i)= 0$ with $C=1$.

\begin{proposition}\label{prop:G}
Under Assumptions \ref{as:alpha}, Assumption \ref{as:mom omega} and Assumption \ref{as:G}, if $\hat{\mu}\xrightarrow{p}\mu_0$, then 
\begin{equation}
\frac{1}{N} \sum_{l=1}^L \sum_{i\in I_l}\frac{\partial m(W_i,\tilde{\alpha}_{il},\hat{\mu})}{\partial \mu'} \xrightarrow{p} G
\label{eq:con G}
\end{equation} 

\end{proposition}
\medskip

The next theorem combines the above propositions to establish the validity of inference for $\hat{\mu}$:

\begin{theorem}\label{thm1}
Under Assumptions \ref{ass:cross}-\ref{as:G} (Assumption \ref{as:mom eb} is needed only when $\alpha_i$'s are estimated using either EB or SURE corrections), and if $\hat{\mu}$ defined in (\ref{eq:mu est}) is consistent for $\mu_0$: $\hat{\mu}\xrightarrow{p} \mu_0$, then it satisfies the following asymptotic expansion:
\begin{equation}
\sqrt{N} (\hat{\mu}-\mu_0)\xrightarrow{d} \mathcal{N}(0,V),\quad V= (G'\Upsilon G)^{-1} G'\Upsilon \Omega \Upsilon G (G'\Upsilon G)^{-1}.
\label{eq:mu asym}
\end{equation}
with $G= \mathbb{E}(\frac{\partial m(W_i,\alpha_i,\mu_0)}{\partial \mu})$, $\Omega=\mathbb{E}(m^*(W_i,\alpha_i,Z_i,\beta_0,\mu_0,a) m^*(W_i,\alpha_i,Z_i,\beta_0,\mu_0,a)')$ and $\Upsilon$ being the probability limit of the weighting matrix. 

In addition, a consistent estimator for the asymptotic variance $V$ is given by (\ref{eq:vhat}):
\begin{equation}
\hat{V} = (\hat{G}'\hat{\Upsilon}\hat{G})^{-1} \hat{G}'\hat{\Upsilon} \hat{\Omega} \hat{\Upsilon} \hat{G} (\hat{G}'\hat{\Upsilon}\hat{G})^{-1}\xrightarrow{p} V
\label{con:v}
\end{equation}

\end{theorem}
\medskip

Theorem \ref{thm1} requires consistency of $\hat{\mu}$ for $\mu_0$, which follows under standard assumptions. The conditions required are specified in Appendix \ref{ap:orth}.

\section{Monte-Carlo Simulation}\label{simulation}

We will compare the estimator and inference procedure in \ref{procedure} with those in \cite{xie2025} and \cite{CGK2025} through the lens of example \ref{ex:ols}. More specifically, we will consider a panel auto-regression model as the first stage:
\begin{equation}
Y_{it}=\beta_0 Y_{it-1}+\alpha_i+u_{it}, \text{ for } t=1,\cdots,T
\label{eq:1st dynamic}
\end{equation}

The fixed effects $\alpha_i$'s are independent $\mathcal{N}(0,\frac{1}{2})$ draws, and $u_{it}$'s (unobserved) are also $\mathcal{N}(0,\frac{1}{2})$ constructed from the sum of two independent random variables $u_{1it}+u_{2it}$ each distributed as $\mathcal{N}(0,\frac{1}{4})$. Both $u_{1it}$'s and $u_{2it}$'s are also independent across time and individuals. The initial condition $Y_{i0}$ is generated from the stationary distribution as $Y_{i0}\sim \mathcal{N}(\frac{\alpha_i}{1-\beta_0},\frac{1}{2(1-\beta_0^2)})$. 

The second stage is specified as in (\ref{eq:ex}) with a slight modification:

\begin{equation}
W_i = \alpha_i + v_i - c \frac{1}{ \lfloor \frac{T}{2} \rfloor -1 } \sum_{t=1}^{ \lfloor \frac{T}{2} \rfloor -1 } u_{1it},
\label{eq:2nd ols}
\end{equation}
where $\lfloor x \rfloor$ denotes the largest integer that is smaller than or equal to $x$. 
In the notation of example \ref{ex:ols}, $\mu_{01}=0$ and $\mu_{02}=1$, and the parameter of interest is $\mu_{02}$. 
We draw $v_i$ from $\mathcal{N}(0,1)$ independent from all other variables and $c$ is a scaling factor that gauges the level of correlation between $\bar{\bar{u}}_i$ and $v_i-c\frac{1}{ \lfloor \frac{T}{2} \rfloor -1 } \sum_{t=1}^{ \lfloor \frac{T}{2} \rfloor -1 } u_{1it}$. For $c\neq 0$, then $\bar{\bar{u}}_i$ and $v_i-c \frac{1}{ \lfloor \frac{T}{2} \rfloor -1 } \sum_{t=1}^{ \lfloor \frac{T}{2} \rfloor -1 } u_{1it}$ are correlated, which violates the conditions for inference in both \cite{xie2025} and \cite{CGK2025}. In contrast, as explained in Section \ref{as theory}, we expect the orthogonal moment condition to be insensitive to the potential correlation between them and have null rejection probability close to the nominal size. 

As \cite{xie2025} and \cite{CGK2025} abstracts away from the problem of estimating $\beta_0$, we employ estimators from \cite{BB} by constructing the following moment conditions:
\[\mathbb{E}\begin{bmatrix} \Delta y_{it}-\beta \Delta y_{it-1}\\
y_{it-2} (\Delta y_{it}-\beta \Delta y_{it-1}) \\
y_{it-3} (\Delta y_{it}-\beta \Delta y_{it-1})\\
\Delta y_{it-1} (y_{it}-\beta y_{it-1})
\end{bmatrix}  =0 \quad \text{for } t=4,\cdots,T. \]
 $\beta_0$ is then estimated via two-step GMM.\footnote{I only included moments $t>3$ for the simulation results with $T=12$ and $t>5$ for simulation results with $T=22$ in order to avoid multicollinearity in constructing the optimal weighting matrix. The cutoffs are chosen based on finite sample performance in the simulations.}

To implement procedures in  \cite{xie2025}, I consider both $\hat{\alpha}_i=\frac{1}{T} \sum_{t=1}^{T} Y_{it}- \frac{1}{T} \sum_{t=0}^{T-1} Y_{it} \tilde{\beta}$ and also its EB ($\alpha_i^{EB}$) and SURE ($\alpha_i^{SURE}$) variants. \cite{xie2025} estimates the standard error of $\mu_{02}$  with the usual heteroskedastic robust standard error formula by plugging the EB versions of $\alpha_i^{EB}$. 

In addition, I also consider the estimator recommended in \cite{CGK2025} which builds on the observation that 
\[\mu_{02} = \frac{Cov(\alpha_i,W_i)}{\mbox{Var}(\alpha_i)}=\frac{Cov(\hat{\alpha}_i,W_i)}{ \mbox{Var}(\hat{\alpha}_i)} \frac{\mbox{Var}(\hat{\alpha}_i) }{\mbox{Var}(\alpha_i)},\]
exploiting independence between $\bar{u}_i$ and $W_i$ in their assumptions (which does not hold when $c\neq 0$). They suggest replacing each population quantity by its sample analogue. In particular, $\mbox{Var}(\alpha_i)$ can be estimated by $\hat{\mbox{Var}}(\hat{\alpha}_i)-\hat{\mbox{Var}}(\bar{u}_i)$, where $\hat{\mbox{Var}}(\bar{u}_i)$ is a consistent estimator of the variance of $\bar{u}_i$ and can be constructed as $\frac{1}{N} \sum_{i=1}^N \frac{1}{T^2} \sum_{t=1}^T \hat{u}_{it}^2 $. Inference can be conducted using percentile bootstrap. 

For the orthogonal moment constructions, I followed the procedures in Section \ref{procedure} and uses $\tilde{\alpha}_{il}=\frac{1}{T-1} \sum_{t=1}^{T-1} Y_{it}-\frac{1}{T-1} \sum_{t=0}^{T-2} Y_{it}$ so that condition (ii) of assumption \ref{as:alpha} is guaranteed to hold, and also for its EB and SURE variants $\tilde{\alpha}_{il}^{EB}$ and $\tilde{\alpha}_{il}^{SURE}$. A preliminary estimator of $\tilde{\mu}_l=(\tilde{\mu}_{1l},\tilde{\mu}_{2l})$ is then obtained via linear regression. The estimator of conditional mean function $a_1(X_i,\alpha_i,\mu_0)=-\mu_{02}$ for $i\in I_l$ is constructed by $-\tilde{\mu}_{2l}$ and the conditional mean function $a_2(X_i,\alpha_i,\mu_0) = \mathbb{E}(v_i - \mu_{02} \alpha_i|X_i,\alpha_i)$ is constructed via an adaptive elastic net regression of $\hat{v}_j-\tilde{\mu}_{2l} \tilde{\alpha}_{jl} $ on $\tilde{\alpha}_{jl} $ and $(Y_{j0},Y_{j1},\cdots, Y_{jT-1})$ for $j\notin I_l$. Here $\hat{v}_j=W_j-\tilde{\mu}_{1l}-\tilde{\mu}_{2l} \tilde{\alpha}_{jl}$. The penalty factors for the L1 regularization and L2 regularization terms in the Elastic Net regression are chosen using cross-validation. To alleviate potential noise introduced in the sample splitting process, we computed $\hat{\mu}$ and $\hat{V}$ using orthogonal moments for 20 different realizations of sample splits (randomness coming from reshuffling of the sample along the cross-sectional dimension) and then averaged $\hat{\mu}$ and $\hat{V}$ across them.

Table \ref{tab:rej} reports the simulation results for $N=100$, $T=12$ and $\beta_0=0$ and three different choices of $c$ across 1,000 simulations. We also compare performance of the inference procedure with different number of folds $L$.\footnote{As explained in \cite{BB}, larger $\beta_0$ will make the estimates of $\beta$ very imprecise, resembling a weak instrument problem. Our asymptotic theory essentially ignores any noise in estimating $\beta_0$, and poor estimates of $\beta_0$ will propagate and affect the inference results for $\beta_0$ that is too large.} 

\begin{table}[ht]
\scriptsize

\centering
\begin{tabular}{c c  c c c  || c c c|| c c c}
\toprule
~ & CGK & Naive & Xie-EB & Xie-SURE & Orth-mean & Orth-EB & Orth-SURE & Orth-mean & Orth-EB & Orth-SURE\\
\multicolumn{5}{l}{} & \multicolumn{3}{c}{$L=5$} & \multicolumn{3}{c}{$L=10$}\\
\hline
\multicolumn{11}{l}{Panel A: $c=0$}\\
Bias & 0.006 & -0.057 & 0.006 & 0.011 & 0.029 & 0.038 & 0.038 & 0.029 & 0.037 & 0.038 \\
Std & 0.158 & 0.146 & 0.158 & 0.160 & 0.185 & 0.193 & 0.194 & 0.185 & 0.193 & 0.195 \\
RMSE & 0.158 & 0.157 & 0.158 & 0.160 & 0.187 & 0.197 & 0.198 & 0.187 & 0.197 & 0.198 \\
Rej Prob & 0.054 & 0.212 & 0.215 & 0.217 & 0.087 & 0.078 & 0.075 & 0.087 & 0.080 & 0.078 \\
\\
\midrule
\multicolumn{11}{l}{Panel B: $c=4$}\\
Bias & -0.167 & -0.219 & -0.167 & -0.163 & -0.016 & 0.015 & 0.019 & -0.014 & 0.021 & 0.027 \\
Std & 0.214 & 0.199 & 0.214 & 0.215 & 0.253 & 0.269 & 0.269 & 0.253 & 0.268 & 0.270 \\
RMSE & 0.271 & 0.296 & 0.271 & 0.269 & 0.253 & 0.269 & 0.270 & 0.253 & 0.269 & 0.272 \\
Rej Prob & 0.143 & 0.404 & 0.342 & 0.335 & 0.077 & 0.065 & 0.064 & 0.077 & 0.065 & 0.065 \\
\\
\midrule
\multicolumn{11}{l}{Panel C: $c=5$}\\
Bias & -0.211 & -0.260 & -0.210 & -0.207 & -0.027 & 0.010 & 0.015 & -0.028 & 0.015 & 0.023 \\
Std & 0.238 & 0.223 & 0.238 & 0.239 & 0.283 & 0.301 & 0.303 & 0.283 & 0.301 & 0.304 \\
RMSE & 0.318 & 0.342 & 0.317 & 0.316 & 0.284 & 0.301 & 0.303 & 0.284 & 0.301 & 0.304 \\
Rej Prob & 0.170 & 0.435 & 0.370 & 0.362 & 0.074 & 0.062 & 0.061 & 0.078 & 0.067 & 0.061 \\
\\
\bottomrule
\end{tabular}

\caption{Comparisons in performance of different estimators and Inference procedures}
\label{tab:rej}
\vspace{0.2cm}

\begin{minipage}{15cm}
\scriptsize Note: This table reports the performance of various estimators for $\mu_{02}$. I report the bias, standard error of the estimates $\hat{\mu}_2$, root mean squared error (RMSE) and also the null rejection probability with target size 0.05 across three different DGPs. The sample size is $N=100$, $T=12$, and the auto-regressive coefficient $\beta_0=0$. Panel A considers $c=0$ with no correlation between $\bar{\bar{u}}_i$ and $W_i-\alpha_i$, while Panel B and C sets $c=4$ and $c=5$ respectively. CGK corresponds to the measurement error correction estimator suggested in \cite{CGK2025}. Naive, Xie-EB and Xie-SURE uses heteroskedastic robust standard error for $\hat{\alpha}_i$, $\hat{\alpha}_i^{EB}$ and $\hat{\alpha}_i^{SURE}$ respectively. Orth-mean, Orth-EB and Orth-SURE uses the orthogonal moment constructions proposed in this paper but uses $\tilde{\alpha}_{il}$, $\tilde{\alpha}_{il}^{EB}$ and $\tilde{\alpha}_{il}^{SURE}$ respectively. Columns 6-8 report results when the data is split into 5 equal-sized folds, while column 9-11 reports the results with 10 folds. 
\end{minipage}

\end{table}

The first column reports the results for estimators suggested by \cite{CGK2025} and the next three columns report results where heteroskedastic standard error for linear regression is used to compute asymptotic variance with different estimators of $\alpha_i$. Finally, the last six columns employs the orthogonal moment construction with different estimators of $\alpha_i$ and with different number of folds. Comparing the bias of the estimators across different panels, the results show that when the endogeneity increases ($c$ is large), there is significant bias in the estimates using \cite{CGK2025} and \cite{xie2025}. Although these estimators typically has smaller standard error, it does not compensate for the bias. Therefore, estimators using orthogonal moments enjoy smaller Root Mean Squared Error (RMSE) for larger values of $c$. 

  We also consider testing the null hypothesis of $\mu_{02}=1$ using t-test based on different estimators and their estimated standard errors. The null rejection probability is reported in the last row of each panel with target size 0.05. We see that the endogneity between the measurement error and the residuals in (\ref{eq:2nd ols}) made the estimators of \cite{CGK2025} and \cite{xie2025} suffer from severe size distortion. In contrast, estimators that uses orthogonal moments have size control within acceptable levels.  In addition, using the EB or SURE correction delivers better size control than the original $\tilde{\alpha}_{il}$. Nevertheless, the shrinkage correction inflates the RMSE compared with the case where no correction is applied. The simulation results suggest that in general, using orthogonal moment conditions is more important than using EB or SURE corrections to mitigate the concerns of measurement error especially when endogeneity is a concern. In addition, comparing results from column 6-8 with the results from column 9-11, we see that the number of folds $L$ does not have large impacts on either the estimates or the null rejection probabilities. One can choose a smaller fold $L=5$ for computational reasons in the current DGP. 

Figure \ref{fig:power 100 13} plots the power function of  the CGK estimator, the estimator of \cite{xie2025} using SURE and estimator using orthogonal moment conditions and SURE corrections for the DGP in panel A where all estimators are expected to have null rejection probability close to 0.05. We see that a finite-sample distortion is present for the estimator proposed in \cite{xie2025}, while the estimator from \cite{CGK2025} and estimator based on orthogonal moments have smaller null rejection probabilities. In addition, the estimator proposed in \cite{CGK2025} has higher power in rejecting the alternatives compared with the estimator using orthogonal moments. The phenomenon reflects the fact that estimators from \cite{CGK2025} achieves the semi-parametric efficiency bound when endogeneity concern is not present, and in this case their estimator should be preferred for power purposes. 

\begin{figure}
\includegraphics[width=.9\textwidth]{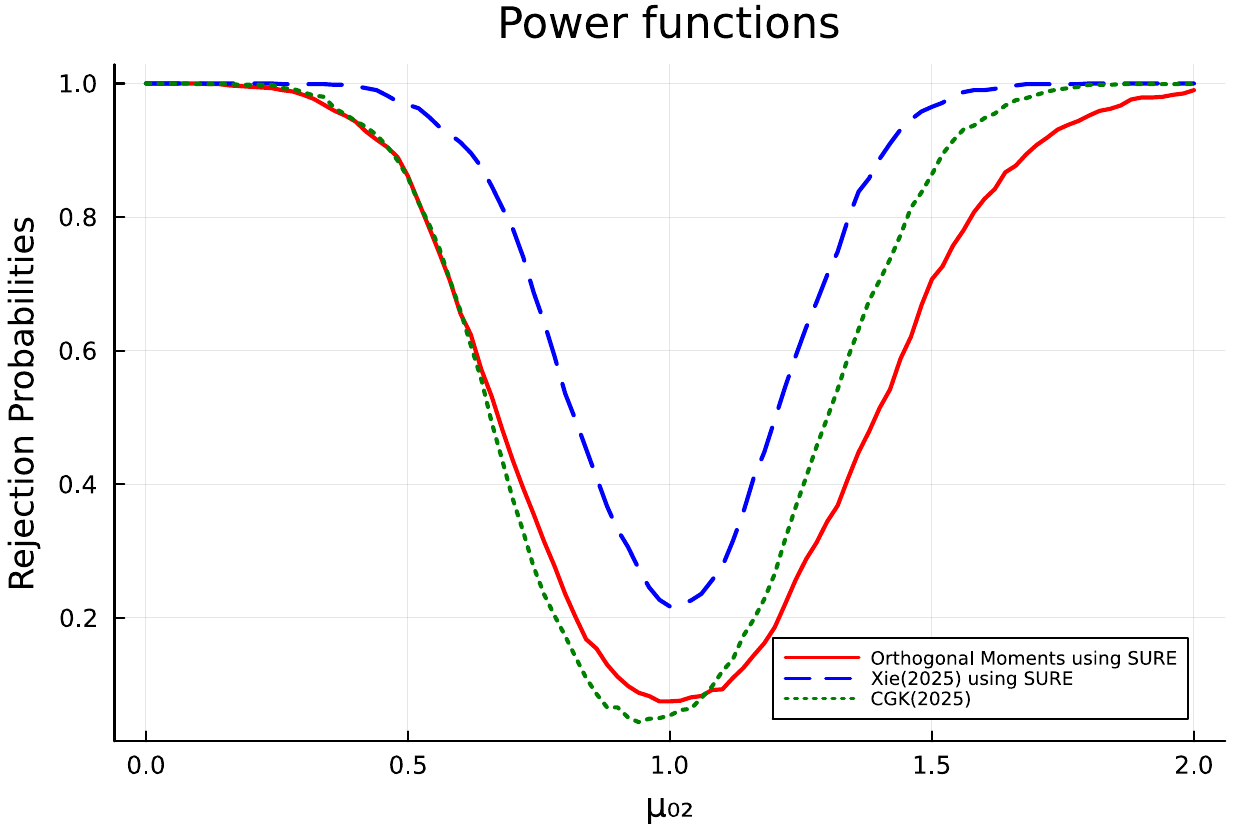}
\caption{Power function for Panel A}
\label{fig:power 100 13}

\vspace{0.2cm}

\begin{minipage}{15cm}
\footnotesize Note: This figure plots the power function of the CGK estimator, the estimator of \cite{xie2025} using SURE and estimator using orthogonal moment conditions and SURE corrections. The other estimators behaves similarly, and are omitted.
\end{minipage}

\end{figure}

Table \ref{tab:rej 2} reports the results with $N=500$, $T=22$ and $\beta_0=0.5$ with different levels of endogeneity represented by $c$. As adaptive elastic net conducts variable selection, the values of $c$ is set to be 5 and 6 respectively to avoid potential local-to-zero coefficients in $a(X_i,\alpha_i,\mu_0)$. As results in Table \ref{tab:rej} shows that the number of folds has little impact on the estimator performance and size control, we report results only with 5-fold cross-fitting in Table \ref{tab:rej 2}. Table \ref{tab:rej 2} shows patterns consistent with Table \ref{tab:rej}. As $c$ increases, the bias of estimators from both \cite{CGK2025} and \cite{xie2025} increases and the tests over-reject the null. In contrast, estimators based on orthogonal moment conditions have better size control and smaller RMSE. Comparing column 6 with column 7 and column 8, we see that using EB or SURE shrinkage still delivers null rejection probabilities closer to the nominal size, albeit with a larger RMSE.

\begin{table}[ht]
\footnotesize
\centering
\begin{tabular}{c c  c c c  || c c c}
\toprule
~ & CGK & Naive & Xie-EB & Xie-SURE & Orth-mean & Orth-EB & Orth-SURE\\
\multicolumn{5}{l}{} & \multicolumn{3}{c}{$L=5$} \\
\hline
\multicolumn{8}{l}{Panel A: $c=0$}\\
Bias & 0.001 & -0.038 & 0.002 & 0.002 & 0.004 & 0.006 & 0.006 \\
Std & 0.076 & 0.072 & 0.076 & 0.076 & 0.086 & 0.088 & 0.088 \\
RMSE & 0.076 & 0.081 & 0.076 & 0.076 & 0.086 & 0.088 & 0.088 \\
Rej Prob (size=0.05) & 0.059 & 0.271 & 0.236 & 0.238 & 0.093 & 0.083 & 0.083 \\
\\
\midrule
\multicolumn{8}{l}{Panel B: $c=5$}\\
Bias & -0.115 & -0.149 & -0.114 & -0.114 & -0.012 & 0.001 & 0.002 \\
Std & 0.088 & 0.084 & 0.088 & 0.088 & 0.098 & 0.102 & 0.102 \\
RMSE & 0.145 & 0.171 & 0.144 & 0.144 & 0.098 & 0.102 & 0.102 \\
Rej Prob (size=0.05) & 0.250 & 0.666 & 0.508 & 0.500 & 0.069 & 0.061 & 0.062 \\
\\
\midrule
\multicolumn{8}{l}{Panel C: $c=6$}\\
Bias & -0.138 & -0.171 & -0.137 & -0.137 & -0.008 & 0.008 & 0.008 \\
Std & 0.093 & 0.089 & 0.093 & 0.093 & 0.105 & 0.109 & 0.110 \\
RMSE & 0.167 & 0.193 & 0.166 & 0.166 & 0.105 & 0.110 & 0.110 \\
Rej Prob (size=0.05) & 0.304 & 0.706 & 0.571 & 0.568 & 0.068 & 0.062 & 0.063 \\
\\
\bottomrule
\end{tabular}

\caption{Comparisons in performance of different estimators and Inference procedures}
\label{tab:rej 2}
\vspace{0.2cm}

\begin{minipage}{15cm}
\footnotesize Note: This table reports the performance of various estimators for $\mu_{02}$. I report the bias, standard error of the estimates $\hat{\mu}_2$, root mean squared error (RMSE) and also the null rejection probability with target size 0.05 across three different DGPs. The sample size is $N=500$, $T=22$, and the auto-regressive coefficient $\beta_0=0.5$. Panel A considers $c=0$ with no correlation between $\bar{\bar{u}}_i$ and $W_i-\alpha_i$, while Panel B and C sets $c=5$ and $c=6$ respectively. CGK corresponds to the measurement error correction estimator suggested in \cite{CGK2025}. Naive, Xie-EB and Xie-SURE uses heteroskedastic robust standard error for $\hat{\alpha}_i$, $\hat{\alpha}_i^{EB}$ and $\hat{\alpha}_i^{SURE}$ respectively. Orth-naive, Orth-EB and Orth-SURE uses the orthogonal moment constructions proposed in this paper but uses $\tilde{\alpha}_{il}$, $\tilde{\alpha}_{il}^{EB}$ and $\tilde{\alpha}_{il}^{SURE}$ respectively. Columns 6-8 report results when the data is split into 5 equal-sized folds. 
\end{minipage}

\end{table}

\section{Empirical Application}\label{application}

We will consider an empirical application that studies the site selection of a policy experiment on agriculture catastrophe insurance in China. \cite{policy} built a comprehensive database for past policy experiments conducted by Chinese central government, and they documented consistent patterns of positive selection of experiment sites. Specifically, the $t$-statistics that compare the average fiscal revenue between experimental sites and non-experimental ones have a significant proportion that exceeds the usual 95\% quantile of the $t$-distribution. The result is robust when they compare agriculture outputs between sites for agricultural related experiments. \cite{policy} also listed potential reasons for positive selection of the experimental sites, but did not commit to a particular explanation of the phenomenon. However, understanding the way that experimental sites are positively selected can be important in interpreting the experimental results and determining whether the new policy should be rolled out to the entire country. 

Policy experiments are typically implemented when there is substantial uncertainty over the effects of the proposed policy, and as a means to encourage local government to find optimal policy instruments for promoting the central government's policy objective \citep{heilmann1}.\footnote{For background and case studies on China's policy experimentation before 2008, see \cite{heilmann1,heilmann2}.}  The central government usually sets the policy agenda, but leaves the particular implementation details to the experimental sites. The successful experimental sites will be designated as ``model sites'', and the less successful ones are encouraged to learn from the experience of the ``model sites''. In such context, positively selecting experimental sites that are better at adjusting policy details and organizing policy reforms can potentially yield valuable experience for non-experimental sites. Therefore, positive selection under this form can be desirable from the view of the central government. 

As an illustration, we focus on a particular policy experiment -- the promotion of agriculture catastrophe insurance (ACI) -- that has taken place in 2017. Prior to the experiment, the insured liability of the agriculture insurance typically only covered direct input costs like fertilizer and seeds, which made the farmers vulnerable to natural disasters. In response to the deficiency of the insurance products, the central government would like to develop insurance that has higher coverage to help farmers guard against risks and encourage uptakes of more advanced agricultural production techniques. One major difference under the new ACI is that it covers land costs in addition to direct input costs.  Under the new ACI, with the same premium, the insured liability increases from 450 CNY per acre to 900 CNY per acre in Shandong Province for instance \citep{insurance}. To facilitate adoption of ACI, the premiums that farmers pay were highly subsidized by both the central and local government.  The focus of the ACI experiment in 2017 was on moderately scaled farmers who were more susceptible to climate risks and also on strategic agriculture products that consist of rice, corn and wheat. The experiment was conducted at county-level, with in total 200 counties participating in the experiment. However, instead of directly choosing counties, the central government hand-picked 13 major grain production provinces and delegated the choice of specific counties to each provincial government with quotas. The guidelines issued by the central government required the experimental counties to have well-established insurance foundations, a sufficient number of moderately scaled farmers and also to be representative in the province.\footnote{See \url{https://jcs.moa.gov.cn/trzgl/201705/t20170527_5626557.htm} } In addition, policy details like the criteria of moderately scaled farmers, the formula for evaluating land costs and means of introducing the insurance to farmers were delegated to local governments. This suggests that selecting a (random) representative sample was not the sole objective of experimental site selection. Furthermore, the central government would also want to learn the best practice of implementation by enrolling counties that were more specialized in agriculture insurance and agriculture production. We would like to test whether the provincial governments were indeed choosing counties that were more specialized in agriculture production as experimental sites. 

\subsection{Data and measurements}

The data comes from the replication package of \cite{policy}. In addition, the selected counties that participated in the ACI experiment are hand-collected from the provincial governments' websites with the exception of the Inner Mongolia province whose information is not directly available. We restricted the sample to counties that are in the remaining 12 major grain producing provinces and dropped counties whose administrative divisions have changed. The total sample consists of 961 counties, among which 168 of them were selected as experiment sites. For each county, we observe its GDP, population, fiscal income, fiscal expenditure, total grain production, agricultural mechanical power, total rural employment and the total number of policy experiment that the county has from 1997-2016. The variable total rural employment is only observed up until 2012. As a preliminary comparison, Table \ref{tab:bt} reports the means of the above measures in 2016 for both experimental and non-experimental counties. The value of rural employment is computed using 2012 data. In addition, we also compared the GDP and fiscal revenue growth rates over the period 2013-2016.   

\begin{table}[ht]
\footnotesize
\centering
\begin{tabular}{p{4cm}| c c c c c}
\toprule
~ & Mean (Experiment=0) & Mean (Experiment=1) & Diff & p-value & Obs\\
\hline
 GDP\newline(100 million) & 225.260 & 312.566 & 87.306 & 0.000 & 961 \\
 Population \newline (10,000) & 55.335 & 88.488 & 33.153 & 0.000 & 961 \\
 Agricultural Mechanical \newline Power (10 million Watts) & 51.256 & 98.084 & 46.829 & 0.000 & 956 \\
 Grain Output \newline (10 thousand tons) & 30.678 & 77.366 & 46.688 & 0.000 & 894 \\
 Fiscal Expense \newline (100 million) & 34.204 & 45.986 & 11.782 & 0.000 & 961 \\
 Fiscal Income \newline (100 million) & 14.938 & 18.694 & 3.756 & 0.022 & 961 \\
 GDP growth \newline & 0.062 & 0.054 & -0.008 & 0.273 & 961 \\
 Fiscal Revenue \newline growth  & 0.077 & 0.061 & -0.016 & 0.188 & 961 \\
 Number of \newline Experiments & 8.851 & 9.607 & 0.756 & 0.071 & 961 \\
 Rural Employment \newline (10,000) & 24.857 & 39.855 & 14.998 & 0.000 & 960 \\
\\
\bottomrule
\end{tabular}

\caption{Balance Table comparing experiment sites and non-experiment sites}
\label{tab:bt}

\vspace{0.2cm}

\begin{minipage}{15cm}
\footnotesize Note: The table reports the mean of the measures for both experimental and non-experimental counties, along with the p-value associated with the t statistics. GDP, Population, Agricultural Mechanical Power, Grain Output, Fiscal Expenditure, Fiscal Income and number of ongoing policy experiments are measured in 2016, while Rural Employment is measured in 2012. In addition, Fiscal revenue growth rate and GDP growth rate are computed from the period between 2013-2016.
\end{minipage}
\end{table}

Compared with the non-experimental counties, the experimental ones tend to be larger in terms of population, GDP and fiscal expense. In addition, they also produced more grains and used more agricultural machinery. This serves as preliminary evidence that when choosing counties, the provincial governments not only aimed at obtaining a representative sample, but also wanted to enroll counties that were more developed and counties that had a larger share of agricultural industry. These counties can potentially lay out a better implementation plan of ACI, and thus establish ``model practices'' for other non-experimental ones. We now turn to more formal econometric analysis of the problem. 

\subsection{Econometric Specification}

As the outcome is binary, we will specify a simple logit model:
\begin{align}
\text{Experiment}_i = \mathds{1}(& \mu_{01} \alpha_i+ \mu_{02}\text{GDP\_per\_cap}_i+ \mu_{03} \text{GDP\_growth}_i\nonumber \\
&+\mu_{04} \text{FR\_growth}_i +\mu_{05} \text{N\_policy}_i+\delta_p \geq v_i ),
\label{eq:logit}
\end{align}
where GDP per Capita and number of policy experiments are measured in 2016, and GDP\_growth and FR\_growth are GDP and fiscal revenue growth rate between 2013 and 2016 for county $i$. $\delta_p$ is province fixed effect that captures heterogeneity across provinces, and $v_i$'s are iid logit error. We also consider a version that omits the province fixed effects. in which case it will be replaced with an intercept. The key variable is $\alpha_i$ which measures the degree of comparative advantage of county $i$ in grain production. If the provincial governments were selecting experimental counties that were more specialized in grain production, then we should expect the coefficient $\mu_{01}$ to be positive. However, the degree of specialization is not directly observable to the researcher, and we propose a panel data model that can help us recover $\alpha_i$'s.

\begin{align}
\text{grain\_per\_emp}_{it} = &\beta_{01} \text{GDP\_per\_cap}_{it} + \beta_{02} \text{Mech\_per\_emp}_{it}\nonumber \\
 & +\beta_{03} \text{FE\_per\_cap}_{it} +\alpha_i + \delta_{pt} +u_{it},
\label{eq:panel grain}
\end{align}
where we regress grain output per rural employment for each county $i$ and year $t$ on that county-year pair's GDP per capita, mechanical power per rural employment and fiscal expenditure per capita, along with county fixed effects $\alpha_i$ and province by time fixed effects $\delta_{pt}$. A larger value of $\alpha_i$ indicates that county $i$ is more efficient in grain production compared with its fellow counties that have similar fiscal support, income level and amount of agricultural machinery. Therefore, $\alpha_i$ captures the level of specialization of grain production for different counties. 

Although the presence of $\delta_{pt}$ may cause some difficulties for directly implementing the procedure in Section \ref{procedure}. In practice, we first transform (\ref{eq:panel grain}) by demeaning the province by time average of each variable to eliminate the province by time fixed effects.
\begin{align*}
 \dot{\text{grain\_per\_emp}}_{it} = & \beta_{01} \dot{\text{GDP\_per\_cap}}_{it} + \beta_{02} \dot{\text{Mech\_per\_emp}}_{it} \nonumber\\
& +\beta_{03} \dot{\text{FE\_per\_cap}}_{it} +\dot{\alpha}_i  + \dot{u}_{it},
\end{align*}
where $\dot{A_{it}}=A_{it}-\bar{A}_{p(i)t}$ is the demeaned version of $A_{it}$ with $\bar{A}_{p(i)t}$ being the mean of $A_{it}$ across province $p$ to which $i$ belongs and time $t$. The $X_{it}$ in (\ref{eq:1st stage}) would correspond to $(\dot{\text{GDP\_per\_cap}}_{it},\dot{\text{Mech\_per\_emp}}_{it},\dot{\text{FE\_per\_cap}}_{it})$. Applying procedure in Section \ref{procedure} then identifies $\dot{\alpha}_i$. As $\alpha_i$ is time invariant, $\dot{\alpha}_i=\alpha_i-\bar{\alpha}_{p(i)}$. Note that $\dot{\alpha}_i$ and $\alpha_i$ differ by a constant that only varies by province. Therefore, rewriting the logit model in (\ref{eq:logit}), we have
\begin{align*}
\text{Experiment}_i = \mathds{1}(&\mu_{01} \dot{\alpha}_i+ \mu_{02}\text{GDP\_per\_cap}_i+ \mu_{03} \text{GDP\_growth}_i \\
&+\mu_{04} \text{FR\_growth}_i +\mu_{05} \text{N\_policy}_i+(\delta_p+\mu_{01} \bar{\alpha}_p) \geq v_i ).
\end{align*}
Therefore, $\mu_0$ can still be identified from the demeaned version of $\alpha_i$ as long as the province fixed effects are included in the cross-sectional logit model. 

To fit into the framework of (\ref{eq:mom mu}), we will work with the moment conditions that are associated with the score of the likelihood function define by (\ref{eq:logit}). With some abuse of notation define $W_i =(\text{GDP\_per\_cap}_i, \text{GDP\_growth}_i,\text{FR\_growth}_i,\text{N\_policy}_i,\delta_p)'$ where $\delta_p$'s are understood as 12 dummy variables. Let $\mu_{0,-}=(\mu_{02},\mu_{03},\mu_{04},\mu_{05},\mu_{0,p})'$ where $\mu_{0,p}$ are fixed effects associated with the 12 provinces. Then we are solving the estimating equations
\begin{equation}
\mathbb{E} \begin{bmatrix}  (\text{Experiment}_i- \Lambda(W_i'\mu_{0,-}+\alpha_i \mu_{01}) ) \alpha_i \\ (\text{Experiment}_i- \Lambda(W_i'\mu_{0,-}+\alpha_i \mu_{01}) ) W_i \end{bmatrix} =\begin{bmatrix}0  \\ \mathbf{0}  \end{bmatrix},
\label{eq:mom logit}
\end{equation}
where $\Lambda(x)=\frac{\exp(x)}{1+\exp(x)}$ is the logistic function. With the moment condition \\
$\mathbb{E}[m(W_i,\text{Experiment}_i,\alpha_i,\mu_0)]=0$ defined as in (\ref{eq:mom logit}), the corresponding $a(X_i,\alpha_i,\mu_0)$ for $X_i=\{X_{it}\}_{t=1997}^{2012} \cup \{\text{grain\_per\_emp}_{it}\}_{t=1997}^{2011} $ is
\begin{equation}
a(X_i,\alpha_i,\mu_0)=\mathbb{E} \left( \begin{bmatrix} 
-\mu_{01} \Lambda_i (1-\Lambda_i) \alpha_i + \text{Experiment}_i- \Lambda_i \\
-\mu_{01} \Lambda_i (1-\Lambda_i) W_i   \end{bmatrix} | X_i,\alpha_i\right),
\label{eq:a logit}
\end{equation} 
where $\Lambda_i= \Lambda(W_i'\mu_{0,-}+\alpha_i \mu_{0,1})$. In practice, we estimate $a(X_i,\alpha_i,\mu_0)$ by performing adaptive elastic net regression on $X_i$ and $\alpha_i$. 

Estimation of the Jacobian matrix of $\mu_0$ is similar to the classical maximum likelihood theory:
\begin{equation}
\hat{G}= \frac{1}{N} \sum_{l=1}^L \sum_{i\in I_l} \Lambda_i (1-\Lambda_i) (\tilde{\alpha}_{il},W_i)'(\tilde{\alpha}_{il}, W_i)
\end{equation}

As in the simulation of Section \ref{simulation}, the estimates of $\mu_0$ and the asymptotic variance (\ref{eq:vhat}) are averaged across 20 different sample splits with $L=5$. 

\subsection{Results}

Before presenting estimates of $\mu_0$, it is useful to show the results of the panel regression (\ref{eq:panel grain}) first.

\begin{table}[ht]
\centering
\begin{tabular}{ c c c c}
\toprule
~ & (1) & (2) & (3)\\
\hline
cons & 0.595*** &  &  \\
 & (0.196) &  &  \\
GDP per cap & -0.15 & -0.041 & -0.016 \\
 & (0.132) & (0.04) & (0.018) \\
Mech per emp & 0.611** & 0.678** & 0.676** \\
 & (0.305) & (0.33) & (0.314) \\
FE per cap & 0.265 & -0.423 & -0.395 \\
 & (0.786) & (0.264) & (0.234) \\
\\
\hline
County Fixed Effects & No & Yes & Yes \\
Province-by-time Fixed Effects & No & No & Yes\\
\bottomrule
\end{tabular}

\caption{Panel regression results for (\ref{eq:panel grain})}
\label{tab:panel results}

\vspace{0.2cm}

\begin{minipage}{15cm}
\footnotesize Note: The table reports estimates of $\beta_0$ for (\ref{eq:panel grain}). Specification (1) is a simple linear regression. Specification (2) only includes county fixed effects, while specification (3) is the preferred specification which includes both county fixed effects and province-by-time fixed effects.
\end{minipage}

\end{table}

We see that the coefficients on both GDP per capita and fiscal expenditure per capita change significantly when county fixed effects are included, while inclusion of province-by-time fixed effects have more limited impacts on the magnitude and sign of all coefficients. The mechanical power per rural employment turns out to be a salient predictor of grain production per rural employment, while the coefficients on GDP per capita and fiscal expenditure per capita are not significantly different from zero. Now we turn to results on $\mu_0$.

\begin{table}[htbp]
\centering
\begin{tabular}{l c c c c}
\toprule
\multicolumn{5}{l}{Panel A: Estimates of $\mu_0$ without county fixed effects}\\
~ & (1) & (2) & (3) & (4) \\
~ & Plug-in & Orth-mean & Orth-EB & Orth-SURE\\
cons & -1.826 & -1.829 & -1.884 & -1.855 \\
 & (0.209) & (0.203) & (0.212) & (0.208) \\
$\alpha_i$ & 0.615*** & 0.399*** & 0.47*** & 0.41*** \\
 & (0.118) & (0.123) & (0.165) & (0.138) \\
GDP per Capita & -0.016 & -0.019 & -0.014 & -0.019 \\
 & (0.03) & (0.028) & (0.03) & (0.029) \\
GDP growth & -0.333 & -0.318 & -0.045 & -0.144 \\
 & (1.394) & (1.406) & (1.487) & (1.457) \\
FR growth & -0.611 & -0.62 & -0.786 & -0.685 \\
 & (0.729) & (0.654) & (0.716) & (0.679) \\
N policy & 0.036** & 0.039** & 0.038** & 0.038** \\
 & (0.018) & (0.018) & (0.019) & (0.018) \\
\\
\hline
County Fixed Effects & No & No & No & No\\
\midrule
\multicolumn{5}{l}{Panel B: Estimates of $\mu_0$ with county fixed effects}\\
~ & (1) & (2) & (3) & (4) \\
~ & Plug-in & Orth-mean & Orth-EB & Orth-SURE\\
$\alpha_i$ & 0.562*** & 0.362*** & 0.467*** & 0.378*** \\
 & (0.12) & (0.122) & (0.152) & (0.13) \\
GDP per Capita & -0.049 & -0.06 & -0.052 & -0.064 \\
 & (0.036) & (0.037) & (0.04) & (0.039) \\
GDP growth & 2.479 & 2.047 & 2.342 & 2.208 \\
 & (2.114) & (2.107) & (2.246) & (2.175) \\
FR growth & 0.175 & 0.209 & 0.134 & 0.247 \\
 & (0.776) & (0.576) & (0.633) & (0.57) \\
N policy & 0.034* & 0.031 & 0.032 & 0.031 \\
 & (0.02) & (0.019) & (0.021) & (0.02) \\
\\
\hline
County Fixed Effects & Yes & Yes & Yes & Yes\\
\bottomrule
\end{tabular}
\caption{Estimates of $\mu_0$ in (\ref{eq:logit})}
\label{tab:logit results}

\vspace{0.2cm}

\begin{minipage}{15cm}
\footnotesize Note: The table reports estimates of $\mu_0$ for (\ref{eq:logit}). Panel A estimates the version without the province fixed effects, while panel B includes the province fixed effects, which is the preferred specification. Across columns, we report estimates of $\mu_0$ using the plug-in $\hat{\alpha}_i$ from specification (3) of Table \ref{tab:panel results}, and also estimates that employs orthogonal moments proposed in this paper with and without shrinkage corrections.
\end{minipage}

\end{table}
 
Panel A of Table \ref{tab:logit results} report the results of estimating (\ref{eq:logit}) without province fixed effects, while panel B report the ones with province fixed effects. Comparing results across panels, the coefficients on GDP growth rate and fiscal revenue growth rate switch signs. Even though the coefficients associated with the two variables are not statistically significant, it demonstrates the importance of including province fixed effects in (\ref{eq:logit}) when (\ref{eq:panel grain}) is estimated with province-by-time fixed effects. Now we turn to the comparison between the plug-in estimates and the estimates based on orthogonal moments. The coefficients are similar in sign and magnitude when we compare the plug-in estimates with estimates based on orthogonal moments except the coefficient $\mu_{01}$ on $\alpha_i$. In both panels, $\mu_{01}$ is estimated to be larger (around 0.6) when we use the simple plug-in methods, while the estimates based on orthogonal moments give estimates that range from 0.36 to 0.47. Therefore, correcting errors in estimating $\alpha_i$ could be important in certain applications. Nevertheless, in our setting, the use of orthogonal moments still gives consistent evidence that there is indeed positive selection of experimental counties in terms of their comparative advantage in grain production. The provincial governments were leaning towards more specialized counties to learn from their potentially better implementation practice and make the ongoing experiment more successful.

\section{Conclusion}

This paper employs debiased machine learning techniques to study inference of parameters characterized by cross-sectional moments involving latent variables that are fixed effects in an auxiliary panel data regression. Under the asymptotic regime that both the cross-sectional dimension and time dimension go to infinity, the asymptotic variance of the estimator based on orthogonal moments is the same as if the true latent variables are known to the researcher. In addition, shrinkage methods based on Empirical Bayes or Stein's Unbiased Risk Estimation can be readily incorporated for estimating the latent variables. The approach enables researchers to deal with non-linearity in the cross-sectional moments and relaxes independence assumption between the panel data residuals and the cross-sectional moment functions.  Simulation results show that using orthogonal moments have better null rejection probabilities compared with existing methods when the panel data residuals are correlated with the cross-sectional moment functions. In an empirical application, we examine the experimental site selection of the agricultural catastrophe insurance experiment in China. Interpreting the score function of a logit model as the targeted moments, we show that specialization in grain production is an important determinant on whether a county is selected into the policy experiment where the level of specialization is captured as county fixed effects in an auxiliary linear panel data model. The estimates using orthogonal moments are smaller in magnitude compared with the simple plug-in approach, while SURE and EB corrections deliver similar results. In general, using orthogonal moments can be important for valid inference especially when researchers are in doubt that the residuals in the panel data model are uncorrelated with the cross-sectional moment functions.

\bibliography{fixedeffects}

\appendix
\counterwithin{lemma}{section} 
\renewcommand{\thelemma}{\thesection.\arabic{lemma}} 

\counterwithin{assumption}{section} 
\renewcommand{\theassumption}{\thesection.\arabic{assumption}} 

\counterwithin{remark}{section} 
\renewcommand{\theremark}{\thesection.\arabic{remark}} 

\section{Proof of results in section \ref{as theory}}\label{ap:orth}

The proof of proposition \ref{prop:omega} and the following lemmas are going to use the following result repeatedly,
\begin{lemma}\label{lm:2square}
For random variables $A_i$ and $B_i$ that satisfies $\frac{1}{N} \sum_{i=1}^N A_i^2=o_p(1)$ and $\frac{1}{N} \sum_{i=1}^N B_i^2=o_p(1)$, then
\[\frac{1}{N} \sum_{i=1}^N (A_i+B_i)^2=o_p(1)\]

\end{lemma}

\begin{proof}
Elementary inequality has $(A_i+B_i)^2\leq 2(A_i^2+B_i^2)$, hence
\[0\leq \frac{1}{N} \sum_{i=1}^N (A_i+B_i)^2\leq \frac{2}{N} \sum_{i=1}^N A_i^2 +B_i^2 =o_p(1)\]

\end{proof}

Throughout the proof, $\tilde{\alpha}_{il}$ will be used to represent an estimator of $\alpha_i$ which can be based on time average, EB or SURE corrections. For cases where distinguishing the type of estimators are important, $\tilde{\alpha}_{il}^*$ will be used to denote estimators that incorporate the shrinkage idea in EB or SURE.

\begin{lemma}\label{lm:no rrhat}
Under assumption \ref{ass:cross}-\ref{as:mom}, and we have
\begin{equation}
\frac{1}{\sqrt{N}} \sum_{i\in I_l} (\tilde{a}_{il}-a(X_i,\alpha_i,\mu_0))(Y_{iT}-X_{iT}'\tilde{\beta}_l-\tilde{\alpha}_{il}) = o_p(1).
\label{eq:no rrhat}
\end{equation}

If further, assumption \ref{as:mom omega} holds, then
\begin{equation}
\frac{1}{N} \sum_{i\in I_l} (\tilde{a}_{il}-a(X_i,\alpha_i,\mu_0))^2 (Y_{iT}-X_{iT}'\tilde{\beta}_l-\tilde{\alpha}_{il})^2 =o_p(1).
\label{eq:omega no rrhat}
\end{equation}

The conclusion also hold when $\tilde{\alpha}_{il}$'s are replaced with their EB or SURE shrinkage versions $\tilde{\alpha}_{il}^*$.
\end{lemma}

\begin{remark}
As the number of folds $L$ is held fixed as sample size grows, the lemma will also imply
\[\frac{1}{\sqrt{N}} \sum_{l=1}^L \sum_{i\in I_l} (\tilde{a}_{il}-a(X_i,\alpha_i,\mu_0))(Y_{iT}-X_{iT}'\tilde{\beta}_l-\tilde{\alpha}_{il}) = o_p(1)\]
In particular, under lemma \ref{lm:no rrhat} we will be able to replace $\tilde{a}_{il}$ in (\ref{eq:orth sample mom}) with the true $a(X_i,\alpha_i,\mu_0)$ without changes in the asymptotic distribution:
\begin{align*}
& \frac{1}{\sqrt{N}} \sum_{l=1}^L \sum_{i\in I_l} \tilde{a}_{il}(Y_{iT}-X_{iT}'\tilde{\beta}_l-\tilde{\alpha}_{il})\\
&=\frac{1}{\sqrt{N}} \sum_{l=1}^L \sum_{i\in I_l} a(X_i,\alpha_i,\mu_0)(Y_{iT}-X_{iT}'\tilde{\beta}_l-\tilde{\alpha}_{il}) +o_p(1)
\end{align*}
\end{remark}

\begin{proof}
Notice that 
\begin{align*}
&Y_{iT}-X_{iT}'\tilde{\beta}_l-\tilde{\alpha}_{il}\\
&=\underbrace{Y_{iT} -X_{iT}'\beta-\alpha_i}_{u_{iT}} +X_{iT}'(\beta-\tilde{\beta}_l)+(\alpha_i-\tilde{\alpha}_{il})
\end{align*}

Hence, we can expand the LHS of (\ref{eq:no rrhat}) into the three terms:
\begin{align*}
& \frac{1}{\sqrt{N}} \sum_{i\in I_l} (\tilde{a}_{il}-a(X_i,\alpha_i,\mu_0))(Y_{iT}-X_{iT}'\tilde{\beta}_l-\tilde{\alpha}_{il})\\
&=\underbrace{\frac{1}{\sqrt{N}} \sum_{i\in I_l} (\tilde{a}_{il}-a(X_i,\alpha_i,\mu_0)) u_{iT}}_{I} \\
&+\underbrace{\frac{1}{\sqrt{N}} \sum_{i\in I_l} (\tilde{a}_{il}-a(X_i,\alpha_i,\mu_0)) X_{iT}'(\beta-\tilde{\beta}_l)}_{II}\\
&+\underbrace{\frac{1}{\sqrt{N}} \sum_{i\in I_l} (\tilde{a}_{il}-a(X_i,\alpha_i,\mu_0))(\alpha_i-\tilde{\alpha}_{il})}_{III}
\end{align*}

We will deal with the three terms separately and show that all of them are $o_p(1)$. For all three terms, notice that $\tilde{a}_{il}$ and $\tilde{\beta}_l$ are estimated using samples not in $I_l$, hence $\tilde{a}_{il}$ is a deterministic function of $X_i,\bar{\bar{Y}}_i$ and hence $\tilde{a}_{il}-a(X_i,\alpha_i,\mu_0)$ are iid distributed conditional on data not in fold $I_l$. Let $\mathcal{W}_l^C$ denote the data in the complement of $I_l$. 

For (I), we have \footnote{Recall that $\bar{\bar{Y}}_i=\frac{1}{T-1}\sum_{t=1}^{T-1} Y_{it}$ which is pre-determined at period $T$}
\begin{align*}
& \mathbb{E}(\frac{1}{\sqrt{N}} \sum_{i\in I_l} (\tilde{a}_{il}-a(X_i,\alpha_i,\mu_0)) u_{iT}|\mathcal{W}_l^C)\\
&=\frac{|I_l|}{\sqrt{N}} \mathbb{E}[\mathbb{E}((\tilde{a}_{il}-a(X_i,\alpha_i,\mu_0)) u_{iT}|\mathcal{W}_l^C,X_i,\bar{\bar{Y}}_i),\alpha_i|\mathcal{W}_l^C] \\
&=\frac{|I_l|}{\sqrt{N}} \mathbb{E}[(\tilde{a}_{il}-a(X_i,\alpha_i,\mu_0))\mathbb{E}( u_{iT}|\mathcal{W}_l^C,X_i, \bar{\bar{Y}}_i,\alpha_i)|\mathcal{W}_l^C] \\
&=0,
\end{align*}
hence by law of iterated expectation $\mathbb{E}\left( \frac{1}{\sqrt{N}} \sum_{i\in I_l} (\tilde{a}_{il}-a(X_i,\alpha_i,\mu_0)) u_{iT}\right)=0$,
and 
\begin{align*}
& \mathbb{E}( \{\frac{1}{\sqrt{N}} \sum_{i\in I_l} (\tilde{a}_{il}-a(X_i,\alpha_i,\mu_0)) u_{iT}\}^2|\mathcal{W}_l^C)\\
&=\frac{|I_l|}{N} \mathbb{E}((\tilde{a}_{il}-a(X_i,\alpha_i,\mu_0))^2 u_{iT}^2|\mathcal{W}_l^C)\\
&=\frac{|I_l|}{N} \mathbb{E}[\mathbb{E}(\tilde{a}_{il}-a(X_i,\alpha_i,\mu_0))^2 u_{iT}^2|\mathcal{W}_l^C,X_i,\bar{\bar{Y}}_i,\alpha_i)|\mathcal{W}_l^C]\\
&=\frac{|I_l|}{N} \mathbb{E}[\mathbb{E}(u_{iT}^2|X_i,\bar{\bar{Y}}_i,\alpha_i) \mathbb{E}(\tilde{a}_{il}-a(X_i,\alpha_i,\mu_0))^2|\mathcal{W}_l^C,X_i,\bar{\bar{Y}}_i,\alpha_i)|\mathcal{W}_l^C]\\
&\leq \frac{C |I_l|}{N} \mathbb{E}((\tilde{a}_{il}-a(X_i,\alpha_i,\mu_0))^2|\mathcal{W}_l^C)\\
&= o_p(1)
\end{align*}
where the last line follows from condition (ii) of assumption \ref{as:rr}. I also used condition (i) of assumption \ref{as:mom} to bound $\mathbb{E}(u_{iT}^2|X_i,\bar{\bar{Y}}_i,\alpha_i)$. 

For (II), recall that conditional on $\mathcal{W}_l^C$, $\tilde{\beta}_{l}$ is a constant. In addition, $\tilde{\beta}_l-\beta = O_p(1/\sqrt{N})$, so it suffices to show that 

$\frac{1}{N} \sum_{i\in I_l} (\tilde{a}_{il}-a(X_i,\alpha_i,\mu_0)) X_{iT} =o_p(1)$

Notice that 
\begin{align*}
& \frac{1}{N} \sum_{i\in I_l} (\tilde{a}_{il}-a(X_i,\alpha_i,\mu_0)) X_{iT}\\
&\leq \sqrt{\frac{1}{N} \sum_{i\in I_l} (\tilde{a}_{il}-a(X_i,\alpha_i,\mu_0))^2} \sqrt{\frac{1}{N} \sum_{i\in I_l} \|X_{iT}\|^2}\\
&= o_p(1) O_p(1)
\end{align*}
here the first term under the square root is $o_p(1)$ by condition (iv) of assumption \ref{as:rr} and the second term is $O_p(1)$ via law of large numbers and condition (ii) of assumption \ref{as:mom}.

For (III), we can apply Cauchy-Schwartz to get for some constant $C$ depending only on $L$
\begin{align*}
& \frac{1}{\sqrt{N}} \sum_{i\in I_l} (\tilde{a}_{il}-a(X_i,\alpha_i,\mu_0))(\alpha_i-\tilde{\alpha}_{il})\\
& \leq C \sqrt{N} \sqrt{\frac{1}{|I_l|}\sum_{i\in I_l} (\tilde{a}_{il}-a(X_i,\alpha_i,\mu_0))^2} \sqrt{\frac{1}{|I_l|} \sum_{i\in I_l}(\alpha_i-\tilde{\alpha}_{il})^2}\\
& = C \sqrt{N} o_p(N^{-\zeta/2}) O_p(T^{-1/2})\\
& =o_p(1),
\end{align*}
where the second last line uses condition (ii) of assumption \ref{as:alpha} and condition (iv) of assumption \ref{as:rr}.

Now we will show (\ref{eq:omega no rrhat}) by using lemma \ref{lm:2square} applied to the terms (I), (II) and (III).
For (I), we need to show that $\frac{1}{N} \sum_{i\in I_l} (\tilde{a}_{il}-a(X_i,\alpha_i,\mu_0))^2 u_{iT}^2=o_p(1)$. Each term inside the summation is iid conditional on $\mathcal{W}_l^C$. In addition,
\begin{align*}
& \mathbb{E}((\tilde{a}_{il}-a(X_i,\alpha_i,\mu_0))^2 u_{iT}^2|\mathcal{W}_l^C)\\
&=\mathbb{E}[(\tilde{a}_{il}-a(X_i,\alpha_i,\mu_0))^2 \mathbb{E}(u_{iT}^2|X_i,\bar{\bar{Y}}_i,\alpha_i)|\mathcal{W}_l^C]\\
&\leq  C \mathbb{E}((\tilde{a}_{il}-a(X_i,\alpha_i,\mu_0))^2|\mathcal{W}_l^C)\\
&=o_p(1)
\end{align*} 
Conditional convergence (conditioning on $\mathcal{W}_l^C$) implies unconditional convergence by lemma 6.1 of \cite{dml}.

For (II), we need $\|\beta-\tilde{\beta}_l\|^2 \frac{1}{N} \sum_{i\in I_l} (\tilde{a}_{il}-a(X_i,\alpha_i,\mu_0))^2 \|X_{iT}\|^2=o_p(1)$ . As $\|\beta-\tilde{\beta}_l\|^2 =O_p(\frac{1}{N})$, we only need $\frac{1}{N} \sum_{i\in I_l} (\tilde{a}_{il}-a(X_i,\alpha_i,\mu_0))^2 \|X_{iT}\|^2=O_p(1)$. This is guaranteed by condition (iii) of assumption \ref{as:mom}:

\begin{align*}
& \mathbb{E}((\tilde{a}_{il}-a(X_i,\alpha_i,\mu_0))^2 \|X_{iT}\|^2)\\
&\leq 2 \mathbb{E}(a(X_i,\alpha_i,\mu_0)^2 \|X_{iT}\|^2 )+\mathbb{E}(\tilde{a}_{il}^2 \|X_{iT}\|^2 )\\
&<\infty
\end{align*}

For (III), we need $1/N \sum_{i\in I_l} (\tilde{a}_{il}-a(X_i,\alpha_i,\mu_0))^2 (\alpha_i-\tilde{\alpha}_{il})^2 =o_p(1)$. Using Cauchy-Schwartz,
\begin{align*}
& \frac{1}{N} \sum_{i\in I_l} (\tilde{a}_{il}-a(X_i,\alpha_i,\mu_0))^2 (\alpha_i-\tilde{\alpha}_{il})^2\\
&\leq \frac{1}{N} \sum_{i\in I_l} (\tilde{a}_{il}-a(X_i,\alpha_i,\mu_0))^2 \sup_{i\in I_l} (\alpha_i-\tilde{\alpha}_{il})^2\\
&\leq N (\frac{1}{N} \sum_{i\in I_l} (\tilde{a}_{il}-a(X_i,\alpha_i,\mu_0))^2 \frac{1}{N} \sum_{i\in I_l} (\alpha_i-\tilde{\alpha}_{il})^2)\\
&= N o_p(N^{-\zeta}) O_p(T^{-1})\\
&=o_p(1)
\end{align*}
which holds by assumption \ref{as:mom omega}. 

Therefore, (\ref{eq:omega no rrhat}) holds as well.

\end{proof}

\begin{lemma}\label{lm:no betahat}
Under assumption \ref{as:rr}, and assumption \ref{as:mom}, we have
\begin{equation}
\frac{1}{\sqrt{N}} \sum_{i\in I_l} a(X_i,\alpha_i,\mu_0) X_{iT}'(\tilde{\beta}_l-\beta) = o_p(1) 
\label{eq:no betahat}
\end{equation}

In addition,
\begin{equation}
\frac{1}{N} \sum_{i\in I_l} a(X_i,\alpha_i,\mu_0)^2 \|X_{iT}\|^2 \|\tilde{\beta}_l-\beta\|^2=o_p(1)
\label{eq:omega no betahat}
\end{equation}

\end{lemma}

\begin{remark}
Lemma \ref{lm:no betahat} combined with lemma \ref{lm:no rrhat} implies:
\begin{align*}
& \frac{1}{\sqrt{N}}\sum_{l=1}^L \sum_{i\in I_l} \tilde{a}_{il}(Y_{iT}-X_{iT}'\tilde{\beta}_l-\tilde{\alpha}_{il})\\
&=\frac{1}{\sqrt{N}} \sum_{l=1}^L \sum_{i\in I_l} a(X_i,\alpha_i,\mu_0)(Y_{iT}-X_{iT}'\tilde{\beta}_l-\tilde{\alpha}_{il}) +o_p(1)\\
&=\frac{1}{\sqrt{N}} \sum_{l=1}^L \sum_{i\in I_l} a(X_i,\alpha_i,\mu_0)(Y_{iT}-X_{iT}'\beta-\tilde{\alpha}_{il}) +o_p(1)
\end{align*}

\end{remark}

\begin{proof}

Let $\mathcal{W}_l^C$ be the data that are in the complement of fold $l$. We first establish (\ref{eq:no betahat}). Condition (i) of assumption \ref{as:rr} implies that $\tilde{\beta}_l-\beta=O_p(1/\sqrt{NT})$, so it suffices to show that 

$\frac{1}{N} \sum_{l=1}^L \sum_{i\in I_l} |a(X_i,\alpha_i,\mu_0)| \|X_{iT}\| =O_p(1)$. 

Condition (iii) of assumption \ref{as:mom} implies that $\mathbb{E}(|a(X_i,\alpha_i,\mu_0)|\|X_{iT}\|)\leq \infty$. As the samples are iid across individuals, one can apply standard law of large numbers for  

$\frac{1}{N} \sum_{l=1}^L \sum_{i\in I_l} |a(X_i,\alpha_i,\mu_0)|\|X_{iT}\|$ to conclude

\[\frac{1}{N} \sum_{l=1}^L \sum_{i\in I_l} |a(X_i,\alpha_i,\mu_0)|\|X_{iT}\| \xrightarrow{p} \mathbb{E}(|a(X_i,\alpha_i,\mu_0)|\|X_{iT}\|)<\infty\]

Therefore, $\frac{1}{N} \sum_{l=1}^L \sum_{i\in I_l} |a(X_i,\alpha_i,\mu_0)|\|X_{iT}\|=O_p(1)$

To prove (\ref{eq:omega no betahat}), condition (iii) of assumption \ref{as:mom} implies $\frac{1}{N} \sum_{i\in I_l} a(X_i,\alpha_i,\mu_0)^2 \|X_{iT}\|^2=O_p(1)$. Along with the fact that $\|\tilde{\beta}_l-\beta\|^2=O_p(1/(NT))$, we have
\[\frac{1}{N} \sum_{i\in I_l} a(X_i,\alpha_i,\mu_0)^2 \|X_{iT}\|^2 \|\tilde{\beta}_l-\beta\|^2=o_p(1)\]

\end{proof}

\begin{lemma}\label{lm:no alphahat}

Assume that assumption \ref{ass:cross}, assumption \ref{as:rr}, assumption \ref{as:alpha} and assumption \ref{as:mom} hold, then
\begin{align}
& \frac{1}{\sqrt{N}} \sum_{i\in I_l} [m(W_i,\tilde{\alpha}_{il},\mu_0)+a(X_i,\alpha_i,\mu_0)(Y_{iT}-X_{iT}'\beta_0-\tilde{\alpha}_{il})] \nonumber\\
&=\frac{1}{\sqrt{N}} \sum_{i\in I_l} [m(W_i,\alpha_i,\mu_0)+a(X_i,\alpha_i,\mu_0)(Y_{iT}-X_{iT}'\beta_0-\alpha_i)] +o_p(1)
\label{eq:no alphahat}
\end{align}

In addition, if one imposes further assumption \ref{as:mom omega},
\begin{equation}
\frac{1}{N} \sum_{i\in I_l} [m(W_i,\tilde{\alpha}_{il},\mu_0)-m(W_i,\alpha_i,\mu_0) + a(X_i,\alpha_i,\mu_0)(\alpha_i-\tilde{\alpha}_{il}))]^2 =o_p(1)
\label{eq:omega no alphahat}
\end{equation}

\end{lemma}

\begin{proof}

We first note that 
\begin{align*}
& \frac{1}{\sqrt{N}} \sum_{i\in I_l} [m(W_i,\tilde{\alpha}_{il},\mu_0)+a(X_i,\alpha_i,\mu_0)(Y_{iT}-X_{iT}'\beta_0-\tilde{\alpha}_{il})]\\
& = \frac{1}{\sqrt{N}} \sum_{i\in I_l} [m(W_i,\alpha_i,\mu_0)+a(W_i,\alpha_i,\mu_0)(Y_{iT}-X_{iT}'\beta_0-\alpha_i)]\\
&+  \frac{1}{\sqrt{N}} \sum_{i\in I_l} [m(W_i,\tilde{\alpha}_{il},\mu_0)-m(W_i,\alpha_i,\mu_0) + a(X_i,\alpha_i,\mu_0)(\alpha_i-\tilde{\alpha}_{il})]
\end{align*}

Therefore, it suffices to show that the term in the last line is $o_p(1)$. 

Using the assumption that $m(W_i,\alpha_i,\mu_0)$ is twice continuously differentiable in $\alpha_i$, a standard second-order Taylor expansion around $\alpha_i$ gives 
\begin{align*}
& \frac{1}{\sqrt{N}} \sum_{i\in I_l} [m(W_i,\tilde{\alpha}_{il},\mu_0)-m(W_i,\alpha_i,\mu_0)]\\
& = \frac{1}{\sqrt{N}} \sum_{i\in I_l} \frac{\partial m(W_i,\alpha_i,\mu_0) }{\partial \alpha_i} (\tilde{\alpha}_{il}-\alpha_i)\\
&+ \frac{1}{\sqrt{N}} \sum_{i\in I_l} \frac{\partial^2 m(W_i,\check{\alpha}_i,\mu_0) }{\partial \alpha_i^2} (\tilde{\alpha}_{il}-\alpha_i)^2,
\end{align*}
for some $\check{\alpha}_i$ that lies between $\alpha_i$ and $\tilde{\alpha}_{il}$. 

Combining the above display with $\frac{1}{\sqrt{N}} \sum_{i\in I_l} a(W_i,\alpha_i,\mu_0)(\alpha_i-\tilde{\alpha}_{il})$, we have

\begin{align*}
& \frac{1}{\sqrt{N}} \sum_{i\in I_l} [m(W_i,\tilde{\alpha}_{il},\mu_0)-m(W_i,\alpha_i,\mu_0) + a(X_i,\alpha_i,\mu_0)(\alpha_i-\tilde{\alpha}_{il})]\\
& =\underbrace{ \frac{1}{\sqrt{N}} \sum_{i\in I_l} (\frac{\partial m(W_i,\alpha_i,\mu_0) }{\partial \alpha_i}- a(X_i,\alpha_i,\mu_0))(\tilde{\alpha}_{il}-\alpha_i)}_{I}\\
&+\underbrace{\frac{1}{\sqrt{N}} \sum_{i\in I_l} \frac{\partial^2 m(W_i,\check{\alpha}_i,\mu_0) }{\partial \alpha_i^2} (\tilde{\alpha}_{il}-\alpha_i)^2 }_{II}
\end{align*}

For (I), I will first prove the result under the assumption that $\tilde{\alpha}_{il}$ is constructed as the time averages: $\bar{\bar{Y}}_i-\bar{\bar{X}}_i'\tilde{\beta}_{l}$. Under this construction $(\frac{\partial m(W_i,\alpha_i,\mu_0) }{\partial \alpha_i}- a(X_i,\alpha_i,\mu_0))(\tilde{\alpha}_{il}-\alpha_i)$ is independent of $(\frac{\partial m(W_j,\alpha_j,\mu_0) }{\partial \alpha_j}- a(X_j,\alpha_j,\mu_0))(\tilde{\alpha}_{jl}-\alpha_j)$ for $i\neq j$. Note that $\tilde{\alpha}_{il}-\alpha_i = \bar{\bar{u}}_i+\bar{\bar{X}}_i'(\beta-\tilde{\beta}_l)$, we can further decompose (I) into two terms:

\begin{align*}
& \frac{1}{\sqrt{N}} \sum_{i\in I_l}  (\frac{\partial m(W_i,\alpha_i,\mu_0) }{\partial \alpha_i}- a(X_i,\alpha_i,\mu_0))(\tilde{\alpha}_{il}-\alpha_i)\\
&=\frac{1}{\sqrt{N}}\sum_{i\in I_l} (\frac{\partial m(W_i,\alpha_i,\mu_0) }{\partial \alpha_i}- a(X_i,\alpha_i,\mu_0)) \bar{\bar{u}}_i\\
&+\frac{1}{\sqrt{N}}\sum_{i\in I_l}  (\frac{\partial m(W_i,\alpha_i,\mu_0) }{\partial \alpha_i}- a(X_i,\alpha_i,\mu_0)) \bar{\bar{X}}_i'(\beta-\tilde{\beta}_l)
\end{align*}

The first term can be shown to be $o_p(1)$ by mean square convergence 
\begin{align*}
& \mathbb{E}(\{1/\sqrt{N} \sum_{i\in I_l}  (\frac{\partial m(W_i,\alpha_i,\mu_0) }{\partial \alpha_i}- a(X_i,\alpha_i,\mu_0))\bar{\bar{u}}_i\}^2|\mathcal{W}_l^C)\\
&=\frac{1}{N} \sum_{i\in I_l} \mathbb{E}((\frac{\partial m(W_i,\alpha_i,\mu_0) }{\partial \alpha_i}- a(X_i,\alpha_i,\mu_0))^2\bar{\bar{u}}_i^2|\mathcal{W}_l^C)\\
&+\frac{1}{N} \frac{|I_l|(|I_l|-1)}{2} \mathbb{E}((\frac{\partial m(W_i,\alpha_i,\mu_0) }{\partial \alpha_i}- a(X_i,\alpha_i,\mu_0))\bar{\bar{u}}_i|\mathcal{W}_l^C)^2\\
&= \mathbb{E}((\frac{\partial m(W_i,\alpha_i,\mu_0) }{\partial \alpha_i}- a(X_i,\alpha_i,\mu_0))^2\bar{\bar{u}}_i^2)\\
&+ \frac{|I_l|(|I_l|-1)}{2 N}  \mathbb{E}((\frac{\partial m(W_i,\alpha_i,\mu_0) }{\partial \alpha_i}- a(X_i,\alpha_i,\mu_0))\bar{\bar{u}}_i)^2\\
&= o_p(1)+o_p(1)=o_p(1)
\end{align*}
where condition (ii) of assumption \ref{as:alpha} ensures the two terms are $o_p(1)$.

The  the second term can be shown to be $o_p(1)$:
\begin{align*}
& \frac{1}{\sqrt{N}}\sum_{i\in I_l}  (\frac{\partial m(W_i,\alpha_i,\mu_0) }{\partial \alpha_i}- a(X_i,\alpha_i,\mu_0)) \bar{\bar{X}}_i'(\beta-\tilde{\beta}_l)\\
&\leq  \sqrt{N} \frac{1}{N}\sum_{i\in I_l}  (\frac{\partial m(W_i,\alpha_i,\mu_0) }{\partial \alpha_i}- a(X_i,\alpha_i,\mu_0)) \bar{\bar{X}}_i' O_p(1/\sqrt{N})\\
&=o_p(1) 
\end{align*}
where $(\frac{\partial m(W_i,\alpha_i,\mu_0) }{\partial \alpha_i}- a(X_i,\alpha_i,\mu_0)) \bar{X}_i'$ is iid and mean zero by law of iterated expectation and the definition of $a(X_i,\alpha_i,\mu_0)=\mathbb{E}(\frac{\partial m(W_i,\alpha_i,\mu_0) }{\partial \alpha_i}|X_i,\alpha_i)$. 

For (II), using condition (i) of assumption \ref{as:mom} that $\sup_{\alpha,W_i} \frac{\partial^2 m(W_i,\check{\alpha}_i,\mu_0) }{\partial \alpha_i^2}=C_m <\infty$, we have
\begin{align*}
& \frac{1}{\sqrt{N}} \sum_{i\in I_l} \frac{\partial^2 m(W_i,\check{\alpha}_i,\mu_0) }{\partial \alpha_i^2} (\tilde{\alpha}_{il}-\alpha_i)^2\\
&\leq C_m \frac{1}{\sqrt{N}} \sum_{i\in I_l} (\tilde{\alpha}_{il}-\alpha_i)^2\\
& = O_p(\sqrt{N}/T)=o_p(1)
\end{align*}

where the last line follows from assumption \ref{as:alpha} and $\lim \frac{\sqrt{N}}{T} =0$. 

To prove (\ref{eq:omega no alphahat}), we apply lemma \ref{lm:2square} to the term (I) and (II). For (I), we need to show that $\frac{1}{N} \sum_{i\in I_l} (\frac{\partial m(W_i,\alpha_i,\mu_0) }{\partial \alpha_i}- a(X_i,\alpha_i,\mu_0))^2(\tilde{\alpha}_{il}-\alpha_i)^2 =o_p(1)$:
\begin{align*}
& \frac{1}{N} \sum_{i\in I_l} (\frac{\partial m(W_i,\alpha_i,\mu_0) }{\partial \alpha_i}- a(X_i,\alpha_i,\mu_0))^2(\tilde{\alpha}_{il}-\alpha_i)^2\\
&\leq \sqrt{\frac{1}{N} \sum_{i\in I_l} (\frac{\partial m(W_i,\alpha_i,\mu_0) }{\partial \alpha_i}- a(X_i,\alpha_i,\mu_0))^4} \sqrt{\frac{1}{N} \sum_{i\in I_l}(\tilde{\alpha}_{il}-\alpha_i)^4 }\\
&=O_p(1) o_p(1)
\end{align*}
where the first term is $O_p(1)$ by condition (iii) of assumption \ref{as:mom}, and the second term is $o_p(1)$ by condition (i) of assumption \ref{as:alpha}. 

For (II), we need to show $1/N \sum_{i\in I_l} (\frac{\partial^2 m(W_i,\check{\alpha}_i,\mu_0) }{\partial \alpha_i^2})^2 (\tilde{\alpha}_{il}-\alpha_i)^4 =o_p(1)$. Using condition (i) of assumption \ref{as:mom}, we have
\begin{align*}
& \frac{1}{N} \sum_{i\in I_l} (\frac{\partial^2 m(W_i,\check{\alpha}_i,\mu_0) }{\partial \alpha_i^2})^2 (\tilde{\alpha}_{il}-\alpha_i)^4\\
&\leq C_m^2 \frac{1}{N} \sum_{i\in I_l} (\tilde{\alpha}_{il}-\alpha_i)^4\\
& = o_p(1)
\end{align*}
where the last line is by condition (i) of assumption \ref{as:alpha}.

Therefore, (\ref{eq:omega no alphahat}) holds.

\textbf{$\tilde{\alpha}_{il}^*$ based on Empirical Bayes ideas}

The key challenge for this case is that $\tilde{\alpha}_{il}^*$ is constructed based on all the data in fold $l$, not just the data for observation $i$. Therefore, one has to account for the potential correlation introduced by the shrinkage. Note that either  $\tilde{\alpha}_{il}^{EB}$ or $\tilde{\alpha}_{il}^{SURE}$ takes the form: $(1-s_{in}) \tilde{\alpha}_{il}+s_{in} \hat{\bar{\alpha}}_{l}$ where 
\[\hat{\bar{\alpha}}_l = \frac{1}{|I_l|} \sum_{i\in I_l}  \tilde{\alpha}_{il}=\bar{\bar{Y}}-\bar{\bar{X}}'\tilde{\beta}_l\]
for $\bar{\bar{Y}}=\frac{1}{(T-1) |I_l|} \sum_{i\in I_l} \sum_{t=1}^{T-1} Y_{it}$ and $\bar{\bar{X}}=\frac{1}{(T-1) |I_l|} \sum_{i\in I_l} \sum_{t=1}^{T-1} X_{it}$

 In particular, the numerator of $s_{in}$ is an estimator of $u_i^2= \mbox{Var} (\frac{1}{T-1} \sum_{t=1}^{T-1} u_{it})$ which converges to zero, while the denominator is strictly positive given condition (iii) of assumption \ref{as:alpha}. Therefore, $s_{in}\rightarrow 0$ and asymptotically, both $\tilde{\alpha}_{il}^{EB}$ and $\tilde{\alpha}_{il}^{SURE}$ will behave like $\tilde{\alpha}_{il}$. The strategy is then to show $\frac{1}{\sqrt{N}} \sum_{i\in I_l} (\frac{\partial m(W_i,\alpha_i,\mu_0) }{\partial \alpha_i}- a(X_i,\alpha_i,\mu_0))(\tilde{\alpha}_{il}-\tilde{\alpha}_{il}^*) =o_p(1)$ for $\tilde{\alpha}_{il}^*$ where $\tilde{\alpha}_{il}^*$ equals either $\tilde{\alpha}_{il}^{EB}$ or $\tilde{\alpha}_{il}^{SURE}$. Abbreviate $a(X_i,\alpha_i,\mu_0)$ as $a_i$ we have

\begin{align*}
& \frac{1}{\sqrt{N}} \sum_{i\in I_l} (\frac{\partial m(W_i,\alpha_i,\mu_0) }{\partial \alpha_i}- a_i)(\tilde{\alpha}_{il}-\tilde{\alpha}_{il}^*)\\
& = \frac{1}{\sqrt{N}} \sum_{i\in I_l} (\frac{\partial m(W_i,\alpha_i,\mu_0) }{\partial \alpha_i}- a_i) s_{in}(\tilde{\alpha}_{il}-\hat{\bar{\alpha}}_l)\\
&= \frac{1}{\sqrt{N}} \sum_{i\in I_l} s_{in} (\frac{\partial m(W_i,\alpha_i,\mu_0) }{\partial \alpha_i}- a_i)(\bar{\bar{Y}}_i-\bar{\bar{X}}_i'\tilde{\beta}_l-\bar{\bar{Y}}-\bar{\bar{X}}'\tilde{\beta}_l)\\
&\leq \sup_{i\in I_l} s_{in}  \frac{1}{\sqrt{N}} \sum_{i\in I_l}  |(\frac{\partial m(W_i,\alpha_i,\mu_0) }{\partial \alpha_i}- a_i)(\bar{\bar{Y}_i}-\bar{\bar{X}}_i'\tilde{\beta}_l-\bar{\bar{Y}}-\bar{\bar{X}}'\tilde{\beta}_l)|\\
& = \underbrace{\sup_{i\in I_l} s_{in}  \frac{1}{\sqrt{N}} \sum_{i\in I_l}  |(\frac{\partial m(W_i,\alpha_i,\mu_0) }{\partial \alpha_i}- a_i)(\alpha_i+\bar{\bar{u}}_i-\mathbb{E}(\alpha))|}_{I'}\\
&+\underbrace{\sup_{i\in I_l} s_{in}  \frac{1}{\sqrt{N}} \sum_{i\in I_l} |(\frac{\partial m(W_i,\alpha_i,\mu_0) }{\partial \alpha_i}- a_i)\bar{\bar{X}}_i' (\beta-\tilde{\beta}_l)|}_{II'}\\
&+\underbrace{ \sup_{i\in I_l} s_{in}  \bar{\bar{X}}' \frac{1}{\sqrt{N}} \sum_{i\in I_l} |(\frac{\partial m(W_i,\alpha_i,\mu_0) }{\partial \alpha_i}- a_i) (\beta-\tilde{\beta}_l)|}_{III'}\\
&+ \underbrace{\sup_{i\in I_l} s_{in}  \frac{1}{\sqrt{N}} \sum_{i\in I_l} |(\frac{\partial m(W_i,\alpha_i,\mu_0) }{\partial \alpha_i}- a_i) \bar{\bar{u}}|}_{IV'}\\
&+ \underbrace{\sup_{i\in I_l} s_{in}  \frac{1}{\sqrt{N}} \sum_{i\in I_l} |(\frac{\partial m(W_i,\alpha_i,\mu_0) }{\partial \alpha_i}- a_i) (\bar{\alpha}_l-\mathbb{E}(\alpha))|}_{V'}
\end{align*}
where the $\bar{\alpha}_l$ in the last term denotes $\frac{1}{|I_l|} \sum_{i\in I_l} \alpha_i$. 

We will show that each term is $o_p(1)$ given the condition that $\sup_{i\in I_l} s_{in}=o_p(N^{-1/2})$:
\begin{align*}
I'&=\sup_{i\in I_l} s_{in}  \frac{1}{\sqrt{N}} \sum_{i\in I_l}  |(\frac{\partial m(W_i,\alpha_i,\mu_0) }{\partial \alpha_i}- a(X_i))(\alpha_i+\bar{\bar{u}}_i-\mathbb{E}(\alpha))|\\
&=o_p(1) \frac{1}{N} \sum_{i\in I_l} |(\frac{\partial m(W_i,\alpha_i,\mu_0) }{\partial \alpha_i}- a(X_i))(\alpha_i+\bar{\bar{u}}_i-\mathbb{E}(\alpha))|\\
& = o_p(1),
\end{align*}
given $\mathbb{E}(|(\frac{\partial m(W_i,\alpha_i,\mu_0) }{\partial \alpha_i}- a(X_i))(\alpha_i+\bar{\bar{u}}_i-\mathbb{E}(\alpha))|)<\infty$ as it is a sum of iid random variables.

\begin{align*}
II'&= \sup_{i\in I_l} s_{in}  \frac{1}{\sqrt{N}} \sum_{i\in I_l} |(\frac{\partial m(W_i,\alpha_i,\mu_0) }{\partial \alpha_i}- a(X_i))\bar{\bar{X}}_i' (\beta-\tilde{\beta}_l)|\\
&= o_p(1) \|\beta-\tilde{\beta}_l\| \frac{1}{N} \sum_{i\in I_l} \|(\frac{\partial m(W_i,\alpha_i,\mu_0) }{\partial \alpha_i}- a(X_i))\bar{\bar{X}}_i'\|\\
&= o_p(1) o_p(\frac{1}{\sqrt{N}})
\end{align*}
given $\mathbb{E}(\|(\frac{\partial m(W_i,\alpha_i,\mu_0) }{\partial \alpha_i}- a(X_i))\bar{\bar{X}}_i'\|)<\infty$. (Notice that the term inside the norm has mean zero) 

\begin{align*}
III'&=\sup_{i\in I_l} s_{in}  \bar{\bar{X}}' \frac{1}{\sqrt{N}} \sum_{i\in I_l} |(\frac{\partial m(W_i,\alpha_i,\mu_0) }{\partial \alpha_i}- a(X_i)) (\beta-\tilde{\beta}_l)|\\
&=o_p(1) \|\beta-\tilde{\beta}_l\| \bar{\bar{X}}' \frac{1}{N} \sum_{i\in I_l} |(\frac{\partial m(W_i,\alpha_i,\mu_0) }{\partial \alpha_i}- a(X_i))|\\
&=o_p(1/\sqrt{NT})
\end{align*}
given $\mathbb{E}(\|(\frac{\partial m(W_i,\alpha_i,\mu_0) }{\partial \alpha_i}- a(X_i,\alpha_i,\mu_0))\|)<\infty$.

\begin{align*}
IV'&= \sup_{i\in I_l} s_{in}  \frac{1}{\sqrt{N}} \sum_{i\in I_l} |(\frac{\partial m(W_i,\alpha_i,\mu_0) }{\partial \alpha_i}- a(X_i)) \bar{\bar{u}}|\\
&= o_p(1) |\bar{\bar{u}}| \frac{1}{N} \sum_{i\in I_l} |(\frac{\partial m(W_i,\alpha_i,\mu_0) }{\partial \alpha_i}- a(X_i))|\\
&= o_p(1) o_p(1/\sqrt{NT}) O_p(1)=o_p(1)
\end{align*}

\begin{align*}
V'&= \sup_{i\in I_l} s_{in}  \frac{1}{\sqrt{N}} \sum_{i\in I_l} |(\frac{\partial m(W_i,\alpha_i,\mu_0) }{\partial \alpha_i}- a(X_i)) (\bar{\alpha}_l-\mathbb{E}(\alpha))|\\
&=o_p(1) | (\bar{\alpha}_l-\mathbb{E}(\alpha))|\frac{1}{N} \sum_{i\in I_l} |(\frac{\partial m(W_i,\alpha_i,\mu_0) }{\partial \alpha_i}- a(X_i))|\\
&= o_p(1) O_p(1)=o_p(1)
\end{align*}

Therefore, as long as $\sup_{i\in I_l} s_{in}=o_p(N^{-1/2})$, using $\tilde{\alpha}_{il}^*$ as an estimate of $\alpha_i$ does not affect the rate of convergence of $I$.

\textbf{Conditions that guarantee $\sup_{i\in I_l} s_{in}=o_p(N^{-1/2})$}

Recall that $s_{in}=\frac{\hat{u}_i^2}{\hat{u}_i^2+ \hat{\sigma}_{\alpha}^2}$. Under condition (iii) of assumption \ref{as:alpha}, as $\hat{\sigma}_{\alpha}^2\xrightarrow{p} \mbox{Var}(\alpha_i)>0$, the denominator will be strictly positive for $\hat{\sigma}_{\alpha}^2$ that equals either the Empirical Bayes version or the SURE version. Therefore, showing that $\sup_{i\in I_l} s_{in}=o_p(N^{-1/2})$ is equivalent to showing that $\sup_{i\in I_l} \hat{u}_i^2 =o_p(N^{-1/2})$

Suppose $\mathbb{E}(T \hat{u}_i^2)^{2+\delta}<\infty$ for some $\delta>0$, then using union bound and Markov inequality we will have
\begin{align*}
& \mathbb{P}(\sup_{i\in I_l} \hat{u}_i^2 >\epsilon)\\
& \leq |I_l| \mathbb{P}(T \hat{u}_i^2 >T \epsilon)\\
& \leq |I_l| \frac{\mathbb{E}[(T \hat{u}_i^2)^{2+\delta}]}{T^{2+\delta} \epsilon^{2+\delta}}
\end{align*}

The above display can be converted to the high probability statement by defining $\iota=|I_l| \frac{\mathbb{E}[(T \hat{u}_i^2)^{2+\delta}]}{T^{2+\delta} \epsilon^{2+\delta}}$: with probability at least $1-\iota$ :
\[\sup_{i\in I_l} \hat{u}_i^2\leq  \frac{(|I_l| \mathbb{E}[(T \hat{u}_i^2)^{2+\delta}] \iota^{-1})^{1/(2+\delta)}}{T}\]
which gives the result 
\[\sup_{i\in I_l} \hat{u}_i^2 =O_p(\frac{|I_l|^{1/(2+\delta)}}{T})\]

If $\frac{|I_l|^{1/(2+\delta)}}{T}=o(N^{-1/2})$, then we have $\sup_{i\in I_l} s_{in}=o_p(N^{-1/2})$. The condition holds if $T \propto N^{g}$ and $g>1/2+1/(2+\delta)$.

\end{proof}

\textbf{Proof of proposition \ref{prop:omega}}

The proposition is a simple consequence of lemma \ref{lm:no rrhat}, lemma \ref{lm:no betahat} and lemma \ref{lm:no alphahat}:

\begin{align*}
& \frac{1}{\sqrt{N}} \sum_{l=1}^L \sum_{i\in I_l} m(W_i,\tilde{\alpha}_{il},\mu_0)+\tilde{a}_{il} (Y_{iT}-X_{iT}' \tilde{\beta}_l -\tilde{\alpha}_{il})\\
& = \frac{1}{\sqrt{N}} \sum_{l=1}^L \sum_{i\in I_l} m(W_i,\tilde{\alpha}_{il},\mu_0)+ a(X_i,\alpha_i,\mu_0) (Y_{iT}-X_{iT}' \tilde{\beta}_l -\tilde{\alpha}_{il}) +o_p(1)\\
&=\frac{1}{\sqrt{N}} \sum_{l=1}^L \sum_{i\in I_l} m(W_i,\tilde{\alpha}_{il},\mu_0)+ a(X_i,\alpha_i,\mu_0) (Y_{iT}-X_{iT}' \beta -\tilde{\alpha}_{il}) +o_p(1)\\
&= \frac{1}{\sqrt{N}} \sum_{i=1}^N m(W_i,\alpha_i,\mu_0)+ a(X_i,\alpha_i,\mu_0)(Y_{iT}-X_{iT}'\beta-\alpha_i)
\end{align*},

where the first equality is lemma \ref{lm:no rrhat}, the second equality is lemma \ref{lm:no betahat} and the last one is lemma \ref{lm:no alphahat}.

\hfill \qedsymbol

\textbf{Proof of proposition \ref{prop:omegahat}}

Recall the definition of $\Omega$:
\[\Omega\equiv \mathbb{E}(m^*(W_i,\alpha_i,Z_i,\beta_0,\mu_0,a) m^*(W_i,\alpha_i,Z_i,\beta_0,\mu_0,a)')\]
which implies that an infeasible consistent estimator is by law of large numbers:
\[\frac{1}{N} \sum_{i=1}^N m^*(W_i,\alpha_i,Z_i,\beta_0,\mu_0,a) m^*(W_i,\alpha_i,Z_i,\beta_0,\mu_0,a)'\]

We will show that 
\begin{align*}
& \frac{1}{N} \sum_{l=1}^L \sum_{i\in I_l} (m(W_i,\tilde{\alpha}_{il},\tilde{\mu}_l)+\hat{\psi}_{il})(m(W_i,\tilde{\alpha}_{il},\tilde{\beta}_l,\tilde{\mu}_l)+\hat{\psi}_{il})'\\
& = \frac{1}{N} \sum_{i=1}^N m^*(W_i,\alpha_i,Z_i,\beta_0,\mu_0,a) m^*(W_i,\alpha_i,Z_i,\beta_0,\mu_0,a)' +o_p(1)
\end{align*}
which will imply (\ref{eq:con omega}).

Using an argument in Lemma E1 of \cite{cher2022}, it suffices to show that 
\begin{equation}
\frac{1}{N} \sum_{i=1}^N \|m(W_i,\tilde{\alpha}_{il},\tilde{\mu}_l)+\hat{\psi}_{il}-m^*(W_i,\alpha_i,Z_i,\beta_0,\mu_0,a)\|^2 =o_p(1)
\label{eq:omega diff}
\end{equation}
as it implies
\begin{align*}
& \|\frac{1}{N} \sum_{l=1}^L \sum_{i\in I_l} (m(W_i,\tilde{\alpha}_{il},\tilde{\mu}_l)+\hat{\psi}_{il})(m(W_i,\tilde{\alpha}_{il},\tilde{\mu}_l)+\hat{\psi}_{il})'\\
&-\frac{1}{N} \sum_{i=1}^N m^*(W_i,\alpha_i,Z_i,\beta_0,\mu_0,a) m^*(W_i,\alpha_i,Z_i,\beta_0,\mu_0,a)'\|\\
&\leq \sum_{l=1}^L 1/N \sum_{i\in I_l} (\|m(W_i,\tilde{\alpha}_{il},\tilde{\mu}_l)+\hat{\psi}_{il}-m^*(W_i,\alpha_i,Z_i,\beta_0,\mu_0,a)\|^2 \\
& +2 \|m^*(W_i,\alpha_i,Z_i,\beta_0,\mu_0,a)\| \|m(W_i,\tilde{\alpha}_{il},\tilde{\mu}_l)+\hat{\psi}_{il}-m^*(W_i,\alpha_i,Z_i,\beta_0,\mu_0,a)\| ) \\
&\leq o_p(1) + 2 \sum_{l=1}^L \sqrt{\frac{1}{N} \sum_{i\in I_l}\|m(W_i,\tilde{\alpha}_{il},\tilde{\mu}_l)+\hat{\psi}_{il}-m^*(W_i,\alpha_i,Z_i,\beta_0,\mu_0,a)\|^2 }\\
& \sqrt{\frac{1}{N} \sum_{i\in I_l} \|m^*(W_i,\alpha_i,Z_i,\beta_0,\mu_0,a)\|^2}\\
& =o_p(1)(1 +O_p(1))=o_p(1) 
\end{align*}

Therefore, the task is to show (\ref{eq:omega diff}). In particular, we have the following decomposition:
\begin{align*}
& \frac{1}{N} \sum_{l=1}^L \sum_{i\in I_l} m(W_i,\tilde{\alpha}_{il},\tilde{\mu}_l)+\hat{\psi}_{il}-m^*(W_i,\alpha_i,Z_i,\beta_0,\mu_0,a)\\
& = \underbrace{\frac{1}{N} \sum_{l=1}^L \sum_{i\in I_l}  m(W_i,\tilde{\alpha}_{il},\tilde{\mu}_l)-m(W_i,\tilde{\alpha}_{il},\mu_0)}_{I}\\
&+ \underbrace{\frac{1}{N} \sum_{l=1}^L \sum_{i\in I_l}  (\tilde{a}_{il}-a(X_i,\alpha_i,\mu_0))(Y_{iT}-X_{iT}'\tilde{\beta}_l-\tilde{\alpha}_{il})}_{II}\\
& + \underbrace{\frac{1}{N} \sum_{l=1}^L \sum_{i\in I_l} a(X_i,\alpha_i,\mu_0) X_{iT}'(\tilde{\beta}_l-\beta)}_{III}\\
&+ \underbrace{\frac{1}{N} \sum_{l=1}^L \sum_{i\in I_l} m(W_i,\tilde{\alpha}_{il},\mu_0)-m(W_i,\alpha_i,\mu_0) + a(X_i,\alpha_i,\mu_0)(\alpha_i-\tilde{\alpha}_{il})}_{IV}
\end{align*}

We then apply lemma \ref{lm:2square} to the above decomposition, so that we need to show for the four terms if we replace the term inside the summation sign with its square, it is still $o_p(1)$.

(\ref{eq:omega no rrhat}) of lemma \ref{lm:no rrhat}  handles (II), (\ref{eq:omega no betahat} of lemma \ref{lm:no betahat} handles (III) and (\ref{eq:omega no alphahat}) of lemma \ref{lm:no alphahat} handles (IV). Therefore, it suffices to show that 
\[\frac{1}{N}  \sum_{i\in I_l}  [m(W_i,\tilde{\alpha}_{il},\tilde{\mu}_l)-m(W_i,\tilde{\alpha}_{il},\mu_0)]\\^2 =o_p(1)\]

We will use lemma \ref{lm:2square} again using the following decomposition:

\begin{align*}
& \frac{1}{N}  \sum_{i\in I_l}  [m(W_i,\tilde{\alpha}_{il},\tilde{\mu}_l)-m(W_i,\tilde{\alpha}_{il},\mu_0)]\\
&=\frac{1}{N} \sum_{i\in I_l} \frac{\partial m(W_i,\tilde{\alpha}_{il},\check{\mu}_l)}{\partial \mu'} (\tilde{\mu}_l-\mu_0)\\
&= \frac{1}{N} \sum_{i\in I_l} \frac{\partial^2 m(W_i,\alpha_i,\check{\mu}_l)}{\partial \alpha_i \partial \mu'} (\tilde{\alpha}_{il}-\alpha_i) (\tilde{\mu}_l-\mu_0)\\
&+\frac{1}{N} \sum_{i\in I_l} \frac{\partial^3 m(W_i,\check{\alpha}_i,\check{\mu}_l)}{\partial \alpha_i^2 \partial \mu'} (\tilde{\alpha}_{il}-\alpha_i)^2 (\tilde{\mu}_l-\mu_0)
\end{align*}
where $\check{\mu}_l$ is some vector that lies on the line segment between $\tilde{\mu}_l$ and $\mu_0$, and $\check{\alpha}_i$ is some scalar lying between $\tilde{\alpha}_{il}$ and $\alpha_i$.

Apply lemma \ref{lm:2square} to the above decomposition, we have
\begin{align*}
& \frac{1}{N} \sum_{i\in I_l} [\frac{\partial^2 m(W_i,\alpha_i,\check{\mu}_l)}{\partial \alpha_i \partial \mu'} (\tilde{\alpha}_{il}-\alpha_i) (\tilde{\mu}_l-\mu_0)]^2\\
&\leq \|\tilde{\mu}_l-\mu_0\|^2 \frac{1}{N} \sum_{i\in I_l} \|\frac{\partial^2 m(W_i,\alpha_i,\check{\mu}_l)}{\partial \alpha_i \partial \mu'}\|^2 (\tilde{\alpha}_{il}-\alpha_i)^2\\
&\leq \|\tilde{\mu}_l-\mu_0\|^2 \sqrt{\frac{1}{N} \sum_{i\in I_l} \|\frac{\partial^2 m(W_i,\alpha_i,\check{\mu}_l)}{\partial \alpha_i \partial \mu'}\|^4} \sqrt{\frac{1}{N} \sum_{i\in I_l} (\tilde{\alpha}_{il}-\alpha_i)^4 }\\
&=o_p(1) O_p(1) o_p(1)=o_p(1)
\end{align*}
by condition (v) of assumption \ref{as:mom omega} and condition (i) of assumption \ref{as:alpha}.

\begin{align*}
& \frac{1}{N} \sum_{i\in I_l} [\frac{\partial^3 m(W_i,\check{\alpha}_i,\check{\mu}_l)}{\partial \alpha_i^2 \partial \mu'} (\tilde{\alpha}_{il}-\alpha_i)^2 (\tilde{\mu}_l-\mu_0)]^2\\
&\leq \|\tilde{\mu}_l-\mu_0\|^2  C_m^2 \frac{1}{N} \sum_{i\in I_l} (\tilde{\alpha}_{il}-\alpha_i)^4\\
&=o_p(1)
\end{align*}
by condition (i) of assumption \ref{as:alpha} and condition (iv) of assumption \ref{as:mom omega}.

Therefore,
\[1/N  \sum_{i\in I_l}  [m(W_i,\tilde{\alpha}_{il},\tilde{\mu}_l)-m(W_i,\tilde{\alpha}_{il},\mu_0)]\\^2 =o_p(1)\]
 and the proposition is proved. 

\hfill \qedsymbol

\textbf{Proof of proposition \ref{prop:G}}

We first show that the difference in using $\alpha_i$ and using $\tilde{\alpha}_{il}$ in constructing $G$ is asymptotically negligible:
\[\frac{1}{N} \sum_{l=1}^L \sum_{i\in I_l}\frac{\partial m(W_i,\tilde{\alpha}_{il},\hat{\mu})}{\partial \mu'}-\frac{1}{N} \sum_{i=1}^N \frac{\partial m(W_i,\alpha_i,\hat{\mu})}{\partial \mu'}=o_p(1)\]

\begin{align*}
& \|\frac{1}{N} \sum_{l=1}^L \sum_{i\in I_l}\frac{\partial m(W_i,\tilde{\alpha}_{il},\hat{\mu})}{\partial \mu'}-\frac{1}{N} \sum_{i=1}^N \frac{\partial m(W_i,\alpha_i,\hat{\mu})}{\partial \mu'}\|\\
&=\|\frac{1}{N} \sum_{l=1}^L \sum_{i\in I_l} \frac{\partial^2 m(W_i,\check{\alpha}_i,\hat{\mu})}{\partial \alpha_i \partial \mu'} (\tilde{\alpha}_{il}-\alpha_i)\|\\
&\leq \sum_{l=1}^L \sqrt{\frac{1}{N} \sum_{i\in I_l} \|\frac{\partial^2 m(W_i,\check{\alpha}_i,\hat{\mu})}{\partial \alpha_i \partial \mu'}\|^2 } \sqrt{\frac{1}{N} \sum_{i\in I_l} (\tilde{\alpha}_{il}-\alpha_i)^2 }\\
&= O_p(1) o_p(1)=o_p(1)
\end{align*}
where $\check{\alpha}_i$ is some scalar lying between $\tilde{\alpha}_{il}$ and $\alpha_i$. Assumption \ref{as:G} ensures that the first term under the square root is bounded in probability and assumption \ref{as:alpha} ensures that the second square root term is $o_p(1)$.

Next we show that
\[1/N \sum_{i=1}^N \frac{\partial m(W_i,\alpha_i,\hat{\mu})}{\partial \mu}-1/N \sum_{i=1}^N \frac{\partial m(W_i,\alpha_i,\mu_0)}{\partial \mu}=o_p(1)\]

Note that applying standard law of large number yields $\frac{1}{N} \sum_{i=1}^N \frac{\partial m(W_i,\alpha_i,\mu_0)}{\partial \mu}\xrightarrow{p} G$. Hence, showing the above display is sufficient to conclude our proof. 

\begin{align*}
& \| \frac{1}{N} \sum_{i=1}^N  \frac{\partial m(W_i,\alpha_i,\hat{\mu})}{\partial \mu'}-\frac{1}{N} \sum_{i=1}^N \frac{\partial m(W_i,\alpha_i,\mu_0)}{\partial \mu'}\|\\
&\leq \frac{1}{N} \sum_{i=1}^N F_{\alpha}(W_i,\alpha_i) \|\hat{\mu}-\mu_0\|^{1/C}\\
&=O_p(1) o_p(1)=o_p(1)
\end{align*}
where the first average is $O_p(1)$ by assumption \ref{as:G} and the second term is $o_p(1)$ by the consistency of $\hat{\mu}$.

\hfill \qedsymbol

\textbf{Proof of theorem \ref{thm1}}

The theorem is proved by doing a Taylor expansion on the objective function as in standard GMM settings. Recall that 
\[\hat{\mu} = \arg_{\mu} \min \hat{m}^*(\mu)' \Upsilon_N \hat{m}^*(\mu)\]
so it satisfies the approximate first order condition:
\[\frac{\partial \hat{m}^*(\hat{\mu})}{\partial \mu'}' \Upsilon_N \hat{m}^*(\hat{\mu})=o_p(\frac{1}{\sqrt{N}}) \]

Now a Taylor expansion on $\hat{m}^*(\hat{\mu})$ yields that for some $\check{\mu}$ that lies on the line segment between $\mu_0$ and $\hat{\mu}$, we have

\[\frac{\partial \hat{m}^*(\hat{\mu})}{\partial \mu'}' \Upsilon_N  (\hat{m}^*(\mu_0)+ \frac{\partial \hat{m}^*(\check{\mu})}{\partial \mu'} (\hat{\mu}-\mu_0))=o_p(\frac{1}{\sqrt{N}})\]

Multiplying $\sqrt{N}$ on both sides, and rearranging, we have
\begin{align*}
& \sqrt{N}(\hat{\mu}-\mu_0)\\
& =(\frac{\partial \hat{m}^*(\hat{\mu})}{\partial \mu'}' \Upsilon_N \frac{\partial \hat{m}^*(\check{\mu})}{\partial \mu'})^{-1}  \frac{\partial \hat{m}^*(\hat{\mu})}{\partial \mu'}' \Upsilon_N \sqrt{N} \hat{m}^*(\mu_0) +o_p(1)
\end{align*}

Proposition \ref{prop:omega} establishes that 
\[\sqrt{N} \hat{m}^*(\mu_0) \xrightarrow{d} \mathcal{N}(0,\Omega) \]

Notice that $\hat{m}^*(\mu)$ depends on $\mu$ only through the sums of the original moment function $m(W_i,\tilde{\alpha}_{il},\mu)$, so
\begin{align*}
& \frac{\partial \hat{m}^*(\hat{\mu})}{\partial \mu'}\\
&= \frac{1}{N} \sum_{l=1}^L \sum_{i\in I_l}\frac{\partial m(W_i,\tilde{\alpha}_{il},\hat{\mu})}{\partial \mu'}\\
&\xrightarrow{p} G= \mathbb{E}(\frac{\partial m(W_i,\alpha_i,\mu_0)}{\partial \mu}) 
\end{align*}
by proposition \ref{prop:G}. As $\check{\mu}$ lies between $\hat{\mu}$ and $\mu_0$, $\check{\mu}$ is also consistent for $\mu_0$. Applying proposition \ref{prop:G} again gives the same convergence result with $\hat{\mu}$ replace by $\check{\mu}$. 

By assumption that $\Upsilon_N\xrightarrow{p} \Upsilon$, applying continuous mapping theorem, we get the asymptotic distribution in (\ref{eq:mu asym}).

The consistency of $\hat{V}$ follows from proposition \ref{prop:omegahat} and proposition \ref{prop:G}

\hfill \qedsymbol

\subsection{Consistency of Estimators}

\begin{lemma}
Suppose assumption \ref{as:rr}, assumption \ref{as:alpha} and assumption \ref{as:mom} hold. In addition,

\begin{enumerate}
	\item $\mu_0$ is the unique solution to the original moment condition $\mathbb{E}(m(W_i,\alpha_i,\mu_0))=0$
	\item The parameter space $\Theta$ for $\mu$ is compact
	\item $\Upsilon_N\xrightarrow{p} \Upsilon$ for a positive definite matrix $\Upsilon$.
	\item $\sup_{\alpha} \|\frac{\partial m(W_i,\alpha_i,\mu)}{\partial \alpha_i }\|^2\leq F(W_i)$ for all $\mu\in \Theta$ and $\mathbb{E}(F(W_i))<\infty$.
	\item $\mathbb{E}(\sup_{\mu\in \Theta}\|m(W_i,\alpha_i,\mu)\|)<\infty $
\end{enumerate}

then $\hat{\mu}\xrightarrow{p} \mu_0$.
\end{lemma}

\begin{proof}
The proof uses Theorem 2.6 in \cite{neweymcfadden}. We need to verify that the limit moment function $\mathbb{E}(m(W_i,\alpha_i,\mu))$ is continuous in $\mu$ and that the convergence is uniform in $\mu\in \Theta$.

We first show consistency of $\tilde{\mu}_l$ which does not involve the correction term:

Pointwise convergence follows from condition (iv):
\begin{align*}
& \frac{1}{N} \sum_{l=1}^L \sum_{i\in I_l} m(W_i,\tilde{\alpha}_{il},\mu)\\
& = \frac{1}{N} \sum_{l=1}^L \sum_{i\in I_l} m(W_i,\alpha_i,\mu) + \frac{1}{N} \sum_{l=1}^L \sum_{i\in I_l} \frac{\partial m(W_i,\check{\alpha}_i,\mu)}{\partial \alpha_i } (\tilde{\alpha}_{il}-\alpha_i)
\end{align*}
where $\check{\alpha}_i$ lies between $\alpha_i$ and $\tilde{\alpha}_{il}$. 

The last summation is $o_p(1)$ as
\begin{align*}
& \sup_{\mu\in \Theta}|\frac{1}{N}  \sum_{i\in I_l} \frac{\partial m(W_i,\check{\alpha}_i,\mu)}{\partial \alpha_i } (\tilde{\alpha}_{il}-\alpha_i)|\\
&\leq \sup_{\mu\in \Theta} \sqrt{\frac{1}{N}  \sum_{i\in I_l} \frac{\partial m(W_i,\check{\alpha}_i,\mu)}{\partial \alpha_i }^2 } \sqrt{\frac{1}{N}  \sum_{i\in I_l} (\tilde{\alpha}_{il}-\alpha_i)^2}\\
&\leq \sqrt{\frac{1}{N}  \sum_{i\in I_l} F(W_i) } o_p(1)\\
&=o_p(1)
\end{align*}

Standard law of large number applied to the first summation gives pointwise convergence. Condition (v) then guarantees the continuity of the limit function and uniform convergence. This establishes the consistency of $\tilde{\mu}_l$. 

To conclude consistency when the correction term $\tilde{a}_{il}(Y_{iT}-X_{iT}'\tilde{\beta}_l-\tilde{\alpha}_{il})\\$ is added, we note that the correction term does not depend on $\mu$. In particular, lemma \ref{lm:no rrhat}, \ref{lm:no betahat} implies that
\begin{align*}
& \frac{1}{N} \sum_{l=1}^L \sum_{i\in I_l} \tilde{a}_{il}(Y_{iT}-X_{iT}'\tilde{\beta}_l-\tilde{\alpha}_{il})\\
& = \frac{1}{N} \sum_{l=1}^L \sum_{i\in I_l} a(X_i,\alpha_i,\mu_0)(Y_{iT}-X_{iT}'\beta_0-\tilde{\alpha}_{il}) +o_p(1)\\
& = \frac{1}{N} \sum_{i=1}^N a(X_i,\alpha_i,\mu_0)(Y_{iT}-X_{iT}'\beta_0-\alpha_i)+\frac{1}{N}\sum_{l=1}^L \sum_{i\in I_l} a(X_i,\alpha_i,\mu_0) (\alpha_i-\tilde{\alpha}_{il})+o_p(1)\\
&= \frac{1}{N} \sum_{i=1}^N a(X_i,\alpha_i,\mu_0)(Y_{iT}-X_{iT}'\beta_0-\alpha_i)+ \sum_{l=1}^L \sqrt{\frac{1}{N} \sum_{i\in I_l} a(X_i,\alpha_i,\mu_0)^2} \sqrt{\frac{1}{N} \sum_{i\in I_l} (\alpha_i-\tilde{\alpha}_{il})^2} +o_p(1)\\
&=o_p(1)
\end{align*} 
The first summation is summing over iid mean zero variables by the orthogonal moment condition, so it will be $o_p(1)$. For the second term, the first term in the square root is $O_p(1)$ by the assumption that $a(X_i,\alpha_i,\mu_0)\in L_2(\mathbb{P}_{X,\alpha})$, and the second term under the square root is $o_p(1)$ by assumption \ref{as:alpha}. Therefore, we conclude that 

\begin{align*}
&\sup_{\mu\in \Theta}|\frac{1}{N} \sum_{l=1}^L \sum_{i\in I_l} m(W_i,\tilde{\alpha}_{il},\mu) - \hat{m}^*(\mu)|\\
& = |\frac{1}{N} \sum_{l=1}^L \sum_{i\in I_l} \tilde{a}_{il}(Y_{iT}-X_{iT}'\tilde{\beta}_l-\tilde{\alpha}_{il})|\\
&=o_p(1)
\end{align*}

In other words, $\hat{m}^*(\mu)$ has the limit uniformly in $\mu$ as $\mathbb{E}(m(W_i,\alpha_i,\mu))$ which is continuous.

\end{proof}

\section{Convergence rates for Elastic Net}

I will provide conditions that guarantee the convergence rate of $\tilde{a}_{il}$ required in assumption \ref{as:rr}. The techniques are adapted from \cite{AEN}, which proposes the adaptive elastic net estimator. Adaptive Elastic Net is expected to perform well in our setting because it has the oracle property of variable selection established in \cite{AEN} but also avoids problems faced by lasso when the covariates are highly correlated \citep{EN}. In our setup, our covariates are functions $b(X_i,\tilde{\alpha}_{ill'})$ where $X_i=(X_{i1}',\cdots,X_{iT}')'$ stacks the variables over all time periods. When the serial correlation among the covariates are strong, the performance of adaptive lasso could be undesirable. 

We are interested estimating the following conditional expectation:
\[\mathbb{E}(\frac{\partial m(W_i,\alpha_i,\mu_0)}{\partial \alpha_i} | X_i,\alpha_i)\]

We will assume that the conditional expectation takes the linear-form and the dimension of $b(X_i,\alpha_i)$ grows at the same rate as $T$ \footnote{The dimension of $b(X_i,\alpha_i)$ should not grow too fast with $N$. We can allow the transformation of $(X_i,\alpha_i)$ to be of larger order than $T$ by potentially imposing a stronger sparsity assumption on the non-zero elements of $\pi_0$. For simplicity, this extension is omitted.}
\begin{equation}
\frac{\partial m(W_i,\alpha_i,\mu_0)}{\partial \alpha_i} = b(X_i,\alpha_i)'\pi_0 + \xi_i, \quad \mathbb{E}(\xi_i|X_i)=0
\label{eq:EN}
\end{equation}

The above problem is about high-dimensional regression with the dimension of the parameter $\pi_0$ growing with $N$ at the rate $T\propto N^g$.\footnote{For dynamic panel data problems, if $Y_{i0}$ is not observed, then one has to drop one time period and work with $X_i$ that has dimension proportional to $T-1$. We will ignore this difference in the following discussion.} As standard in the literature, to learn the parameter $\pi_0$ well, I will impose sparsity assumptions on $\pi_0$ and introduce both $l_1$ and $l_2$ regularization to exploit the sparsity structure. The main challenge in my set-up is that the dependent variable $\frac{\partial m(W_i,\alpha_i,\mu_0)}{\partial \alpha_i}$ and $\alpha_i$ is not observed, and we only have an estimate of it constructed from cross-validation (See section \ref{procedure}): $\frac{\partial m(W_i,\tilde{\alpha}_{ill'},\tilde{\mu}_{l})}{\partial \alpha_i}$ and $\tilde{\alpha}_{ill'}$. Here we are learning (\ref{eq:EN}) only using samples not in fold $l$ and for $i\in l'$ estimate $\tilde{\alpha}_{ill'}$ of $\alpha_i$ is constructed from observation $i$ and other observations not in  fold $l$ or $l'$. Similarly, estimate $\tilde{\mu}_{l}$ of $\mu_0$ is constructed using only data not in fold $l$.

With some abuse of notation, in this section I will use $y_i^{EN}$ to denote $\frac{\partial m(W_i,\alpha_i,\mu_0)}{\partial \alpha_i}$ and $\tilde{y}_i^{EN}$ to denote the estimate $\frac{\partial m(W_i,\tilde{\alpha}_{ill'},\tilde{\mu}_{l})}{\partial \alpha_i}$ for individual $i$. Similarly, $B_i$ will be used to denote $b(X_i,\alpha_i)$, while $\tilde{B}_i$ will be used to denote $b(X_i,\tilde{\alpha}_{ill'})$. Our goal will be to establish condition (ii) and (iv) of assumption \ref{as:rr}. By the remark following assumption \ref{as:rr}, it suffices to establish (ii):  $\mathbb{E}((\tilde{a}_{il}-a(X_i,\alpha_i,\mu_0))^2|\mathcal{W}_l^C)=o_p(N^{-\zeta})$ which will imply (iv): $\frac{1}{|I_l|} \sum_{i\in I_l} [(\tilde{a}_{il}-a(X_i,\alpha_i,\mu_0))^2=o_p(N^{-\zeta})$.

By the assumption in (\ref{eq:EN}), we will have $a(X_i,\alpha_i,\mu_0)=B_i' \pi_0$ and $\tilde{a}_{il} = \tilde{B}_i' \hat{\pi}_l$ for a general estimator $\hat{\pi}_l$ that uses data not in fold $l$. Hence,
\begin{align*}
& \mathbb{E}((\tilde{a}_{il}-a(X_i,\alpha_i,\mu_0))^2|\mathcal{W}_l^C)\\
& = \mathbb{E}[((B_i-\tilde{B}_i)'\pi_0+\tilde{B}_i'(\pi_0-\hat{\pi}_l))^2|\mathcal{W}_l^C]\\
&\leq  2\mathbb{E}[((B_i-\tilde{B}_i)'\pi_0)^2 |\mathcal{W}_l^C] +  2 \mathbb{E}[(\tilde{B}_i'(\pi_0-\hat{\pi}_l))^2)|\mathcal{W}_l^C]\\
& \leq 2 \pi_0'  \mathbb{E}[(B_i-\tilde{B}_i)(B_i-\tilde{B}_i)'|\mathcal{W}_l^C] \pi_0  + 2 (\pi_0-\hat{\pi}_l)'\mathbb{E}(\tilde{B}_i\tilde{B}_i')(\pi_0-\hat{\pi}_l),
\end{align*}
where the last line uses independence for samples in fold $l$ and samples not in $l$. If the maximum eigenvalue of the matrix $\mathbb{E}(\tilde{B}_i\tilde{B}_i')$ is bounded, and $\pi_0'  \mathbb{E}[(B_i-\tilde{B}_i)(B_i-\tilde{B}_i)'|\mathcal{W}_l^C] \pi_0=O_p(T^{-1})$,  the problem can then be reduced to establishing $\|\pi_0-\hat{\pi}_l\|^2 = o_p(N^{-\zeta})$ for some $1-g\leq \zeta<1/2$. When there are only finitely many number of elements in $b(X_i,\alpha_i)$ that depends on $\alpha_i$ and $b(X_i,\alpha_i)$ is a smooth function of $\alpha_i$, $\pi_0'  \mathbb{E}[(B_i-\tilde{B}_i)(B_i-\tilde{B}_i)'|\mathcal{W}_l^C] \pi_0 $ will be $O(T^{-1})$ under suitable regularity conditions on $(X_{it},u_{it})$'s. 
In the following discussion, I will drop the subscript $l$ and abstract away from the use of cross-fitting. I will point out derivations that relies on cross-fitting features. For example, the number of individuals that are used to estimate (\ref{eq:EN}) will be the individuals not in fold $l$ which equals $N\frac{L-1}{L}$, for simplicity, I will just use $N$ instead. 

Now I will introduce the Elastic Net (EN) estimator for $\pi_0$, which is originally proposed by \cite{EN}. 
\begin{equation}
\hat{\pi}(EN) = (1+\lambda_2/N) \{\arg\min_{\pi} \|\tilde{\mathbf{y}}^{EN}-\tilde{\mathbf{B}}\pi\|^2+\lambda_2 \|\pi\|_2^2+\lambda_1 \|\pi\|_1\}
\label{eq:EN def}
\end{equation}
where $\tilde{\mathbf{y}}^{EN}$ stacks $\tilde{y}_i^{EN}$ into an $N\times 1$ vector and $\tilde{\mathbf{B}}$ stacks $\tilde{B}_i'$ into an over $N$. 

Ignoring the factor $(1+\lambda_2/N)$ then $\hat{\pi}(EN)$ equals the lasso estimator if $\lambda_2=0$. When $\lambda_2>0$, the introduction of $l_2$ regularization will improve prediction performance if $X_i$'s are highly correlated. In addition, the rescaling  $(1+\lambda_2/N)$ is used to undo the shrinkage introduced by the $l_2$ penalty term, so that the resulting estimator incorporates both variable selection via the lasso penalty and has desirable properties when there exists high correlation among the covariates.

The Adaptive Elastic Net estimator (AEN) is based on the first-step EN estimator. Specifically, we will introduce different penalty weights for different elements of $\pi$ by defining $\hat{w}_j=(|\hat{\pi}_j(EN)|)^{-\gamma}$, where $\gamma$ is a tuning parameter that satisfies $\gamma > \frac{2g}{1-g}$. \cite{AEN} suggests choosing $\gamma=\left\lceil \frac{2g}{1-g} \right\rceil  +1$, where $\lceil x \rceil$ denotes the smallest integer larger than $x$.  The AEN estimator introduced in \cite{AEN} is defined as 

\begin{equation}
\hat{\pi}(AEN) = (1+\lambda_2/N) \{\arg\min_{\pi} \|\tilde{\mathbf{y}}^{EN}-\mathbf{X}\pi\|^2+\lambda_2 \|\pi\|_2^2+\lambda_1^* \sum_{j=1}^T \hat{w}_j |\pi_j|\}
\label{eq:AEN def}
\end{equation}

\begin{remark}
\begin{enumerate}
	\item The penalty factors $\hat{w}_j$ are only introduced on the $l_1$ regularization part but not on the $l_2$ part.
	\item The use of the penalty factors is to enforce consistent variable selections as in adaptive lasso \citep{Alasso}. When $|\hat{\pi}_j(EN)|$ is small, then AEN will impose a larger penalty on $\pi_j$ and further shrinks it towards zero. Conversely, when $|\hat{\pi}_j(EN)|$ is relatively large, smaller regularization is imposed to reduce potential bias of estimating $\pi_0$.
	\item To avoid the problem of dividing a small number close to zero in finite sample, one can take $\hat{w}_j=(|\hat{\pi}_j(EN)|+\frac{1}{N})^{-\gamma}$ as suggested in \cite{AEN}.\footnote{I used the above penalty factor in the simulations.} 
\end{enumerate}
\end{remark}

The following assumption and lemmas are adapted from \cite{AEN}.
\begin{assumption}\label{as:AEN}
Let $\lambda_{min}(M)$ and $\lambda_{max}(M)$ denote the minimum and maximum eigenvalues of a positive-definite matrix $M$. Suppose
\begin{enumerate}
	\item $d\leq \lambda_{min}(\frac{1}{N} \tilde{\mathbf{B}}'\tilde{\mathbf{B}})\leq \lambda_{max}(\frac{1}{N} \tilde{\mathbf{B}}'\tilde{\mathbf{B}})\leq D$ for $0<d\leq D<\infty$.
	\item $\lim_{N\rightarrow \infty} \frac{\max_{i=1,\cdots,N} \sum_{j=1}^T B_{ij}^2}{N} =0$, $\pi_0'  \mathbb{E}[(B_i-\tilde{B}_i)(B_i-\tilde{B}_i)'|\mathcal{W}_l^C] \pi_0=O_p(T^{-1})$
	\item $\mathbb{E}(|\xi|^{2}|\mathbf{X})<\sigma_{\xi}^2$ for some $\sigma_{\xi}^2>0$
	\item $\lim_{N\rightarrow \infty} \frac{T}{N^g} =c $ for some $0<c<\infty$ and $1/2<g<1$
	\item $\lim_{N\rightarrow \infty} \frac{\lambda_2}{N}=0$, $\lim_{N\rightarrow \infty} \frac{\lambda_1}{\sqrt{N}}=0$
	\item $\lim_{N\rightarrow \infty} \frac{\lambda_2}{\sqrt{N}}\|\pi_0\|_2=0$, $\lim_{N\rightarrow \infty} \lambda_1^{*} N^{-1-(g-1)(1+\gamma)/2}=\infty$,
	$\lim_{N\rightarrow \infty} \frac{\lambda_1^*}{\sqrt{N}}=0$
	\[\lim_{N\rightarrow \infty}\min(\frac{N}{\lambda_1 \sqrt{T}}, (\frac{\sqrt{N}}{\sqrt{T} \lambda_1^*})^{1/\gamma} )   \min_{j:|(\pi_0)_j|>0} |(\pi_0)_j| \rightarrow \infty, \quad \min_{j:|(\pi_0)_j|>0} |(\pi_0)_j| > 0 \]
	\item For $F(W_i)$ defined in assumption \ref{as:G}, $\mathbb{E}( F(W_i) \|\tilde{\mu}_{l}-\mu_0\|^2 |\mathcal{W}_{ll'}^C)=O_p(\frac{1}{N})$
	\item $\sup_{\alpha} \|\frac{\partial^2 m(W_i,\alpha_i,\mu)}{\partial \alpha_i^2}\|^2\leq \tilde{F}(W_i) $ for all $\mu$ within a neighborhood $\mathcal{N}(\mu_0)$ of $\mu_0$ with $\tilde{F}(W_i)$ satisfying $\lim_{N \rightarrow \infty} \mathbb{E}(\tilde{F}(W_i) \|\bar{X}_i\|^2)<\infty$ and $\lim_{N\rightarrow \infty}  T\mathbb{E}(\tilde{F}(W_i) \bar{u}_i^2)< \infty$. 
	\item The number of non-zero elements of $\pi_0$ denoted as $|\mathcal{A}|$ grows slower than $T$: $\lim_{N\rightarrow \infty} \frac{|\mathcal{A}|}{T}=0$.
\end{enumerate}
\end{assumption}

\begin{remark}
\begin{enumerate}
	\item The first assumption can also be modified to be imposed upon the (unobserved) matrix $\frac{1}{N} \mathbf{B}'\mathbf{B}$ under suitable assumptions on $b(X_i,\alpha_i)$. For example, when only finitely many $K$ elements of $b(X_i,\alpha_i)$ depend on $\alpha_i$ with $B_i-\tilde{B}_i = O_p(T^{-1/2})$, then $\lambda_{max}(\frac{1}{N} \mathbf{B}'\mathbf{B}-\frac{1}{N} \tilde{\mathbf{B}}'\tilde{\mathbf{B}}) = O_p(T^{-1})$, and by theorem A.46 in \cite{spectral}, one can then bound $|\lambda_{max}(\frac{1}{N} \mathbf{B}'\mathbf{B})-\lambda_{max}(\frac{1}{N} \tilde{\mathbf{B}}'\tilde{\mathbf{B}})|\leq \lambda_{max}(\frac{1}{N} \mathbf{B}'\mathbf{B}-\frac{1}{N} \tilde{\mathbf{B}}'\tilde{\mathbf{B}}) = O_p(T^{-1})$. One can then convert it to a high probability statement as an upper bound of $\lambda_{max}(\frac{1}{N} \tilde{\mathbf{B}}'\tilde{\mathbf{B}})$, as $O_p(T^{-1})=o_p(N^{-\zeta})$, the following analysis can be done conditioning on the event that $\lambda_{max}(\frac{1}{N} \tilde{\mathbf{B}}'\tilde{\mathbf{B}})$ is bounded.
	
	\item The last three conditions are new in my context compared with \cite{AEN}. These are needed to control the error of not observing $\frac{\partial m(W_i,\alpha_i,\mu_0)}{\partial \alpha_i}$ and $B_i$. The conditions are standard regularity conditions that guarantees the existence of finite moments. For example, $\lim_{N\rightarrow \infty} T \mathbb{E}(\tilde{F}(W_i) \bar{u}_i^2)< \infty$ only requires the existence of moments for the random variable $\tilde{F}(W_i)(\frac{1}{\sqrt{T}} \sum_{t=1}^T u_{it})^2$ and typically a CLT for weakly dependent data yields $\frac{1}{\sqrt{T}} \sum_{t=1}^T u_{it}=O_p(1)$. One can also give more primitive mixing conditions to guarantee finite moments. For brevity, I will not give more detailed results here. 
\end{enumerate}
\end{remark}

In establishing the properties of EN estimator, a crucial step is to show that the difference between $y_i^{EN}$ and $\tilde{y}_i^{EN}$ will not affect the asymptotic properties of $\hat{\pi}(AEN)$. To do this we need to bound the error: $\mathbb{E}((\mathbf{y}^{EN}-\tilde{\mathbf{y}}^{EN})\tilde{\mathbf{B}} \tilde{\mathbf{B}}'(\mathbf{y}^{EN}-\tilde{\mathbf{y}}^{EN}))$

\begin{lemma}\label{lm:EN error}
Under assumption \ref{as:G}   and \ref{as:AEN},
\[\mathbb{E}((\mathbf{y}^{EN}-\tilde{\mathbf{y}}^{EN})'\tilde{\mathbf{B}} \tilde{\mathbf{B}}'(\mathbf{y}^{EN}-\tilde{\mathbf{y}}^{EN}))\\=o_p(NT)\]
\end{lemma}

\begin{proof}

\begin{align*}
&  \mathbb{E}((\mathbf{y}^{EN}-\tilde{\mathbf{y}}^{EN})\tilde{\mathbf{B}} \tilde{\mathbf{B}}'(\mathbf{y}^{EN}-\tilde{\mathbf{y}}^{EN}))\\
&\leq  \sum_{l'\neq l} \mathbb{E}[\mathbb{E}(\sum_{i\in I_{l'}}\|\tilde{B}_i'(y_i^{EN}-\tilde{y}_i^{EN})\|^2|\mathcal{W}_{ll'}^C)]
\end{align*}

Let $\tilde{\mathbf{B}}_{l'}$ stack the $\tilde{B}_i$ in rows for samples only in $l'$ and similarly for $\mathbf{y}_{l'}^{EN}$ and $\tilde{\mathbf{y}}_{l'}^{EN}$
By the variational representation of eigenvalues, we have
\begin{align*}
& \mathbb{E}[\mathbb{E}(\sum_{i\in I_{l'}}\|\tilde{B}_i'(y_i^{EN}-\tilde{y}_i^{EN})\|^2|\mathcal{W}_{ll'}^C)]\\
& =\mathbb{E}((\mathbf{y}_{l'}^{EN}-\tilde{\mathbf{y}}_{l'}^{EN})\tilde{\mathbf{B}}_{l'} \tilde{\mathbf{B}}_{l'}'(\mathbf{y}_{l'}^{EN}-\tilde{\mathbf{y}}_{l'}^{EN})|\mathcal{W}_{ll'}^C)\\
&\leq \mathbb{E}(\lambda_{max}(\tilde{\mathbf{B}}_{l'} \tilde{\mathbf{B}}_{l'}') \|\mathbf{y}_{l'}^{EN}-\tilde{\mathbf{y}}_{l'}^{EN}\|^2)\\
&\leq \lambda_{max}(\tilde{\mathbf{B}}_{l'}'\tilde{\mathbf{B}}_{l'}) \mathbb{E}(\|\mathbf{y}_l^{EN}-\tilde{\mathbf{y}}_l^{EN}\|^2)\\
&\leq D N \mathbb{E}(\|\mathbf{y}_{l'}^{EN}-\tilde{\mathbf{y}}_{l'}^{EN}\|^2)\\
& \leq  D N^2 \mathbb{E}(y_i^{EN}-\tilde{y}_i^{EN})^2
\end{align*}
here I used the property that $\tilde{\mathbf{B}}_{l'} \tilde{\mathbf{B}}_{l'}'$ and $\tilde{\mathbf{B}}_{l'}'\tilde{\mathbf{B}}_{l'}$ shares the same  bounds on eigenvalues of $\frac{1}{N} \tilde{\mathbf{B}}'\tilde{\mathbf{B}}$.

Recall that $y_i^{EN}=\frac{\partial m(W_i,\alpha_i,\mu_0)}{\partial \alpha_i}$ and $\tilde{y}_i^{EN}=\frac{\partial m(W_i,\tilde{\alpha}_{ill'},\tilde{\mu}_{l})}{\partial \alpha_i}$, a first-order Mean-Value expansion gives that for some $\check{\mu}_{l}$ lying between $\tilde{\mu}_{l}$ and $\mu_0$ and $\check{\alpha}_{ill'}$ lying between $\alpha_i$ and $\tilde{\alpha}_{ill'}$

\begin{align*}
& (y_i^{EN}-\tilde{y}_i^{EN})^2\\
& = (\frac{\partial m(W_i,\alpha_i,\mu_0)}{\partial \alpha_i}-\frac{\partial m(W_i,\tilde{\alpha}_{ill'},\tilde{\mu}_{l})}{\partial \alpha_i})^2\\
&=(\frac{\partial^2 m(W_i,\check{\alpha}_{ill'},\check{\mu}_{l})}{\partial \alpha_i^2}(\tilde{\alpha}_{ill'}-\alpha)+\frac{\partial^2 m(W_i,\check{\alpha}_{ill'},\check{\mu}_{l})}{\partial \alpha_i \partial \mu'} (\tilde{\mu}_{l}-\mu_0))^2\\
&\leq 2 (\frac{\partial^2 m(W_i,\check{\alpha}_{ill'},\check{\mu}_{l})}{\partial \alpha_i^2})^2(\tilde{\alpha}_{ill'}-\alpha)^2\\
&+2\|\frac{\partial^2 m(W_i,\check{\alpha}_{ill'},\check{\mu}_{l})}{\partial \alpha_i \partial \mu'}\|^2 \|\tilde{\mu}_{l}-\mu_0\|^2
\end{align*}

We immediately see that under assumption \ref{as:G}, we have
\begin{align*}
& \mathbb{E}(\|\frac{\partial^2 m(W_i,\check{\alpha}_{ill'},\check{\mu}_{l})}{\partial \alpha_i \partial \mu'} (\tilde{\mu}_{l}-\mu_0)]\|^2|\mathcal{W}_{ll'}^C)\\
&\leq   \mathbb{E}( F(W_i) \|\tilde{\mu}_{l}-\mu_0\|^2 |\mathcal{W}_{ll'}^C)\\
&= O_p(\frac{1}{N}) 
\end{align*}

For the quantity $(\tilde{\alpha}_{ill'}-\alpha)^2$, I will only consider the case where $\tilde{\alpha}_{ill'}$ is defined as the fixed effect averages $\bar{Y}_i-\bar{X}_i'\tilde{\beta}_{ll'}$. When $\tilde{\alpha}_{ill'}$ incorporates EB or SURE corrections, similar techniques used to prove lemma \ref{lm:no alphahat} can be used to establish the same rates. For simplicity, formal derivations are omitted. 

\begin{align*}
& (\tilde{\alpha}_{ill'}-\alpha_i)^2\\
&=(\bar{\bar{X}}_i'(\tilde{\beta}_{ll'}-\beta_0)+\bar{\bar{u}}_i)^2\\
&\leq 2 \|\bar{\bar{X}}_i\|_2^2 \|\tilde{\beta}_{ll'}-\beta_0\|_2^2+2 \bar{\bar{u}}_i^2
\end{align*}

Hence,
\begin{align*}
& \mathbb{E}(\| \frac{\partial^2 m(W_i,\check{\alpha}_{ill'},\check{\mu}_{l})}{\partial \alpha_i^2} (\tilde{\alpha}_{ill'}-\alpha)\|^2|\mathcal{W}_{ll'}^C)\\
&\leq 2  \|\tilde{\beta}_{ll'}-\beta_0\|_2^2 \mathbb{E}(  \|\bar{\bar{X}}_i\|_2^2 \tilde{F}(W_i) |\mathcal{W}_{ll'}^C)\\
&+2  \mathbb{E}(  \bar{\bar{u}}_i^2 \tilde{F}(W_i) |\mathcal{W}_{ll'}^C)\\
&\leq O_p(1/N)+O_p(1/T)
\end{align*}

Finally, combining the above results, we see that 
\begin{align*}
& \mathbb{E}((\mathbf{y}^{EN}-\tilde{\mathbf{y}}^{EN})\tilde{\mathbf{B}} \tilde{\mathbf{B}}'(\mathbf{y}^{EN}-\tilde{\mathbf{y}}^{EN}))\\
&\leq 2 L^2 D N^2 \mathbb{E}\{\mathbb{E}(\|\frac{\partial^2 m(W_i,\check{\alpha}_{ill'},\check{\mu}_{l})}{\partial \alpha_i \partial \mu'} (\tilde{\mu}_{ill'}-\mu_0)]\|^2|\mathcal{W}_{ll'}^C)\\
&+ \mathbb{E}(\| \frac{\partial^2 m(W_i,\check{\alpha}_{ill'},\check{\mu}_{l})}{\partial \alpha_i^2} (\tilde{\alpha}_{ill'}-\alpha)\|^2|\mathcal{W}_{ll'}^C)\}\\
&=N^2 (O_p(1/N)+O_p(1/T))\\
&=O_p(N^2/T)
\end{align*}

Now use the fact that $T\propto N^g$ for some $g>1/2$ to conclude that the above display is $o_p(NT)$

\end{proof}

\begin{lemma}\label{lm:EN b error}
Under assumption \ref{as:AEN},
\[\mathbb{E}( \pi_0' (\mathbf{B}-\tilde{\mathbf{B}})' \tilde{\mathbf{B}} \tilde{\mathbf{B}}' (\mathbf{B}-\tilde{\mathbf{B}}) \pi_0)=o_p(NT)\]

Similarly,
\[\mathbb{E}( \pi_0' (\mathbf{B}-\tilde{\mathbf{B}})' \tilde{\mathbf{B}}_{\mathcal{A}} \tilde{\mathbf{B}}_{\mathcal{A}}' (\mathbf{B}-\tilde{\mathbf{B}}) \pi_0)=o_p(NT)\]
where the subscript in $\tilde{\mathbf{B}}_{\mathcal{A}}$ for $\mathcal{A}\subset \{1,\cdots,T\}$ means taking submatrix with columns of $\tilde{\mathbf{B}}$ that in the set $\mathcal{A}$
\end{lemma}

\begin{proof}
\begin{align*}
& \mathbb{E}( \pi_0' (\mathbf{B}-\tilde{\mathbf{B}})' \tilde{\mathbf{B}} \tilde{\mathbf{B}}' (\mathbf{B}-\tilde{\mathbf{B}}) \pi_0)\\
&\leq \mathbb{E}( \lambda_{max}( \tilde{\mathbf{B}} \tilde{\mathbf{B}}') \pi_0' (\mathbf{B}-\tilde{\mathbf{B}})' (\mathbf{B}-\tilde{\mathbf{B}}) \pi_0))\\
&\leq N D \mathbb{E}(\pi_0' (\mathbf{B}-\tilde{\mathbf{B}})' (\mathbf{B}-\tilde{\mathbf{B}}) \pi_0)\\
&=N^2 D \mathbb{E}(\pi_0' (B_i-\tilde{B}_i)'(B_i-\tilde{B}_i) \pi_0)\\
&= N^2  D O(T^{-1})\\
&=o(NT)
\end{align*}
where the second last line uses condition (ii) of assumption \ref{as:AEN} and the last line uses $T$ grow faster than $N^{1/2}$.

The second result follows the same proof by noticing that $\lambda_{max}( \tilde{\mathbf{B}}_{\mathcal{A}} \tilde{\mathbf{B}}_{\mathcal{A}}')\leq \lambda_{max}( \tilde{\mathbf{B}} \tilde{\mathbf{B}}')$
\end{proof}

\begin{lemma}\label{lm:EN rate}
Under assumption \ref{as:rr} and \ref{as:AEN}, for a generic penalty factor $\hat{w}_j$ that may depend on $\mathbf{X}$ and $\tilde{\mathbf{y}}^{EN}$, define

\[\hat{\pi}_{\hat{w}}(\lambda_2,\lambda_1) = \arg\min_{\pi} \|\tilde{\mathbf{y}}^{EN}-\tilde{\mathbf{B}}\pi\|^2+\lambda_2 \|\pi\|_2^2+\lambda_1 \sum_{j=1}^T \hat{w}_j |\pi_j|\]
then,
\begin{equation}
\mathbb{E}(\|\hat{\pi}_{\hat{w}}(\lambda_2,\lambda_1)-\pi_0\|_2^2)\leq \frac{2 \lambda_1^2 \mathbb{E}(\sum_{j=1}^T \hat{w}_j^2) +8 \lambda_2^2 \|\pi_0\|_2^2 + 8 D N T \sigma_{\xi}^2 +o(NT)}{(N d+\lambda_2)^{-2}}
\label{eq:AEN 1step}
\end{equation}
\end{lemma}

\begin{remark}
\begin{enumerate}
\item For EN estimator $\hat{w}_j=1$ for all $j$, and one can obtain a similar bound by replacing $\mathbb{E}(\sum_{j=1}^T \hat{w}_j^2)$ with $T$ as noticed in \cite{AEN}.
\item The lemma suggests that the convergence rate for $\hat{\pi}_{\hat{w}}(\lambda_2,\lambda_1)$ is $O_p(T/N)$ and is consistent for $\pi_0$ if $T$ grows slower than $N$. 
\end{enumerate}
\end{remark}

\begin{proof}
We will also consider a ridge regression estimator defined as
\[\hat{\pi}(\lambda_2,0)=\arg\min_{\pi} \|\tilde{\mathbf{y}}^{EN}-\tilde{\mathbf{B}}\pi\|^2+\lambda_2 \|\pi\|_2^2 \]

We are going to use 
\[\|\hat{\pi}_{\hat{w}}(\lambda_2,\lambda_1)-\pi_0\|_2^2\leq 2 \underbrace{\|\hat{\pi}_{\hat{w}}(\lambda_2,\lambda_1)-\hat{\pi}(\lambda_2,0)\|_2^2}_{I}+2 \underbrace{\|\hat{\pi}(\lambda_2,0)-\pi_0\|_2^2}_{II}\]
and bound the terms on the RHS separately.

\textbf{Bounding term (I)}

By the definition of $\hat{\pi}(\lambda_2,0)$ and $\hat{\pi}_{\hat{w}}(\lambda_2,\lambda_1)$, we have
\begin{align*}
& \|\tilde{\mathbf{y}}^{EN}-\tilde{\mathbf{B}}\hat{\pi}(\lambda_2,0)\|^2+\lambda_2 \|\hat{\pi}(\lambda_2,0)\|_2^2+\lambda_1 \sum_{j=1}^T \hat{w}_j |\hat{\pi}_j(\lambda_2,0)|\\
&\geq  \|\tilde{\mathbf{y}}^{EN}-\tilde{\mathbf{B}}\hat{\pi}_{\hat{w}}(\lambda_2,\lambda_1)\|^2+\lambda_2 \|\hat{\pi}_{\hat{w}}(\lambda_2,\lambda_1)\|_2^2+\lambda_1 \sum_{j=1}^T \hat{w}_j |\hat{\pi}_{\hat{w}}(\lambda_2,\lambda_1)_j|
\end{align*}

and
\[\|\tilde{\mathbf{y}}^{EN}-\tilde{\mathbf{B}}\hat{\pi}_{\hat{w}}(\lambda_2,\lambda_1)\|^2+\lambda_2 \|\hat{\pi}_{\hat{w}}(\lambda_2,\lambda_1)\|_2^2\geq \|\tilde{\mathbf{y}}^{EN}-\tilde{\mathbf{B}}\hat{\pi}(\lambda_2,0)\|^2+\lambda_2 \|\hat{\pi}(\lambda_2,0)\|_2^2 \]

The first one implies
\begin{align*}
& \lambda_1 \sum_{j=1}^T \hat{w}_j (|\hat{\pi}_j(\lambda_2,0)|-|\hat{\pi}_{\hat{w}}(\lambda_2,\lambda_1)_j|)\\
&\geq (\|\tilde{\mathbf{y}}^{EN}-\tilde{\mathbf{B}}\hat{\pi}_{\hat{w}}(\lambda_2,\lambda_1)\|^2+\lambda_2 \|\hat{\pi}_{\hat{w}}(\lambda_2,\lambda_1)\|_2^2)\\
& -(\|\tilde{\mathbf{y}}^{EN}-\tilde{\mathbf{B}}\hat{\pi}(\lambda_2,0)\|^2+\lambda_2 \|\hat{\pi}(\lambda_2,0)\|_2^2)
\end{align*}

The RHS of the above display can be simplified by using the closed-form solution of $\hat{\pi}(\lambda_2,0)=(\tilde{\mathbf{B}}'\tilde{\mathbf{B}}+\lambda_2 I)^{-1} \tilde{\mathbf{B}}' \tilde{\mathbf{y}}^{EN}$:
\begin{align*}
& (\|\tilde{\mathbf{y}}^{EN}-\tilde{\mathbf{B}}\hat{\pi}_{\hat{w}}(\lambda_2,\lambda_1)\|^2+\lambda_2 \|\hat{\pi}_{\hat{w}}(\lambda_2,\lambda_1)\|_2^2)\\
& -(\|\tilde{\mathbf{y}}^{EN}-\tilde{\mathbf{B}}\hat{\pi}(\lambda_2,0)\|^2+\lambda_2 \|\hat{\pi}(\lambda_2,0)\|_2^2)\\
&=(\hat{\pi}_{\hat{w}}(\lambda_2,\lambda_1)-\hat{\pi}(\lambda_2,0))(\tilde{\mathbf{B}}'\tilde{\mathbf{B}}+\lambda_2 I)(\hat{\pi}_{\hat{w}}(\lambda_2,\lambda_1)-\hat{\pi}(\lambda_2,0))
\end{align*}

We can further upper bound the LHS using Cauchy-Schwartz:
\[\sum_{j=1}^T \hat{w}_j (|\hat{\pi}_j(\lambda_2,0)|-|\hat{\pi}_{\hat{w}}(\lambda_2,\lambda_1)_j|)\leq \sqrt{\sum_{j=1}^T \hat{w}_j^2} \|\hat{\pi}_j(\lambda_2,0)-\hat{\pi}_{\hat{w}}(\lambda_2,\lambda_1)\|_2\]

Finally, use the property that $\lambda_{min}(\tilde{\mathbf{B}}'\tilde{\mathbf{B}}+\lambda_2 I)=\lambda_{min}(\tilde{\mathbf{B}}'\tilde{\mathbf{B}})+\lambda_2$, we have
\begin{align*}
& \lambda_1 \sqrt{\sum_{j=1}^T \hat{w}_j^2} \|\hat{\pi}_j(\lambda_2,0)-\hat{\pi}_{\hat{w}}(\lambda_2,\lambda_1)\|_2\\
&\geq \lambda_1 \sum_{j=1}^T \hat{w}_j (|\hat{\pi}_j(\lambda_2,0)|-|\hat{\pi}_{\hat{w}}(\lambda_2,\lambda_1)_j|)\\
&\geq (\hat{\pi}_{\hat{w}}(\lambda_2,\lambda_1)-\hat{\pi}(\lambda_2,0))(\tilde{\mathbf{B}}'\tilde{\mathbf{B}}+\lambda_2 I)(\hat{\pi}_{\hat{w}}(\lambda_2,\lambda_1)-\hat{\pi}(\lambda_2,0))\\
&\geq (\lambda_{min}(\tilde{\mathbf{B}}'\tilde{\mathbf{B}})+\lambda_2)\|\hat{\pi}_{\hat{w}}(\lambda_2,\lambda_1)-\hat{\pi}(\lambda_2,0)\|_2^2
\end{align*}

Rearranging the term yields
\[\|\hat{\pi}_{\hat{w}}(\lambda_2,\lambda_1)-\hat{\pi}(\lambda_2,0)\|_2 \leq \frac{\lambda_1 \sqrt{\sum_{j=1}^T \hat{w}_j^2}}{\lambda_{min}(\tilde{\mathbf{B}}'\tilde{\mathbf{B}})+\lambda_2}\]

\textbf{Bounding term (II)}

I will have to deal with the fact that $\tilde{y}_i^{EN}$ and $y_i^{EN}$ differ and that $B_i$ and $\tilde{B}_i$ differ for the term (II). 

Note 
\begin{align*}
& \hat{\pi}(\lambda_2,0)-\pi_0\\
&=(\tilde{\mathbf{B}}'\tilde{\mathbf{B}}+\lambda_2 I)^{-1} \tilde{\mathbf{B}}'(\mathbf{B} \pi_0+\xi+\tilde{\mathbf{y}}^{EN}-\mathbf{y}^{EN})-\pi_0\\
&= -\lambda_2 (\tilde{\mathbf{B}}'\tilde{\mathbf{B}}+\lambda_2 I)^{-1} \pi_0+(\tilde{\mathbf{B}}'\tilde{\mathbf{B}}+\lambda_2 I)^{-1} \tilde{\mathbf{B}}'\xi\\
&+(\tilde{\mathbf{B}}'\tilde{\mathbf{B}}+\lambda_2 I)^{-1} \tilde{\mathbf{B}}'(\tilde{\mathbf{y}}^{EN}-\mathbf{y}^{EN})+ (\tilde{\mathbf{B}}'\tilde{\mathbf{B}}+\lambda_2 I)^{-1} \tilde{\mathbf{B}}' (\mathbf{B}-\tilde{\mathbf{B}}) \pi_0
\end{align*}
Here the last two terms are new compared with results in \cite{AEN}.

\begin{align*}
& \mathbb{E}(\|\hat{\pi}(\lambda_2,0)-\pi_0\|_2^2)\\
&\leq 4 \lambda_2^2 (\lambda_{min}(\tilde{\mathbf{B}}'\tilde{\mathbf{B}})+\lambda_2)^{-2} \|\pi_0\|_2^2+4\mathbb{E}(\|\tilde{\mathbf{B}}'\tilde{\mathbf{B}}+\lambda_2 I)^{-1} \tilde{\mathbf{B}}'\xi\|_2^2)\\
&+ 4 \mathbb{E}(\|(\tilde{\mathbf{B}}'\tilde{\mathbf{B}}+\lambda_2 I)^{-1} \tilde{\mathbf{B}}'(\tilde{\mathbf{y}}^{EN}-\mathbf{y}^{EN})\|_2^2)\\
&+ 4 \mathbb{E}(\|(\tilde{\mathbf{B}}'\tilde{\mathbf{B}}+\lambda_2 I)^{-1}  \tilde{\mathbf{B}}' (\mathbf{B}-\tilde{\mathbf{B}}) \pi_0 \|_2^2 )\\
&\leq 4 \lambda_2^2 (\lambda_{min}(\tilde{\mathbf{B}}'\tilde{\mathbf{B}})+\lambda_2)^{-2} \|\pi_0\|_2^2\\
&+4 (\lambda_{min}(\tilde{\mathbf{B}}'\tilde{\mathbf{B}})+\lambda_2)^{-2} \mathbb{E}(\xi'\tilde{\mathbf{B}} \tilde{\mathbf{B}}'\xi)\\
&+4 (\lambda_{min}(\tilde{\mathbf{B}}'\tilde{\mathbf{B}})+\lambda_2)^{-2} \mathbb{E}(\|\tilde{\mathbf{B}}'(\tilde{\mathbf{y}}^{EN}-\mathbf{y}^{EN}\|_2^2)\\
&+ 4 (\lambda_{min}(\tilde{\mathbf{B}}'\tilde{\mathbf{B}})+\lambda_2)^{-2} \mathbb{E}( \pi_0' (\mathbf{B}-\tilde{\mathbf{B}})' \tilde{\mathbf{B}} \tilde{\mathbf{B}}' (\mathbf{B}-\tilde{\mathbf{B}}) \pi_0)\\
&\leq 4  (\lambda_{min}(\tilde{\mathbf{B}}'\tilde{\mathbf{B}})+\lambda_2)^{-2}(\lambda_2^2 \|\pi_0\|_2^2 +Tr(\tilde{\mathbf{B}}'\tilde{\mathbf{B}}) \sigma_{\xi}^2 +o(NT))\\
&\leq 4 (\lambda_{min}(\mathbf{X}'\mathbf{X})+\lambda_2)^{-2}(\lambda_2^2 \|\pi_0\|_2^2 + \lambda_{max}(\tilde{\mathbf{B}}'\tilde{\mathbf{B}}) T \sigma_{\xi}^2 +o(NT))
\end{align*}
where lemma \ref{lm:EN error} is used to bound $\mathbb{E}(\|\tilde{\mathbf{B}}'(\tilde{\mathbf{y}}^{EN}-\mathbf{y}^{EN}\|_2^2)$ and lemma \ref{lm:EN b error} to bound $\mathbb{E}( \pi_0' (\mathbf{B}-\tilde{\mathbf{B}})' \tilde{\mathbf{B}} \tilde{\mathbf{B}}' (\mathbf{B}-\tilde{\mathbf{B}}) \pi_0)$.

Now combining the bounds on the two terms, we have 
\begin{align*}
& \mathbb{E}(\|\hat{\pi}_{\hat{w}}(\lambda_2,\lambda_1)-\pi_0\|_2^2)\\
&\leq \frac{2 \lambda_1^2 \mathbb{E}(\sum_{j=1}^T \hat{w}_j^2) +8 \lambda_2^2 \|\pi_0\|_2^2 + 8 \lambda_{max}(\mathbf{X}'\mathbf{X}) T \sigma_{\xi}^2 +o(NT)}{(N d+\lambda_2)^{2}}\\
&\leq \frac{2 \lambda_1^2 \mathbb{E}(\sum_{j=1}^T \hat{w}_j^2) +8 \lambda_2^2 \|\pi_0\|_2^2 + 8 D N T \sigma_{\xi}^2 +o(NT)}{(N d+\lambda_2)^{2}}
\end{align*}
where the leading term in the numerator is $8 D N T \sigma_{\xi}^2$ which is of order $NT$.

\end{proof}

Let $\mathcal{A}$ be the set of coordinates of $\pi_0$ that is different from zero. Formally,
\[\mathcal{A}=\{j:(\pi_0)_j \neq 0, j=1,\cdots,T \}\]
Let $\tilde{\mathbf{B}}_{\mathcal{A}}$ denote the submatrix of $\tilde{\mathbf{B}}$ that only contains columns whose index is in $\mathcal{A}$, and similarly, let $\pi_{\mathcal{A}}$ denote the subvector of $\pi$ that contains indexes in $\mathcal{A}$.

Now I will state an oracle property of the AEN estimator:
\begin{lemma}\label{lm:EN oracle}
Let $\pi_0=((\pi_0)_{\mathcal{A}},0)$, and define
\[\hat{\pi}_{\mathcal{A}} = \arg\min_{\pi} \{\|\tilde{\mathbf{y}}^{EN}-\tilde{\mathbf{B}}_{\mathcal{A}} \pi\|_2^2 + \lambda_2 \|\pi\|_2^2 + \lambda_1^* \sum_{j=1}^T \hat{w}_j |\pi_j|\}\]
Then, under assumption \ref{as:AEN} and assumption \ref{as:G}, with probability going to 1, $((1+\lambda_2/N)\hat{\pi}_{\mathcal{A}},0)$ is going to be equal to $\hat{\pi}(AEN)$.
\end{lemma}

The above lemma shows that the EN estimator successfully find out the elements of $\pi_0$ that is zero under the stated assumptions. This is equivalent to theorem 3.2 in \cite{AEN}.

\begin{proof}
The proof is to establish that the Kuhn-Tucker condition for the AEN estimator is satisfied. More specifically, we need to bound the probability that AEN will set $|\hat{\pi}(AEN)_j|$ to be greater than zero for all $j\in \mathcal{A}^C$. Here, $\mathcal{A}^C$ is understood as $\{1,\cdots,T\}-\mathcal{A}$.

Formally, we need to show that 
\[\mathbb{P}(\exists j\in \mathcal{A}^C: |-2 \tilde{\mathbf{B}}_j'(\tilde{\mathbf{y}}^{EN}-\tilde{\mathbf{B}}_{\mathcal{A}} \hat{\pi}_{\mathcal{A}})| >\lambda_1^* \hat{w}_j )\rightarrow 0\]

We further define $\eta = \min_{j\in \mathcal{A}} |(\pi_0)_j|$ and $\hat{\eta}=\min_{j\in \mathcal{A}} |(\hat{\pi}(EN))_j|$. Notice
\begin{align}
& \mathbb{P}(\hat{\eta}\leq \eta/2) \nonumber \\
&\leq \mathbb{P}(\|\hat{\pi}(EN)-\pi_0\|\geq \eta/2)\nonumber \\
&\leq \frac{4\mathbb{E}(\|\hat{\pi}(EN)-\pi_0\|_2^2)}{\eta^2} \nonumber \\
&\leq 4\frac{2 \lambda_1^2 T +8 \lambda_2^2 \|\pi_0\|_2^2 + 8 D N T \sigma_{\xi}^2 +o(NT)}{(N d+\lambda_2)^{2} \eta^2 },
\label{eq:EN eta}
\end{align}
where the first inequality is because the norm of $\hat{\pi}(EN)-\pi_0$ must be greater than the difference in the value of smallest elements. The following steps use Markov inequality and results from lemma \ref{lm:EN rate}.

We now apply union bounds to have
\begin{align*}
& \mathbb{P}(\exists j\in \mathcal{A}^C: |-2 \tilde{\mathbf{B}}_j'(\tilde{\mathbf{y}}^{EN}-\tilde{\mathbf{B}}_{\mathcal{A}} \hat{\pi}_{\mathcal{A}})| >\lambda_1^* \hat{w}_j )\\
&\leq \sum_{j\in \mathcal{A}^C} \mathbb{P}(|-2 \tilde{\mathbf{B}}_j'(\tilde{\mathbf{y}}^{EN}-\tilde{\mathbf{B}}_{\mathcal{A}} \hat{\pi}_{\mathcal{A}})| >\lambda_1^* \hat{w}_j, \hat{\eta}>\eta/2)+\underbrace{\mathbb{P}(\hat{\eta}\leq \eta/2)}_{I}
\end{align*}
We already have a bound on term (I) from (\ref{eq:EN eta}). To further bound the summation in the second line, let $M=(\frac{\lambda_1^*}{N})^{1/(1+\gamma)}$ (recall $\gamma$ is used to construct $\hat{w}_j = (\hat{\pi}(EN)_j)^{-\gamma}$.

\begin{align*}
& \sum_{j\in \mathcal{A}^C} \mathbb{P}(|-2 \tilde{\mathbf{B}}_j'(\tilde{\mathbf{y}}^{EN}-\tilde{\mathbf{B}}_{\mathcal{A}} \hat{\pi}_{\mathcal{A}})| >\lambda_1^* \hat{w}_j, \hat{\eta}>\eta/2)\\
&\leq \sum_{j\in \mathcal{A}^C} \mathbb{P}(|-2 \tilde{\mathbf{B}}_j'(\tilde{\mathbf{y}}^{EN}-\tilde{\mathbf{B}}_{\mathcal{A}} \hat{\pi}_{\mathcal{A}})| >\lambda_1^* \hat{w}_j, \hat{\eta}>\eta/2, |\hat{\pi}(EN)_j|<M)\\
&+ \sum_{j\in \mathcal{A}^C}  \mathbb{P}(|\hat{\pi}(EN)_j|\geq M)\\
&\leq \sum_{j\in \mathcal{A}^C} \mathbb{P}(|-2 \tilde{\mathbf{B}}_j'(\tilde{\mathbf{y}}^{EN}- \tilde{\mathbf{B}}_{\mathcal{A}} \hat{\pi}_{\mathcal{A}})| >\lambda_1^* M^{-\gamma}, \hat{\eta}>\eta/2)+\frac{\mathbb{E}(\sum_{j\in \mathcal{A}^C} |\hat{\pi}(EN)_j|^2)}{M^2}\\
&\leq \frac{4 M^{2\gamma}}{\lambda_1^{*2}} \mathbb{E}(\sum_{j\in \mathcal{A}^C} |\tilde{\mathbf{B}}_j'(\tilde{\mathbf{y}}^{EN}-\tilde{\mathbf{B}}_{\mathcal{A}} \hat{\pi}_{\mathcal{A}})|^2 \mathds{1}(\hat{\eta}>\eta/2))\\
&+ \underbrace{\frac{\mathbb{E}(\|\hat{\pi}(EN)_j-\pi_0\|_2^2)}{M^2}}_{II}
\end{align*}

Using lemma \ref{lm:EN rate}, we have a bound on the term (II) as 
\[(II)\leq \frac{2 \lambda_1^2 T +8 \lambda_2^2 \|\pi_0\|_2^2 + 8 D N T \sigma_{\xi}^2 +o(NT)}{(N d+\lambda_2)^{2} M^2}\]

We now deal with $\sum_{j\in \mathcal{A}^C} |\tilde{B}_j'(\tilde{\mathbf{y}}^{EN}- \tilde{\mathbf{B}}_{\mathcal{A}} \hat{\pi}_{\mathcal{A}})|^2$:

\begin{align*}
&\sum_{j\in \mathcal{A}^C} |\tilde{\mathbf{B}}_j'(\tilde{\mathbf{y}}^{EN}-\tilde{\mathbf{B}}_{\mathcal{A}} \hat{\pi}_{\mathcal{A}})|^2\\
&= \sum_{j\in \mathcal{A}^C} |\tilde{\mathbf{B}}_j'(\tilde{\mathbf{B}}_{\mathcal{A}} [(\pi_0)_{\mathcal{A}}-\hat{\pi}_{\mathcal{A}}]+ (\mathbf{B}_{\mathcal{A}}-\tilde{\mathbf{B}}_{\mathcal{A}})(\pi_0)_{\mathcal{A}}+\xi+\tilde{\mathbf{y}}^{EN}-\mathbf{y}^{EN})|^2\\
&\leq 4 \sum_{j\in \mathcal{A}^C} |\tilde{\mathbf{B}}_j' (\tilde{\mathbf{B}}_{\mathcal{A}} [(\pi_0)_{\mathcal{A}}-\hat{\pi}_{\mathcal{A}}])|^2 + 4 \sum_{j\in \mathcal{A}^C}|\tilde{\mathbf{B}}_j' (\mathbf{B}_{\mathcal{A}}-\tilde{\mathbf{B}}_{\mathcal{A}})(\pi_0)_{\mathcal{A}} |^2 \\
& + 4 \sum_{j\in \mathcal{A}^C} |\tilde{\mathbf{B}}_j'\xi|^2 + 4 \sum_{j\in \mathcal{A}^C} |\tilde{\mathbf{B}}_j'(\tilde{\mathbf{y}}^{EN}-\mathbf{y}^{EN})|^2\\
&\leq 4 D^2 N^2 \|(\pi_0)_{\mathcal{A}}-\hat{\pi}_{\mathcal{A}})\|_2^2 +4 (\pi_0)_{\mathcal{A}}'(\mathbf{B}_{\mathcal{A}}-\tilde{\mathbf{B}}_{\mathcal{A}})' \tilde{\mathbf{B}} \tilde{\mathbf{B}}' (\mathbf{B}_{\mathcal{A}}-\tilde{\mathbf{B}}_{\mathcal{A}}) (\pi_0)_{\mathcal{A}} \\
&+ 4  \sum_{j\in \mathcal{A}^C} |\tilde{B}_j'\xi|^2 + 4 \|\tilde{\mathbf{B}} (\tilde{\mathbf{y}}^{EN}-\mathbf{y}^{EN})\|_2^2
\end{align*}

Taking expectations, we have
\begin{align*}
& \mathbb{E}(\sum_{j\in \mathcal{A}^C} |\tilde{\mathbf{B}}_j'(\tilde{\mathbf{y}}^{EN}-\tilde{\mathbf{B}} \hat{\pi}_{\mathcal{A}})|^2 \mathds{1}(\hat{\eta}>\eta/2))\\
&\leq 4 D^2 N^2 \mathbb{E}(\|(\pi_0)_{\mathcal{A}}-\hat{\pi}_{\mathcal{A}})\|_2^2\mathds{1}(\hat{\eta}>\eta/2))+4 D N T \sigma_{\xi}^2+ o(NT)
\end{align*}

We now bound $\|(\pi_0)_{\mathcal{A}}-\hat{\pi}_{\mathcal{A}})\|_2^2$: repeating a similar argument as in lemma \ref{lm:EN rate}, we have
\[\|\hat{\pi}_{\mathcal{A}}-\hat{\pi}_{\mathcal{A}}(\lambda_2,0)\|_2\leq \frac{\lambda_1^* \sqrt{\sum_{j\in \mathcal{A} }\hat{w}_j}}{\lambda_{min}(\tilde{\mathbf{B}}_{\mathcal{A}}'\tilde{\mathbf{B}}_{\mathcal{A}})+\lambda_2}\leq \frac{\lambda_1^* \hat{\eta}^{-\gamma} \sqrt{T}}{d N+\lambda_2}\]
where the last step uses the definition of $\hat{\eta}$ and the definition of $\hat{w}_j$. In addition, $dN \leq \lambda_{min}(\tilde{\mathbf{X}}'\tilde{\mathbf{B}})\leq \lambda_{min}(\tilde{\mathbf{B}}_{\mathcal{A}}'\tilde{\mathbf{B}}_{\mathcal{A}})$ and $\lambda_{max}(\tilde{\mathbf{B}}_{\mathcal{A}}'\tilde{\mathbf{B}}_{\mathcal{A}})\leq \lambda_{max}(\tilde{\mathbf{B}}'\tilde{\mathbf{B}})\leq D N$

Continuing the argument of lemma \ref{lm:EN rate}, and note that $\|\tilde{\mathbf{B}}_{\mathcal{A}}'\xi\|_2^2\leq \|\tilde{\mathbf{B}}'\xi\|_2^2$ and $\|\tilde{\mathbf{B}}_{\mathcal{A}}(\tilde{\mathbf{y}}^{EN}-\mathbf{y}^{EN})\|_2^2\leq \|\tilde{\mathbf{B}}(\tilde{\mathbf{y}}^{EN}-\mathbf{y}^{EN})\|_2^2$
we have

\[\|\hat{\pi}_{\mathcal{A}}(\lambda_2,0)-(\pi_0)_{\mathcal{A}}\|_2^2\leq 4 \frac{\lambda_2^2 \|\pi_0\|_2^2 + \lambda_{max}(\tilde{\mathbf{B}}'\tilde{\mathbf{B}}) |\mathcal{A}| \sigma_{\xi}^2 +o(NT)}{(\lambda_{min}(\tilde{\mathbf{B}}'\tilde{\mathbf{B}})+\lambda_2)^2}\]

Therefore,
\begin{align*}
& \mathbb{E}(\|(\pi_0)_{\mathcal{A}}-\hat{\pi}_{\mathcal{A}})\|_2^2\mathds{1}(\hat{\eta}>\eta/2))\\
&\leq  8 \frac{ \lambda_1^{*2} (\eta/2)^{-2\gamma} T +\lambda_2^2 \|\pi_0\|_2^2 + \lambda_{max}(\tilde{\mathbf{B}}'\tilde{\mathbf{B}}) |\mathcal{A}| \sigma_{\xi}^2 +o(NT)}{(\lambda_{min}(\tilde{\mathbf{B}}'\tilde{\mathbf{B}})+\lambda_2)^2}\\
&\leq 8 \frac{ \lambda_1^{*2} (\eta/2)^{-2\gamma} T +\lambda_2^2 \|\pi_0\|_2^2 + \lambda_{max}(\tilde{\mathbf{B}}'\tilde{\mathbf{B}}) T \sigma_{\xi}^2 +o(NT)}{(d N+\lambda_2)^2}
\end{align*}

Combining all the derivations above, we have
\begin{align*}
& \mathbb{P}(\exists j\in \mathcal{A}^C| |-2 \tilde{\mathbf{B}}_j'(\tilde{\mathbf{y}}^{EN}-\tilde{\mathbf{B}} \hat{\pi}_{\mathcal{A}})| >\lambda_1 \hat{w}_j )\\
&\leq \underbrace{4\frac{2 \lambda_1^2 T +8 \lambda_2^2 \|\pi_0\|_2^2 + 8 D N T \sigma_{\xi}^2 +o(NT)}{(N d+\lambda_2)^{2} \eta^2 }}_{I}\\
&+ \underbrace{\frac{2 \lambda_1^2 T +8 \lambda_2^2 \|\pi_0\|_2^2 + 8 B N T \sigma_{\xi}^2 +o(NT)}{(N d+\lambda_2)^{2} M^2}}_{II}\\
&+ \underbrace{\frac{4 M^{2\gamma}}{\lambda_1^{*2}} \{32 D^2 N^2  \frac{ \lambda_1^{*2} (\eta/2)^{-2\gamma} T +\lambda_2^2 \|\pi_0\|_2^2 + \lambda_{max}(\tilde{\mathbf{B}}'\tilde{\mathbf{B}}) T \sigma_{\xi}^2 +o(NT)}{(d N+\lambda_2)^2} +4 D N T \sigma_{\xi}^2+ o(NT)\}}_{III}
\end{align*}

\[I = O(\frac{T}{N \eta^2})=O((\lambda_1^* \sqrt{T/N}\eta^{-\gamma})^{2/\gamma} (\frac{T}{N} (\frac{N}{\lambda_1^*})^{2/(1+\gamma)})^{\frac{1+\gamma}{\gamma}} T^{-2/\gamma}) =o(1)\]

\[II=O(\frac{T}{N} (\frac{N}{\lambda_1^*})^{2/(1+\gamma)})=o(1)\]
because $\lim_{N\rightarrow \infty} \lambda_1^{*} N^{-1-(g-1)(1+\gamma)/2}=\infty$ implies $\lim_{N\rightarrow \infty}(\frac{N}{\lambda_1^*})^{2/(1+\gamma)} N^{g-1}=0$

\[III = O(\lambda_1^{*\frac{-2}{1+\gamma}} N^{\frac{2}{1+\gamma}+g-1})+O((\frac{\lambda_1^*}{N})^{\frac{2\gamma}{1+\gamma}} \eta^{-2\gamma} T)=o(1)\]

\end{proof}

We are now ready to prove the result that $\|\hat{\pi}(AEN)-\pi_0\|_2^2 =o_p(N^{g-1})$ which establishes the claim that $\|\hat{\pi}(AEN)-\pi_0\|_2^2 =o_p(N^{-\zeta})$ for $\zeta=1-g$. 

\begin{lemma}
Under assumption \ref{as:AEN}, \ref{as:rr} and \ref{as:G}, we have 
\[\|\hat{\pi}(AEN)-\pi_0\|_2^2 =o_p(N^{g-1})\]
\end{lemma}

\begin{proof}
It suffices to show that conditional on the event $\hat{\pi}(AEN)=(1+\lambda_2/N) (\hat{\pi}_{\mathcal{A}},0)$ for $\hat{\pi}_{\mathcal{A}}$ defined in lemma \ref{lm:EN oracle}, we have $\|\hat{\pi}(AEN)-\pi_0\|_2^2 =o_p(N^{g-1})$. This is because we have
\begin{align*}
& \mathbb{P}(\|\hat{\pi}(AEN)-\pi_0\|_2^2 >C)\\
&\leq \mathbb{P}(\|\hat{\pi}(AEN)-\pi_0\|_2^2 >C, \hat{\pi}(AEN)=(1+\lambda_2/N) (\hat{\pi}_{\mathcal{A}},0))+\mathbb{P}(\hat{\pi}(AEN)\neq (1+\lambda_2/N) (\hat{\pi}_{\mathcal{A}},0))
\end{align*}

By lemma \ref{lm:EN oracle}, the second probability converges to zero, so for any $\delta>0$, we can choose $N$ sufficiently large such that the second probability is smaller than $\delta/2$ and then choose $C$ large enough so that the first probability is also smaller than $\delta/2$. Hence, I will now assume 
\[\hat{\pi}(AEN)=(1+\lambda_2/N) (\hat{\pi}_{\mathcal{A}},0)\]

Recall that $1+\lambda_2/N\rightarrow 1$, so equivalently, it suffices to establish that 
\[\|\hat{\pi}_{\mathcal{A}} -\frac{N}{N+\lambda_2}(\pi_0)_{\mathcal{A}}\|_2^2=o_p(N^{g-1})\]

Note the following inequality:
\begin{align*}
& \|\hat{\pi}_{\mathcal{A}} -\frac{N}{N+\lambda_2}(\pi_0)_{\mathcal{A}}\|_2^2\\
&\leq \|\frac{\lambda_2}{N+\lambda_2}(\pi_0)_{\mathcal{A}}+ (\hat{\pi}_{\mathcal{A}}-\hat{\pi}_{\mathcal{A}}(\lambda_2,0)) + (\hat{\pi}_{\mathcal{A}}(\lambda_2,0)-(\pi_0)_{\mathcal{A}})\|_2^2\\
&\leq 3 \underbrace{\|\frac{\lambda_2}{N+\lambda_2}(\pi_0)_{\mathcal{A}} \|_2^2}_{I} + 3 \underbrace{\|\hat{\pi}_{\mathcal{A}}-\hat{\pi}_{\mathcal{A}}(\lambda_2,0)\|_2^2}_{II}\\
&+3 \underbrace{\|\hat{\pi}_{\mathcal{A}}(\lambda_2,0)-(\pi_0)_{\mathcal{A}}\|_2^2}_{III}
\end{align*}

It suffices to bound the three terms in the brackets.

\textbf{Bounding term (I)}

\begin{align*}
& N^{1-g }\|\frac{\lambda_2}{N+\lambda_2}(\pi_0)_{\mathcal{A}} \|_2^2\\
&\leq N^{1-g} \frac{\lambda_2^2}{(N+\lambda_2)^2}\|\pi_0 \|_2^2\\
&\leq N^{1-g} o(N^{-1})\\
&=o(1)
\end{align*}
where the third line uses $\lim_{N\rightarrow \infty} \frac{\lambda_2}{\sqrt{N}}\|\pi_0\|_2=0$.

\textbf{Bounding term (II)}

Recall that in lemma \ref{lm:EN oracle}, we have established that 
\[\|\hat{\pi}_{\mathcal{A}}-\hat{\pi}_{\mathcal{A}}(\lambda_2,0)\|_2\leq \frac{\lambda_1^{*} \hat{\eta}^{-\gamma} \sqrt{T}}{d N+\lambda_2}\]

We will try to replace $\hat{\eta}$ by $\eta$ in the above expression. This is done by noticing that by the definition of $\eta$ and $\hat{\eta}$,

\begin{align*}
& (\hat{\eta}-\eta)^2\\
&\leq \|\hat{\pi}(\lambda_2,\lambda_1)-\pi_0\|_2^2\\
&\leq \frac{2 \lambda_1^2 T +8 \lambda_2^2 \|\pi_0\|_2^2 + 8 D N T \sigma_{\xi}^2 +o(NT)}{(N d+\lambda_2)^{2}}\\
&=o_p(1)
\end{align*}

using lemma \ref{lm:EN rate}. Therefore, for $\eta$ bounded away from zero, $\frac{\hat{\eta}}{\eta}\xrightarrow{p} 1$. 

Consequently,
\begin{align*}
& N^{1-g} \|\hat{\pi}_{\mathcal{A}}-\hat{\pi}_{\mathcal{A}}(\lambda_2,0)\|_2^2\\
&\leq N^{1-g} \frac{\lambda_1^{*2} \hat{\eta}^{-2\gamma} T}{(b N+\lambda_2)^2}\\
&= N^{1-g} \frac{\lambda_1^{*2} \eta^{-2\gamma} T}{(b N+\lambda_2)^2} (\frac{\hat{\eta}}{\eta})^{-2\gamma}\\
&= N^{1-g} o(1/N) O_p(1)\\
&= o_p(N^{-g})=o_p(1)
\end{align*}
where I used that $(\frac{\sqrt{N}}{\sqrt{T} \lambda_1^*})^{1/\gamma}   \min_{j:|(\pi_0)_j|>0} |(\pi_0)_j| \rightarrow \infty$

\textbf{Boundng term (III)}

\begin{align*}
& \|\hat{\pi}_{\mathcal{A}}(\lambda_2,0)-(\pi_0)_{\mathcal{A}}\|_2^2\\
&\leq 4 \|\lambda_2 (\tilde{\mathbf{B}}_{\mathcal{A}}'\tilde{\mathbf{B}}_{\mathcal{A}}+\lambda_2 I)^{-1} (\pi_0)_{\mathcal{A}}\|_2^2\\
&+ 4 \|(\tilde{\mathbf{B}}_{\mathcal{A}}'\tilde{\mathbf{B}}_{\mathcal{A}}+\lambda_2 I)^{-1} \tilde{\mathbf{B}}_{\mathcal{A}}' \xi\|_2^2\\
&+4 \|(\tilde{\mathbf{B}}_{\mathcal{A}}'\tilde{\mathbf{B}}_{\mathcal{A}}+\lambda_2 I)^{-1} \tilde{\mathbf{B}}_{\mathcal{A}}' (\tilde{\mathbf{y}}^{EN}-\mathbf{y}^{EN})\|_2^2\\
&+4 \|(\tilde{\mathbf{B}}_{\mathcal{A}}'\tilde{\mathbf{B}}_{\mathcal{A}}+\lambda_2 I)^{-1} \tilde{\mathbf{B}}_{\mathcal{A}}' (\mathbf{B}_{\mathcal{A}}-\tilde{\mathbf{B}}_{\mathcal{A}}) (\pi_0)_{\mathcal{A}} \|_2^2 \\
&\leq 4 \frac{\lambda_2^2 \|\pi_0\|_2^2}{(d N+ \lambda_2)^2}+4 \frac{Tr(\tilde{\mathbf{B}}_{\mathcal{A}}\tilde{\mathbf{B}}_{\mathcal{A}}' ) \sigma_{\xi}^2}{(d N+ \lambda_2)^2}\\
&+ 4 \frac{(\mathbf{y}^{EN}-\tilde{\mathbf{y}}^{EN})'\tilde{\mathbf{B}} \tilde{\mathbf{B}}'(\mathbf{y}^{EN}-\tilde{\mathbf{y}}^{EN})}{(d N+ \lambda_2)^2}\\
&+4 \frac{(\pi_0)_{\mathcal{A}}'(\mathbf{B}_{\mathcal{A}}-\tilde{\mathbf{B}}_{\mathcal{A}}) \tilde{\mathbf{B}}_{\mathcal{A}} \tilde{\mathbf{B}}_{\mathcal{A}}' (\mathbf{B}_{\mathcal{A}}-\tilde{\mathbf{B}}_{\mathcal{A}}) (\pi_0)_{\mathcal{A}}}{(d N+ \lambda_2)^2}\\
&\leq 4 o(N^{-1}) +4 \frac{ |\mathcal{A}| N B }{(d N+ \lambda_2)^2} + 4 o(\frac{NT}{N^2})\\
&\leq 3 o(N^{-1+g}) +4 O(\frac{|\mathcal{A}|}{N}) + 4 o(\frac{T}{N})
\end{align*}
All three terms will be $o(N^{-1+g})$ if $|\mathcal{A}|/T \rightarrow 0$.

\end{proof}

\section{Supplemental information for empirical application}

\subsection{Missing values}

Although we restricted the time period to 1997-2012, during which observations have much fewer missing values, we still encounter missing values for some counties in the data. Most of these are counties whose administrative divisions have changed during the sample period, or their names have changed and hence dropped out of the dataset. We drop these counties, as they are not eligible to be selected as experimental sites. In addition, we do not observe grain production level for all counties in Sichuan Province for the year 2012. We used 2-nearest neighbor average to impute the grain production per rural employment for year 2012. 
For the rest of missing values that are scattered across the whole dataset, we simply replace it with the province-by-time average. 

\subsection{Regularity conditions in section \ref{as theory}}

In this section, we verify that higher-order derivatives of the moment function satisfy the uniform bound assumed in section \ref{as theory} for the logit model in our empirical application. Before proceeding to the details, we first establish a useful lemma for the results.

\begin{lemma}\label{lm:exp}
Let $x,y\in \mathbb{R}$, consider the function $g(x,y) = \frac{x^s \exp(t x+y)}{(1+\exp(t x+y))^2}$ for $s\in \mathbb{Z}^+$, $s\geq 1$ and $t\neq 0$. It is uniformly bounded over $\{(x,y): (x,y)\in \mathbb{R}\times K\}$ for some compact set $K$. 

\end{lemma}

\begin{proof}

Consider the set $\mathbb{S}=\{(x,y): |x|\leq M, y\in K\}$. The set $\mathbb{S}$ is compact, and $g(x,y)$ is continuous over $\mathbb{S}$, so $g(x,y)$ is uniformly bounded over $\mathbb{S}$. 

Suppose $t>0$, 
For fixed $y\in K$, consider the function defined by $f_y(x)=g(x,y)$. Applying L'Hopital's law and observing $0<\frac{\exp(t x+y)}{(1+\exp(t x+y))}<1$ and $0<\frac{1}{(1+\exp(t x+y))}<1$, we have

\begin{align*}
& \lim_{x\rightarrow \infty} f_y(x) \\
&= \lim_{x\rightarrow \infty} \frac{x^s \exp(t x+y)}{(1+\exp(t x+y))^2} \\
&\leq \lim_{x\rightarrow \infty} \frac{x^s}{1+\exp(t x+y)}\\
&= 0. 
\end{align*}
As $f_y(x)>0$ for $x>0$, the limit $\lim_{x\rightarrow \infty} f_y(x)$ equals 0. 
In addition, the compactness of $K$ implies that the above can be made uniform in $y\in K$. For instance, there exists a lower bound $-M_y$ such that $y\geq -M_y$ for all $y\in K$. Then simply choose $x_0$ such that for all $x>x_0$, $\frac{x^s}{1+\exp(t x-M_y)}<\epsilon$, then by the monotonicity of $\frac{x^s}{1+\exp(t x+y)}$ in $y$ for fixed $x$, for all $x>x_0$, $f(x,y)<\epsilon$.

Next consider the limit as $x\rightarrow -\infty$. When $s$ is odd, note that
\begin{align*}
& \lim_{x\rightarrow -\infty} f_y(x) \\
&= \lim_{x\rightarrow -\infty} \frac{x^s \exp(t x+y)}{(1+\exp(t x+y))^2} \\
&\geq  \lim_{x\rightarrow -\infty} x^s \exp(t x+y)\\
& =0.
\end{align*} 
As $f_y(x)<0$ for $x<0$ and $s$ odd, the limit $\lim_{x\rightarrow -\infty} f_y(x)$ also equals 0. 

When $s$ is even, we can upper bound $f_y(x)$ by $x^s \exp(t x+y)$ which also has limit zero. 
Using the same reasoning, the convergence can be made uniform for all $y\in K$ using the monotonicity of $x^s \exp(t x+y)$ in $y$ for fixed $x$. 

The case for $t<0$ is completely symmetric, where we upper bound $f_y(x)$ by $x^s \exp(t x+y)$ for $x\rightarrow \infty$  and by $\frac{x^s}{1+\exp(t x+y)}$ when $x\rightarrow -\infty$ and $s$ even and lower bound $f_y(x)$ by $\frac{x^s}{1+\exp(t x+y)}$ for $x\rightarrow -\infty$ and $s$ odd.

Therefore, by taking $M$ that defines $\mathbb{S}$ large enough, we will have $g(x,y)$ being uniformly bounded over $\mathbb{R}\times K$.

\end{proof}

Recall the moment condition for logit model (\ref{eq:mom logit}):

\begin{align*}
\mathbb{E} \begin{bmatrix}  (\text{Experiment}_i- \Lambda(W_i'\mu_{0,-}+\alpha_i \mu_{0,1}) ) \alpha_i \\
(\text{Experiment}_i- \Lambda(W_i'\mu_{0,-}+\alpha_i \mu_{0,1}) ) W_i \end{bmatrix} =\begin{bmatrix}0  \\ \mathbf{0}  \end{bmatrix},
\end{align*}

Let $\Delta_{W_i,\alpha_i}=W_i'\mu_{0,-}+\alpha_i \mu_{0,1}$. Recall the short hand $\Lambda_i= \Lambda(\Delta_{W_i,\alpha_i})$ for the function $\Lambda(x)=\frac{\exp(x)}{1+\exp(x)}$, we get 
\begin{align*}
\frac{\partial m(W_i,\text{Experiment}_i,\alpha_i,\mu_0)}{\partial \alpha_i}=\begin{bmatrix} 
-\mu_{01} \Lambda_i (1-\Lambda_i) \alpha_i + \text{Experiment}_i- \Lambda_i \\
-\mu_{01} \Lambda_i (1-\Lambda_i) W_i \end{bmatrix}
\end{align*}

\begin{align}
\frac{\partial m(W_i,\text{Experiment}_i,\alpha_i,\mu_0)}{\partial \mu} =- \Lambda_i (1-\Lambda_i)  \begin{bmatrix}\alpha_i \alpha_i & W_i' \alpha_i  \\ \alpha_i W_i & W_i W_i'  \end{bmatrix}
\label{eq:ap mu}
\end{align}

\begin{align}
\frac{\partial^2 m(W_i,\text{Experiment}_i,\alpha_i,\mu_0)}{\partial \alpha_i^2} &= \begin{bmatrix}
- 2\mu_{01} \Lambda_i (1-\Lambda_i)- \alpha_i \mu_{01}^2 \Lambda_i (1-\Lambda_i) \frac{1-\exp(\Delta_{W_i,\alpha_i})}{1+\exp(\Delta_{W_i,\alpha_i})}\\
 -\mu_{01}^2 \Lambda_i (1-\Lambda_i) \frac{1-\exp(\Delta_{W_i,\alpha_i})}{1+\exp(\Delta_{w,\alpha})} W_i 
\end{bmatrix}
\label{eq:ap alpha2}
\end{align}

\begin{align}
\frac{\partial^2 m(W_i,\text{Experiment}_i,\alpha_i,\mu_0)}{\partial \alpha_i \partial \mu'}&= -\mu_{01} \Lambda_i (1-\Lambda_i) \frac{1-\exp(\Delta_{W_i,\alpha_i})}{1+\exp(\Delta_{W_i,\alpha_i})} \begin{bmatrix}\alpha_i \alpha_i & W_i' \alpha_i  \\ \alpha_i W_i & W_i W_i'  \end{bmatrix} \nonumber\\
&-\Lambda_i (1-\Lambda_i) \begin{bmatrix}  2 \alpha_i & W_i' \\ W_i & 0 \end{bmatrix} 
\label{eq:ap alphamu}
\end{align}

\begin{align}
& \frac{\partial^3 m(W_i,\text{Experiment}_i,\alpha_i,\mu_0)}{\partial \alpha_i^2 \partial \mu'} \nonumber\\
&=-\mu_{01}^2 \frac{\exp(\Delta_{W_i,\alpha_i})[\exp(2\Delta_{W_i,\alpha_i})-4\exp(\Delta_{W_i,\alpha_i})+1]}{[1+\exp(\Delta_{W_i,\alpha_i})]^4} \begin{bmatrix}\alpha_i \alpha_i & W_i' \alpha_i  \\ \alpha_i W_i & W_i W_i' \end{bmatrix}\nonumber \\
& - 2 \mu_{01}  \frac{\exp(\Delta_{W_i,\alpha_i})[1-\exp(\Delta_{W_i,\alpha_i})]}{[1+\exp(\Delta_{W_i,\alpha_i})]^3}\begin{bmatrix}  \alpha_i & W_i' \\ 0 & \mathbf{0}   \end{bmatrix} \nonumber \\
&-\begin{bmatrix} 2\Lambda_i (1-\Lambda_i) + 2\alpha_i \mu_{01} \Lambda_i (1-\Lambda_i) \frac{1-\exp(\Delta_{W_i,\alpha_i})}{1+\exp(\Delta_{W_i,\alpha_i})}    & \mathbf{0} \\ 2 \mu_{01}\Lambda_i (1-\Lambda_i) \frac{1-\exp(\Delta_{W_i,\alpha_i})}{1+\exp(\Delta_{W_i,\alpha_i})} W_i    & \mathbf{0} \end{bmatrix}
\label{eq:ap alpha2mu}
\end{align}

With the expressions (\ref{eq:ap mu})-(\ref{eq:ap alpha2mu}), we can now verify the assumptions imposed on the moment function in section \ref{as theory} for the logit model. 

\begin{proposition}\label{prop:EA}
Suppose that the random vector $W_i$ has support over a compact set $K_W$, then the following results hold:
\begin{enumerate}
	\item $\sup_{\alpha_i,W_i} \frac{\partial^2 m(W_i,\text{Experiment}_i,\alpha_i,\mu_0) }{\partial \alpha_i^2}=C_m <\infty$
	\item $\sup_{\alpha_i,W_i} \sup_{\mu\in \mathcal{N}(\mu_0)} \frac{\partial^3 m(W_i,\text{Experiment}_i,\alpha_i,\mu_0)}{\partial \alpha_i^2 \partial \mu}<C_m'<\infty$
	\item $\sup_{\alpha} \|\frac{\partial^2 m(W_i,\text{Experiment}_i,\alpha_i,\mu)}{\partial \alpha_i \partial \mu}\|^2\leq F(W_i)$ for all $\mu$ within a neighborhood $\mathcal{N}(\mu_0)$ of $\mu_0$ and $\mathbb{E}(F(W_i))<\infty$. 
	\item For $\mu$ lying in a compact parameter space, 
	\[  \|\frac{\partial m(W_i,\alpha_i,\mu)}{\partial \mu}-\frac{\partial m(W_i,\alpha_i,\mu_0)}{\partial \mu}\|\leq F_{\alpha}(W_i,\alpha_i) \|\mu-\mu_0\|^{1/C}  \quad  \mathbb{E}(F_{\alpha}(W_i,\alpha_i))<\infty\]
	if the following moment conditions hold:
	\begin{enumerate}
		\item $\mathbb{E}(\|W_i\|^4)<\infty$
		\item $\mathbb{E}(\|W_i\|^3 |\alpha_i|)<\infty$
		\item $\mathbb{E}(\|W_i\|^2 \alpha_i^2)<\infty$
		\item $\mathbb{E}(\|W_i\| |\alpha_i|^3)<\infty$
		\item $\mathbb{E}(\alpha_i^4)<\infty$
	\end{enumerate}
\end{enumerate}

\end{proposition}

\begin{proof}

\textbf{Proof of (i)}: Recall that

\begin{align*}
\frac{\partial^2 m(W_i,\text{Experiment}_i,\alpha_i,\mu_0)}{\partial \alpha_i^2} &= \begin{bmatrix}
- 2\mu_{01} \Lambda_i (1-\Lambda_i)- \alpha_i \mu_{01}^2 \Lambda_i (1-\Lambda_i) \frac{1-\exp(\Delta_{W_i,\alpha_i})}{1+\exp(\Delta_{W_i,\alpha_i})}\\
 -\mu_{01}^2 \Lambda_i (1-\Lambda_i) \frac{1-\exp(\Delta_{W_i,\alpha_i})}{1+\exp(\Delta_{w,\alpha})} W_i 
\end{bmatrix}
\end{align*}
We will verify that each term is uniformly bounded in $\alpha_i$ and $W_i$. First, observe that $\Lambda_i$ and $(1-\Lambda_i)$ lies between 0 and 1, so $|-2 \mu_{01} \Lambda_i (1-\Lambda_i)|\leq 2\mu_{01}$ for all $\alpha_i$ and $W_i$. 

Next we consider $\mu_{01}^2 \Lambda_i (1-\Lambda_i) \frac{1-\exp(\Delta_{W_i,\alpha_i})}{1+\exp(\Delta_{w,\alpha})} W_i $. As $W_i$ has a compact support, it suffices to verify that the term $\Lambda_i (1-\Lambda_i) \frac{1-\exp(\Delta_{W_i,\alpha_i})}{1+\exp(\Delta_{w,\alpha})}$ is uniformly bounded. In particular,
\begin{align*}
&\Lambda_i (1-\Lambda_i) \frac{1-\exp(\Delta_{W_i,\alpha_i})}{1+\exp(\Delta_{w,\alpha})}\\
& = \frac{\exp(\Delta_{W_i,\alpha_i})}{(1+\exp(\Delta_{W_i,\alpha_i}))^3}-\frac{\exp(\Delta_{W_i,\alpha_i})^2}{(1+\exp(\Delta_{W_i,\alpha_i}))^3},
\end{align*}
where the two summands are products of $\Lambda_i$ and $1-\Lambda_i$, so both of them are bounded by 1. Therefore, $\mu_{01}^2 \Lambda_i (1-\Lambda_i) \frac{1-\exp(\Delta_{W_i,\alpha_i})}{1+\exp(\Delta_{w,\alpha})} W_i $ is also uniformly bounded.

Finally, we verify $\alpha_i \mu_{01}^2 \Lambda_i (1-\Lambda_i) \frac{1-\exp(\Delta_{W_i,\alpha_i})}{1+\exp(\Delta_{W_i,\alpha_i})}$ is uniformly bounded. In addition, recall that $\Delta_{W_i,\alpha_i}=W_i'\mu_{0,-}+\alpha_i \mu_{0,1}$. We apply lemma \ref{lm:exp} with $s=1$, $x=\alpha_i$, $t=\mu_{01}$ and $y=W_i'\mu_{0,-}$ which lies in a compact set. Therefore, $\alpha_i \mu_{01}^2 \frac{\exp(\Delta_{W_i,\alpha_i})}{(1+\exp(\Delta_{W_i,\alpha_i}))^2}$ is also uniformly bounded.
Now we further note that $\Lambda_i=\frac{\exp(\Delta_{W_i,\alpha_i})}{1+\exp(\Delta_{W_i,\alpha_i})}$ and $1-\Lambda_i=\frac{1}{1+\exp(\Delta_{W_i,\alpha_i})}$ lie between zero and one. Therefore, any uniformly bounded function multiplied by them will still be unformly bounded. Therefore, we have
\begin{align*}
& \alpha_i \mu_{01}^2 \Lambda_i (1-\Lambda_i) \frac{1-\exp(\Delta_{W_i,\alpha_i})}{1+\exp(\Delta_{W_i,\alpha_i})}\\
&=\alpha_i \mu_{01}^2 \frac{\exp(\Delta_{W_i,\alpha_i})}{(1+\exp(\Delta_{W_i,\alpha_i}))^3}- \alpha_i \mu_{01}^2 \frac{\exp(\Delta_{W_i,\alpha_i})^2}{(1+\exp(\Delta_{W_i,\alpha_i}))^3}
\end{align*}
being the difference of two uniformly bounded functions, so it is also uniformly bounded.

\bigskip

\textbf{Proof of (ii)}: Recall that 

\begin{align*}
& \frac{\partial^3 m(W_i,\text{Experiment}_i,\alpha_i,\mu_0)}{\partial \alpha_i^2 \partial \mu}\\
&=-\mu_{01}^2 \frac{\exp(\Delta_{W_i,\alpha_i})[\exp(2\Delta_{W_i,\alpha_i})-4\exp(\Delta_{W_i,\alpha_i})+1]}{[1+\exp(\Delta_{W_i,\alpha_i})]^4} \begin{bmatrix}\alpha_i \alpha_i & W_i' \alpha_i  \\ \alpha_i W_i & W_i W_i' \end{bmatrix}\\
& - 2 \mu_{01}  \frac{\exp(\Delta_{W_i,\alpha_i})[1-\exp(\Delta_{W_i,\alpha_i})]}{[1+\exp(\Delta_{W_i,\alpha_i})]^3}\begin{bmatrix}  \alpha_i & W_i \\ 0 & \mathbf{0}   \end{bmatrix} \\
&-\begin{bmatrix} 2\Lambda_i (1-\Lambda_i) + 2\alpha_i \mu_{01} \Lambda_i (1-\Lambda_i) \frac{1-\exp(\Delta_{W_i,\alpha_i})}{1+\exp(\Delta_{W_i,\alpha_i})}    & \mathbf{0} \\ 2 \mu_{01}\Lambda_i (1-\Lambda_i) \frac{1-\exp(\Delta_{W_i,\alpha_i})}{1+\exp(\Delta_{W_i,\alpha_i})} W_i    & \mathbf{0} \end{bmatrix}
\end{align*}

Note that each function is continuously differentiable in $\mu$ for $\mu$ in a neighborhood of $\mu_0$, so without loss of generality, we can take the neighborhood to be compact. Therefore, it suffices to prove uniformity in $\alpha_i$ given $W_i$ has compact support.

A similar reasoning as in the proof of (i) shows that the two matrices in the second last line and the last line of the above display are uniformly bounded. It suffices to work with the first matrices. We apply lemma \ref{lm:exp} again, with $s\in \{1,2\}$, $x=\alpha_i$, $t=\mu_{01}$ and $y=W_i'\mu_{0,-}$ again to establish that $\mu_{1}^2 \alpha_i^2 \frac{\exp(\Delta_{W_i,\alpha_i})}{(1+\exp(\Delta_{W_i,\alpha_i}))^2}$ is uniformly bounded. Therefore, for the upper-left corner of the matrix, noticing that 
\begin{align*}
& \mu_{01}^2 \alpha_i^2 \frac{\exp(\Delta_{W_i,\alpha_i})[\exp(2\Delta_{W_i,\alpha_i})-4\exp(\Delta_{W_i,\alpha_i})+1]}{[1+\exp(\Delta_{W_i,\alpha_i})]^4}\\
&=\mu_{01}^2 \alpha_i^2 \frac{\exp(\Delta_{W_i,\alpha_i})^3}{[1+\exp(\Delta_{W_i,\alpha_i})]^4}\\
&-4 \mu_{01}^2 \alpha_i^2 \frac{\exp(\Delta_{W_i,\alpha_i})^2}{[1+\exp(\Delta_{W_i,\alpha_i})]^4}\\
&+\mu_{01}^2 \alpha_i^2 \frac{\exp(\Delta_{W_i,\alpha_i})}{[1+\exp(\Delta_{W_i,\alpha_i})]^4},
\end{align*}
we can write each term as a product of $\mu_{1}^2 \alpha_i^2 \frac{\exp(\Delta_{W_i,\alpha_i})}{(1+\exp(\Delta_{W_i,\alpha_i}))^2}$ and $\Lambda_i$'s and $1-\Lambda_i$'s which are bounded between 0 and 1.

Therefore, $\mu_{1}^2 \alpha_i^2 \frac{\exp(\Delta_{W_i,\alpha_i})}{(1+\exp(\Delta_{W_i,\alpha_i}))^2}$ is uniformly bounded. The reasoning for the other elements of the matrix is exactly the same.

\bigskip

\textbf{Proof of (iii)}: We shall prove that the element-wise squared norm of $\frac{\partial^2 m(W_i,\text{Experiment}_i,\alpha_i,\mu_0)}{\partial \alpha_i \partial \mu}$ is uniformly bounded over a neighborhood of $\mu_0$, $\alpha_i$ and $W_i$. 

\begin{align*}
\frac{\partial^2 m(W_i,\text{Experiment}_i,\alpha_i,\mu_0)}{\partial \alpha_i \partial \mu}&= -\mu_{01} \Lambda_i (1-\Lambda_i) \frac{1-\exp(\Delta_{W_i,\alpha_i})}{1+\exp(\Delta_{W_i,\alpha_i})} \begin{bmatrix}\alpha_i \alpha_i & W_i' \alpha_i  \\ \alpha_i W_i & W_i W_i'  \end{bmatrix} \nonumber\\
&-\Lambda_i (1-\Lambda_i) \begin{bmatrix}  2 \alpha_i & W_i' \\ W_i & 0 \end{bmatrix} 
\end{align*}

It suffices to show that both summands has element-wise squared norm uniformly bounded by an integrable function of $W_i$ over $\mu\in \mathcal{N}(\mu_0)$ and $\alpha_i$. We start with the second matrix. By the boundedness of $\Lambda_i$ and $1-\Lambda_i$ and the compact support of $W_i$, it suffices to show that $\Lambda_i^2 (1-\Lambda_i)^2\alpha_i^2$ is uniformly bounded:
\begin{align*}
& \Lambda_i^2 (1-\Lambda_i)^2\alpha_i^2\\
&=\alpha_i^2 \frac{\exp(\Delta_{W_i,\alpha_i})^2}{(1+\exp(\Delta_{W_i,\alpha_i}))^4}
\end{align*}
The above display is a product of $\alpha_i^2 \frac{\exp(\Delta_{W_i,\alpha_i})}{(1+\exp(\Delta_{W_i,\alpha_i}))^2}$ and $\Lambda_i (1-\Lambda_i)$, so it is uniformly bounded by results in the proof of (ii).

We can also take the neighborhood of $\mu_0$ to be compact so that the uniformity over $\mu$ does not have any impact on the proof. Now we turn to the first matrix. We show that the square of $\mu_{01} \alpha_i^2 \Lambda_i (1-\Lambda_i) \frac{1-\exp(\Delta_{W_i,\alpha_i})}{1+\exp(\Delta_{W_i,\alpha_i})}$ is also uniformly bounded. The reasoning for the other entries are similar and omitted.

Note that 
\begin{align*}
& (\mu_{01} \alpha_i^2 \Lambda_i (1-\Lambda_i) \frac{1-\exp(\Delta_{W_i,\alpha_i})}{1+\exp(\Delta_{W_i,\alpha_i})})^2\\
&= \mu_{01}^2 \alpha_i^4 \frac{\exp(\Delta_{W_i,\alpha_i})^2 }{(1+\exp(\Delta_{W_i,\alpha_i}))^6}\\
&-\mu_{01}^2 \alpha_i^4 \frac{\exp(\Delta_{W_i,\alpha_i})^4 }{(1+\exp(\Delta_{W_i,\alpha_i}))^6}\\
\end{align*}

Here, we apply lemma \ref{lm:exp} again with $s=4$, $x=\alpha_i$, $t=\mu_{1}$ and $y= W_i'\mu_{-}$ to obtain that $\mu_{01}^2 \alpha_i^4 \frac{\exp(\Delta_{W_i,\alpha_i}) }{(1+\exp(\Delta_{W_i,\alpha_i}))^2}$. The two terms are products of the above with $\Lambda_i$'s and $(1-\Lambda_i)$'s which are still uniformly bounded.

\bigskip

\textbf{Proof of (iv):}

We will prove the result with $C=1$. 
It suffices to show that for each element of $\frac{\partial m(W_i,\alpha_i,\mu)}{\partial \mu}$ is satisfies the Lipschitz condition. \[\|\frac{\partial m(W_i,\alpha_i,\mu)}{\partial \mu}-\frac{\partial m(W_i,\alpha_i,\mu_0)}{\partial \mu}\|\leq F_{\alpha}(W_i,\alpha_i) \|\mu-\mu_0\|\] 

Recall that \begin{align*}
\frac{\partial m(W_i,\text{Experiment}_i,\alpha_i,\mu_0)}{\partial \mu} =- \Lambda_i (1-\Lambda_i)  \begin{bmatrix}\alpha_i \alpha_i & W_i' \alpha_i  \\ \alpha_i W_i & W_i W_i'  \end{bmatrix}
\end{align*}

So it suffices to show that the function
\[h(\mu;W_i,\alpha_i)=\frac{\exp(W_i'\mu_{-}+\alpha_i \mu_{1})}{[1+\exp(W_i'\mu_{-}+\alpha_i \mu_{1})]^2 }\] 
is Lipschitz in $\mu$. 

Taking derivative with respect to $\mu$, we have
\[\frac{\partial h(\mu;W_i,\alpha_i)}{\partial \mu'}= \{\frac{\exp(W_i'\mu_{-}+\alpha_i \mu_{1})}{[1+\exp(W_i'\mu_{-}+\alpha_i \mu_{1})]^3}-\frac{\exp(W_i'\mu_{-}+\alpha_i \mu_{1})^2}{[1+\exp(W_i'\mu_{-}+\alpha_i \mu_{1})]^3} \} \begin{bmatrix} \alpha_i &  W_i' \end{bmatrix}\]

Recall that $0<\frac{\exp(W_i'\mu_{-}+\alpha_i \mu_{1})}{[1+\exp(W_i'\mu_{-}+\alpha_i \mu_{1})]^3}<1$ and $0<\frac{\exp(W_i'\mu_{-}+\alpha_i \mu_{1})^2}{[1+\exp(W_i'\mu_{-}+\alpha_i \mu_{1})]^3}<1$. Therefore, the multiplier before the above matrix is bounded in absolute value by 2. Therefore, we can conclude that 

\[ |h(\mu;W_i,\alpha_i)-h(\mu_0;W_i,\alpha_i)|\leq 2\|(\alpha_i,W_i')\| \|\mu-\mu_0\|. \]

Therefore, we have
\begin{align*}
& \| \frac{\partial m(W_i,\alpha_i,\mu)}{\partial \mu}-\frac{\partial m(W_i,\alpha_i,\mu_0)}{\partial \mu}\|^2\\
&\leq 2  F_{\alpha}(W_i,\alpha_i)     \|\mu-\mu_0\|^2,
\end{align*}
where $F_{\alpha}(W_i,\alpha_i)$ is the Frobenius norm of the matrix 
\[\begin{bmatrix} \alpha_i^4+ \|W_i\| ^2 \alpha_i^2+2 |\alpha_i^3| \|W_i\| & \alpha_i^3 W_i' + |\alpha_i| \|W_i\|^2 W_i' +2 \alpha_i^2 \|W_i\|W_i'\\
|\alpha_i^3| W_i + |\alpha_i| \|W_i\|^2 W_i +2 \alpha_i^2 \|W_i\|W_i &  \alpha_i^2 W_i W_i'+\|W_i\|^2 W_i W_i' + 2 |\alpha_i| \|W_i\| W_i W_i'     \end{bmatrix}\]

The moment condition then ensures that the function $F_{\alpha}(W_i,\alpha_i)$ has finite expectation.

\end{proof}

\end{document}